\documentclass[11pt,a4paper]{article}%
\usepackage{jheppub}
\usepackage{amsmath}
\usepackage{graphicx}
\usepackage{amsfonts}
\usepackage{amssymb}%
\usepackage[latin1]{inputenc}%

\newtheorem{theorem}{Theorem}
\newtheorem{proposition}{Proposition}

\newtheorem{definition}{Definition}
\newtheorem{lemma}{Lemma}
\newtheorem{remark}{Remark}

\newtheorem{conjecture}{Conjecture}
\newenvironment{proof}[1][Proof]{\textbf{#1.} }{\ \rule{0.5em}{0.5em}}

\title{\boldmath Asymptotic factorization of n-particle $SU(N)$ form factors}

\author[a]{Hrachya M. Babujian}
\author[b]{Angela Foerster}
\author[c]{and Michael Karowski}

\affiliation[a]{Yerevan Physics Institute,\\
Alikhanian Brothers 2, Yerevan, 375036 Armenia and\\
International Institute of Physics,
Universidade Federal do Rio Grande do Norte (UFRN),\\
59078-400 Natal-RN, Brazil}
\affiliation[b]{Instituto de F\'{\i}sica da UFRGS,\\
Av. Bento Gon\c{c}alves 9500, Porto Alegre, RS - Brazil}
\affiliation[c]{Institut f\"{u}r Theoretische Physik, FU-Berlin,\\
Arnimallee 14, 14195 Berlin, Germany}

\emailAdd{babujian@yerphi.am}
\emailAdd{angela@if.ufrgs.br}
\emailAdd{karowski@physik.fu-berlin.de}

\abstract{
We investigate the high energy behavior of the $SU(N)$ chiral Gross-Neveu
model in $1+1$ dimensions. The model is integrable and matrix elements of
several local operators (form factors) are known exactly. The form factors show
rapidity space clustering, which means factorization, if a group of rapidities
is shifted to infinity. We analyze this phenomenon for the $SU(N)$ model.
For several operators the factorization formulas are presented explicitly.
\\[3mm]
}

\begin{document}

\maketitle

\section{Introduction}

The Bjorken scattering or inelastic lepton-hadron scattering at high energies
has been a very important and crucial stage in the development of modern QCD
\cite{Bjorken1,Bjorken2,Feynman}. This well known experimental investigation
in high energy physics is very actual and has now a modern continuation, being
part of lepton-hadron experimental research \cite{Ducloue,Wing}. The essential
point in these studies is the behavior of the structure functions of the
hadrons \cite{Feynman}. They describe the parton (quark) structure of the
hadrons and the nature of the interaction between the quarks inside of the
hadrons. The amplitude of the lepton-hadron interaction consists of two parts,
where the lepton part is well known. The hadron part, whose invariant
decomposition provides the hadron form factors or structure functions
\cite{Feynman}, is not known. In QCD the calculation of the structure function
for all values of the Bjorken variable $x$ is still an open problem.

On the other side, the existence of exact integrable models in 1+1 dimensional
asymptotically free theories may be relevant, providing valuable insights into
this discussion. Remarkably, due to integrability, it is possible to obtain
exact form factors of local operators \cite{KW,Sm,BFKZ,BFK1}\footnote{Other
approaches to form factors in integrable quantum field theories can be found
in \cite{CM2,FMS,YZ,Lu0,Lu,Lu1,BL,LuZ,Or}.}. In the remarkable papers
\cite{Balog1,Balog2,BaW} Balog and Weisz define analogs of the structure
functions in two-dimensional integrable quantum field theories. In particular,
they consider form factors of the current operator (related to the
structure-function) of the $O(3)$ sigma model, which are accurately computed
over the whole $x$ range; in addition, the structure functions and some
moments are compared with renormalized perturbation theory. They also
calculate structure functions in the $O(N)$ sigma model using 1/N expansion
and make some conjectures on possible universal formulae in 4 dimensional QCD
for small $x$. Interestingly, in \cite{BaW} the authors employ the so called
cluster behavior of the form factors to calculate the same structure
functions. Here we mention that in all of the previously cited papers the
authors use only 2,3 and 4 particle form factors in $O(3)$ or in $O(N)$ sigma models.

In this article we will start an investigation of the above mentioned problems
in an opposite order: we will analyze the cluster behavior of the $SU(N)$
chiral Gross-Neveu model\footnote{For $N=2$ also called $SU(2)$ Thirring
model.}, which is an asymptotically free theory. For this, we do not only use
the 2,3 and 4 particle form factors, but also the general n -particle form
factors. We should point out that the first investigation of the cluster
behavior of the exact form factors was performed by Smirnov \cite{Sm} in the
case of the sine-Gordon, the $SU(2)$ Thirring model and the $O(3)$ sigma
model. He also applied these results to the current algebra \cite{Sm}. For the
sinh-Gordon model the cluster property of form factors was investigated in
\cite{KoMu}. Here we will consider the high energy behavior of the exact form
factors in 1+1 dimensional asymptotically free quantum field theories
\cite{Sm,BFK}, with connection to the factorization property and the Bjorken scattering.

The paper is organized as follows: In Section \ref{s2} we recall some known
formulae, which will be used in the following. In particular we present the
$SU(N)$ S-matrix and construct the form factors which are $n$-particle matrix
elements of local operators. In Section \ref{s3} we investigate the
\textquotedblleft rapidity space clustering\textquotedblright\ of form
factors, which describes the behavior of form factors, if a group of
rapitities is shifted to infinity. Several examples of operators are
considered, as the Noether current, the energy-momentum tensor, the
fundamental field of the $SU(N)$ chiral Gross-Neveu model, etc. In Section
\ref{s4} we present the proofs. Some more technical details are delegated to
the Appendices.

\section{Generalities}

\label{s2}

\subsection{SU(N) S-matrix}

The two particle S-matrix is $S(\theta)=\mathbf{1}\,b(\theta)+\mathbf{P}%
c(\theta)$ or in terms of matrix elements \cite{BKKW,BFK0,BFK1,BFK3}%
\begin{equation}
S_{\alpha\beta}^{\delta\gamma}(\theta)=\delta_{\alpha}^{\gamma}\delta_{\beta
}^{\delta}\,b(\theta)+\delta_{\alpha}^{\delta}\delta_{\beta}^{\gamma
}\,c(\theta) \label{S}%
\end{equation}
where $\alpha,\beta,\gamma,\delta=1,\dots,N$ denote fundamental particles. We
introduce also%
\begin{equation}
\tilde{S}_{\alpha\beta}^{\delta\gamma}(\theta)=S_{\alpha\beta}^{\delta\gamma
}(\theta)/a(\theta)=\delta_{\alpha}^{\gamma}\delta_{\beta}^{\delta}\tilde
{b}(\theta)+\delta_{\alpha}^{\delta}\delta_{\beta}^{\gamma}\tilde{c}(\theta)
\label{St}%
\end{equation}
where%
\begin{align*}
a(\theta)  &  =b(\theta)+c(\theta)=\frac{\Gamma\left(  -\frac{\theta}{2\pi
i}\right)  \Gamma\left(  1-\frac{1}{N}+\frac{\theta}{2\pi i}\right)  }%
{\Gamma\left(  \frac{\theta}{2\pi i}\right)  \Gamma\left(  1-\frac{1}{N}%
-\frac{\theta}{2\pi i}\right)  }\\
\tilde{b}(\theta)  &  =\frac{b(\theta)}{a(\theta)}=\frac{\theta}{\theta-i\eta
},~\tilde{c}(\theta)=\frac{c(\theta)}{a(\theta)}=\frac{-i\eta}{\theta-i\eta
},~\eta=\frac{2\pi}{N}.
\end{align*}
\label{2}

\subsection{SU(N) form factors}

\paragraph{ Minimal form factor function $F(\theta)$, $\phi$- and $\tau
$-function:}

To construct the form factors we need the "minimal form factor function
$F\left(  \theta\right)  $" for two particles \cite{BFK1,BFK3}%
\begin{align}
F\left(  \theta\right)   &  =c\exp\int\limits_{0}^{\infty}\frac{dt}{t\sinh
^{2}t}e^{\frac{t}{N}}\sinh t\left(  1-1/N\right)  \left(  1-\cosh t\left(
1-\theta/(i\pi)\right)  \right) \label{F}\\
&  =\frac{G\left(  \frac{1}{2\pi i}\theta\right)  G\left(  1-\frac{1}{2\pi
i}\theta\right)  }{G\left(  1-\frac{1}{N}+\frac{1}{2\pi i}\theta\right)
G\left(  2-\frac{1}{N}-\frac{1}{2\pi i}\theta\right)  }~,~~c=F\left(
i\pi\right)  =\frac{G^{2}\left(  \frac{1}{2}\right)  }{G^{2}\left(  \frac
{3}{2}-\frac{1}{N}\right)  }\nonumber
\end{align}
where $G(z)$ is Barnes G-function. It is the minimal solution of the
equations
\[
F(\theta)=F(-\theta)a(\theta)\,,~~F(i\pi-\theta)=F(i\pi+\theta)
\]
where $a(\theta)$ is the highest weight amplitude of the corresponding channel
of the S-matrix (\ref{S}).

The $\phi$-function satisfies \cite{BFK1,BFK3}%
\[
\prod_{k=0}^{N-2}\tilde{\phi}\left(  -\theta-ki\eta\right)  \prod_{k=0}%
^{N-1}F\left(  \theta+ki\eta\right)  =1
\]
with the solution%
\begin{equation}
\tilde{\phi}(\theta)=\left(  F\left(  -\theta\right)  \bar{F}(i\pi
+\theta)\right)  ^{-1}=\Gamma\left(  -\frac{\theta}{2\pi i}\right)
\Gamma\left(  1-\frac{1}{N}+\frac{\theta}{2\pi i}\right)  \label{phi}%
\end{equation}
where
\begin{align}
\bar{F}\left(  \theta\right)   &  =\bar{c}\exp\int\limits_{0}^{\infty}%
\frac{dt}{t\sinh^{2}t}e^{\frac{t}{N}}\sinh t/N\left(  1-\cosh t\left(
1-\theta/(i\pi)\right)  \right) \label{Fbar}\\
&  =\frac{G\left(  \frac{1}{2}-\frac{1}{N}+\frac{1}{2}\frac{\theta}{i\pi
}\right)  G\left(  \frac{3}{2}-\frac{1}{N}-\frac{1}{2}\frac{\theta}{i\pi
}\right)  }{G\left(  \frac{1}{2}+\frac{1}{2}\frac{\theta}{i\pi}\right)
G\left(  \frac{3}{2}-\frac{1}{2}\frac{\theta}{i\pi}\right)  }~,~~\bar{c}%
=\bar{F}\left(  i\pi\right)  =G^{2}\left(  1-\frac{1}{N}\right) \nonumber
\end{align}
is the minimal F-function for a particle and an anti-particle satisfying
\begin{equation}
\bar{F}\left(  \theta\right)  =-\bar{F}\left(  -\theta\right)  b(i\pi-\theta).
\label{Fbb}%
\end{equation}
The $\tau$-function is%
\begin{equation}
\tau(z)=\left(  \tilde{\phi}(z)\tilde{\phi}(-z)\right)  ^{-1}=\frac{1}%
{2\pi^{2}}\frac{z\sinh\frac{1}{2}z}{\Gamma\left(  1-\frac{1}{N}+\frac{1}%
{2}\frac{z}{i\pi}\right)  \Gamma\left(  1-\frac{1}{N}-\frac{1}{2}\frac{z}%
{i\pi}\right)  }. \label{tau}%
\end{equation}

\paragraph{n particle form factors:}

The matrix element of a local operator $\mathcal{O}(x)$ for a state of $n$
particles of kind $\alpha_{i}$ with rapidities $\theta_{i}$
\begin{equation}
\langle\,0\,|\,\mathcal{O}(x)\,|\,\theta_{1},\dots,\theta_{n}\,\rangle
_{\underline{\alpha}}^{in}=e^{-ix(p_{1}+\cdots+p_{n})}F_{\underline{\alpha}%
}^{\mathcal{O}}(\underline{\theta})\, \label{F1}%
\end{equation}
defines the generalized form factor $F_{1\dots n}^{\mathcal{O}}({\underline
{\theta}})$, which is a co-vector valued function with components
$F_{\underline{\alpha}}^{\mathcal{O}}(\underline{\theta})\,$. The form factors
satisfy the \textbf{form factor equations (i) - (v)} (see Appendix \ref{sf}).
Solutions of these equations can be written as follows:

As usual we split off the minimal part \cite{KW}%
\begin{equation}
F_{\underline{\alpha}}^{\mathcal{O}}(\underline{\theta})=N_{n}F(\underline
{\theta})K_{\underline{\alpha}}(\underline{\theta}),~~F(\underline{\theta
})=\prod_{1\leq i<j\leq n}F(\theta_{ij})) \label{FK}%
\end{equation}
where $\underline{\alpha}=(\alpha_{1},\dots,\alpha_{n}),~\underline{\theta
}=\left(  \theta_{1},\dots,\theta_{n}\right)  $ and $F(\theta)$ is defined by
(\ref{F}). The K-function is given by \textbf{an `off-shell' Bethe ansatz} in
terms of the multiple contour integral%
\begin{equation}
\fbox{$\rule[-0.2in]{0in}{0.5in}\displaystyle~K_{\underline{\alpha}%
}^{\mathcal{O}}(\underline{\theta})=\int_{\mathcal{C}_{\underline{\theta}}%
}d\underline{z}\,\tilde{h}(\underline{\theta},\underline{z})\,p^{\mathcal{O}%
}(\underline{\theta},\underline{z})\,\tilde{\Psi}_{\underline{\alpha}%
}(\underline{\theta},\underline{z})$~} \label{K}%
\end{equation}
with $\underline{z}=\left(  z_{1},\dots,z_{m}\right)  $ and $\int
_{\mathcal{C}_{\underline{\theta}}}d\underline{z}=\frac{1}{m!}\int
_{\mathcal{C}_{\underline{\theta}}}dz_{1}\dots\int_{\mathcal{C}_{\underline
{\theta}}}dz_{m}$. The integration contour $\mathcal{C}_{\underline{\theta}}$
(see Fig. \ref{f5.1}) and the scalar function $h(\underline{\theta}%
,\underline{z})$ depend only on the S-matrix and not on the specific operator
$\mathcal{O}(x)$%
\begin{equation}
\tilde{h}(\underline{\theta},\underline{z})=\prod_{i=1}^{n}\prod_{j=1}%
^{m}\tilde{\phi}(\theta_{i}-z_{j})\prod_{1\leq i<j\leq m}\tau(z_{i}%
-z_{j})\,,~\tau(z)=\frac{1}{\tilde{\phi}(-z)\tilde{\phi}(z)}\,. \label{h}%
\end{equation}
The dependence on the specific operator $\mathcal{O}(x)$ is encoded in the
scalar p-function $p^{\mathcal{O}}(\underline{\theta},\underline{z})$ which is
in general a simple function of $e^{\theta_{i}}$ and $e^{z_{j}}$%
.\begin{figure}[tbh]%
\[
{\unitlength4mm\begin{picture}(27,13)
\thicklines\put(1,0){
\put(0,0){$\bullet~\theta_n-2\pi i$}
\put(0,4.8){$\bullet$}
\put(.15,3.8){$\theta_n-2\pi i\frac1N$}
\put(0,6){$\bullet~\theta_n$}
\put(0,11){$\bullet\,\theta_n+2\pi i(1-\frac1N)$}
}
\put(7,6.5){\dots}
\put(12,0){
\put(0,0){$\bullet~\theta_2-2\pi i$}
\put(0,4.8){$\bullet$}
\put(.19,3.8){$\theta_2-2\pi i\frac1N$}
\put(0,6){$\bullet~\theta_2$}
\put(0,11){$\bullet\,\theta_2+2\pi i(1-\frac1N)$}
}
\put(20,1){
\put(0,0){$\bullet~\theta_1-2\pi i$}
\put(0,4.8){$\bullet$}
\put(.19,3.8){$\theta_1-2\pi i\frac1N$}
\put(0,6){$\bullet~\theta_1$}
\put(0,11){$\bullet\,\theta_1+2\pi i(1-\frac1N)$}
}
\put(9,2){\vector(1,0){0}}
\put(9,5.6){\vector(-1,0){0}}
\put(9,8.6){\vector(1,0){0}}
\put(-1,2.5){\oval(36,1)[br]}
\put(24,2.5){\oval(14,1)[tl]}
\put(24,4.8){\oval(3,3.6)[r]}
\put(24,6.1){\oval(14,1)[tl]}
\put(2,6.1){\oval(30,1)[br]}
\put(2,7.1){\oval(4,3)[l]}
\put(2,8.6){\line(1,0){25}}
\end{picture}
}
\]
\caption{The integration contour $\mathcal{C}_{\underline{\theta}}$. The
bullets refer to poles of the integrand in (\ref{K}).}%
\label{f5.1}%
\end{figure}
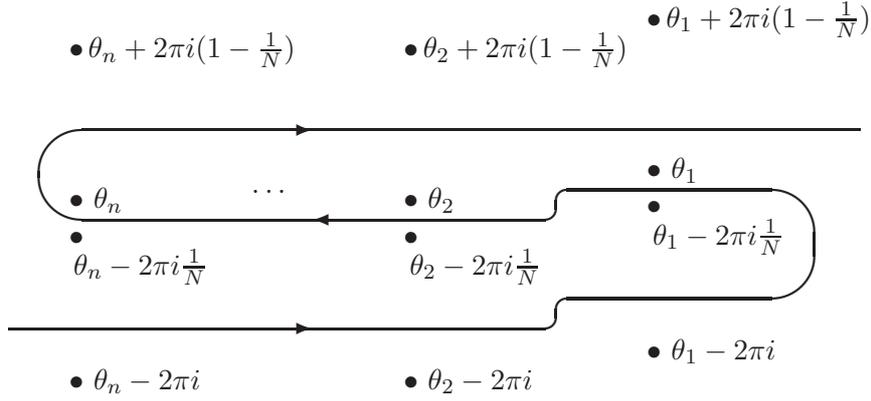

\paragraph{Bethe state:}

The state $\tilde{\Psi}_{\underline{\alpha}}$ in (\ref{K}) is a linear
combination of the basic Bethe ansatz co-vectors
\begin{equation}
\tilde{\Psi}_{\underline{\alpha}}(\underline{\theta},\underline{z}%
)=L_{\underline{\beta}}(\underline{z})\tilde{\Phi}_{\underline{\alpha}%
}^{\underline{\beta}}(\underline{\theta},\underline{z})\,,~~\text{with
}1<\beta_{i}\leq N\,. \label{Psi}%
\end{equation}
As usual in the context of the algebraic Bethe ansatz \cite{FST,TF} the basic
Bethe ansatz co-vectors are obtained from the monodromy matrix%
\begin{align}
\tilde{T}_{1\dots n,0}(\underline{\theta},z)  &  =\tilde{S}_{10}(\theta
_{1}-z)\,\cdots\tilde{S}_{n0}(\theta_{n}-z)=%
\begin{array}
[c]{c}%
\unitlength3mm\begin{picture}(8,4.5)(1,0) \put(1,2){\line(1,0){7}} \put(7,.8){$0$} \put(3,0){\line(0,1){4}} \put(2.2,0){$1$} \put(6,0){\line(0,1){4}} \put(5.,0){$n$} \put(3.7,2.7){$\dots$} \end{picture}
\end{array}
\label{T}\\
&  \equiv\left(
\begin{array}
[c]{cc}%
\tilde{A}_{1\dots n}(\underline{\theta},z) & \tilde{B}_{1\dots n,\beta
}(\underline{\theta},z)\\
\tilde{C}_{1\dots n}^{\beta}(\underline{\theta},z) & \tilde{D}_{1\dots
n,\beta}^{\beta^{\prime}}(\underline{\theta},z)
\end{array}
\right)  ,~~2\leq\beta,\beta^{\prime}\leq N\,.\nonumber
\end{align}
where the S-matrix $\tilde{S}_{i0}$ is given by (\ref{St}).

The reference co-vector is defined as usual by $\Omega\tilde{B}_{\beta}=0$
which implies
\[
\,\Omega_{\underline{\alpha}}=\delta_{\alpha_{1}}^{1}\dots\delta_{\alpha_{n}%
}^{1}\,.
\]
It is an eigenstates of $\tilde{A}$ and $\tilde{D}_{\beta}^{\beta^{\prime}}$%
\[
\Omega\,\tilde{A}(\underline{\theta},z)=\Omega\,,\,~\Omega\tilde{D}_{\beta
}^{\beta^{\prime}}(\underline{\theta},z)=\delta_{\beta}^{\beta^{\prime}}%
\prod\limits_{i=1}^{n}\tilde{b}(\theta_{i}-z)\Omega\,.
\]
where the indices $1\dots n$ are suppressed. The basic Bethe ansatz co-vectors
in (\ref{Psi}) are defined as%
\begin{equation}
\tilde{\Phi}_{\underline{\alpha}}^{\underline{\beta}}(\underline{\theta
},\underline{z})=\left(  \Omega\tilde{C}^{\beta_{m}}(\underline{\theta}%
,z_{m})\cdots\tilde{C}^{\beta_{1}}(\underline{\theta},z_{1})\right)
_{\underline{\alpha}}=%
\begin{array}
[c]{c}%
\unitlength4mm\begin{picture}(9,7) \put(9,5){\oval(14,2)[lb]} \put(9,5){\oval(18,6)[lb]} \put(4,1){\line(0,1){4}} \put(3.8,.4){$\alpha_1$} \put(8,1){\line(0,1){4}} \put(7.8,.4){$\alpha_n$} \put(-.2,5.4){$\beta_1$} \put(1.8,5.4){$\beta_m$} \put(3.8,5.4){$1$} \put(7.8,5.4){$1$} \put(9.2,1.8){$1$} \put(9.2,3.8){$1$} \put(3,2.5){$\theta_1$} \put(6.8,2.5){$\theta_{n}$} \put(.8,2.5){$z_1$} \put(1.7,3.6){$z_m$} \put(5.4,4.5){$\dots$} \put(8.5,2.6){$\vdots$} \end{picture}
\end{array}
\label{Phi0}%
\end{equation}
where $1<\beta_{i}\leq N$.

The technique of the \textbf{`nested Bethe ansatz'} means that for the
coefficients $L_{\underline{\beta}}(\underline{z})$ in (\ref{Psi}) one makes
the analogous construction as for $K_{\underline{\alpha}}(\underline{\theta})$
in (\ref{K}), where now the indices $\underline{\beta}$ take only the values
$2\leq\beta_{i}\leq N$. This nesting is repeated until the space of the
coefficients becomes one dimensional. The final result is%
\begin{equation}
K_{\underline{\alpha}}^{\mathcal{O}}(\underline{\theta})=\int d\underline
{\underline{z}}\,\tilde{h}(\underline{\theta},\underline{\underline{z}%
})\,p^{\mathcal{O}}(\underline{\theta},\underline{\underline{z}})\,\tilde
{\Phi}_{\underline{\alpha}}(\underline{\theta},\underline{\underline{z}})
\label{Kc}%
\end{equation}
with the complete h-function
\begin{equation}
\tilde{h}\,(\underline{\theta},\underline{\underline{z}})=\prod_{j=0}%
^{N-2}\tilde{h}(\underline{z}_{j},\underline{z}_{j+1})\,,~~\underline{z}%
_{0}=\underline{\theta} \label{hh}%
\end{equation}
and the complete Bethe ansatz state
\begin{equation}
\tilde{\Phi}_{\underline{\alpha}}(\underline{\theta},\underline{\underline{z}%
})=(\tilde{\Phi}^{(N-2)})_{\underline{\alpha}_{N-2}}^{\underline{\alpha}%
_{N-1}}(\underline{z}^{(N-2)},\underline{z}^{(N-1)})\dots(\tilde{\Phi}%
^{(1)})_{\underline{\alpha}_{1}}^{\underline{\alpha}_{2}}(\underline{z}%
^{(1)},\underline{z}^{(2)})\tilde{\Phi}_{\underline{\alpha}}^{\underline
{\alpha}_{1}}(\underline{\theta},\underline{z}^{(1)}) \label{PHI}%
\end{equation}
where $\underline{\underline{z}}=(\underline{z}^{(1)},\dots,\underline
{z}^{(N-1)}),~\underline{z}^{(j)}=(z_{1}^{(j)},\dots,z_{n_{j}}^{(j)})$ and
$\underline{\alpha}_{N-1}=(N,\dots,N)$.

$\,$It is well known (see \cite{BKZ2}) that the `off-shell' Bethe ansatz
states are highest weight states if they satisfy certain matrix difference
equations. If there are $n$ particles the $SU(N)$ weights are \cite{BFK3}%
\begin{align}
w  &  =\left(  n-n_{1},n_{1}-n_{2},\dots,n_{N-2}-n_{N-1},n_{N-1}\right)
\label{w}\\
&  =w^{\mathcal{O}}+L(1,\dots,1)\nonumber
\end{align}
where $n_{1}=m,n_{2},\dots$ are the numbers of $C$ operators in the various
levels of the nesting, $w^{\mathcal{O}}$ is the weight vector of the operator
$\mathcal{O}$ and $L=0,1,2,\dots$; note that $w=(1,\dots,1)$ correspond to the
vacuum sector.

\section{Rapidity space clustering}

\label{s3}

We shift $k$ of the $n$ rapidities in the form factor $F_{\underline{\alpha}%
}^{\mathcal{O}}(\underline{\theta})$ (\ref{FK}) to $\infty$ and define

$\underline{\theta}_{W}=(\theta_{1}+W,\dots,\theta_{k}+W,\theta_{k+1}%
,\dots,\theta_{k+l})=(\underline{\hat{\theta}}+W,\underline{\check{\theta}})$.

\noindent We investigate the behavior of $F_{\underline{\alpha}}^{\mathcal{O}%
}(\underline{\theta}_{W})$ for $W\rightarrow\infty:$ The result is of the form%
\begin{equation}
F_{\underline{\alpha}}^{\mathcal{O}}(\underline{\theta}_{W})\overset
{W\rightarrow\infty}{\rightarrow}c_{\hat{\mathcal{O}}\mathcal{\check{O}}%
}^{\mathcal{O}}(k,l,W)F_{\underline{\hat{\alpha}}}^{\hat{\mathcal{O}}%
}(\underline{\hat{\theta}})F_{\underline{\check{\alpha}}}^{\mathcal{\check{O}%
}}(\underline{\check{\theta}})\,. \label{FFF}%
\end{equation}
We calculate the functions $c_{\hat{\mathcal{O}}\mathcal{\check{O}}%
}^{\mathcal{O}}(W)$ for several operators.

\subsection{Examples of local fields:}

In this article we consider the following fields:

\paragraph{The $SU(N)$ Noether current}%

\[
J_{a}^{\mu}=\bar{\psi}_{\beta}\gamma^{\mu}\left(  T_{a}\right)  _{\alpha
}^{\beta}\psi^{\alpha}%
\]
transforms as the adjoint representation with highest weights $w^{J}%
=(2,1,\dots,1,0)$. The $N^{2}-1$ generators of $SU(N)$ satisfy
\[
\left[  T_{a},T_{b}\right]  =if_{abc}T_{c},~~\operatorname*{Tr}T_{a}%
=0,~~\operatorname*{Tr}(T_{a}T_{b})=\tfrac{1}{2}\delta_{ab}\,.
\]
The conservation law $\partial_{\mu}J_{a}^{\mu}(x)=0$ implies that $J_{a}%
^{\mu}(x)$ may be written in terms of the pseudo potential $J_{a}(x)$ as%
\begin{equation}
J_{a}^{\mu}(x)=\epsilon^{\mu\nu}\partial_{\nu}J_{a}(x) \label{J}%
\end{equation}
with the quantum numbers%
\begin{equation}%
\begin{array}
[c]{lcl}%
\text{charge} &  & Q^{J}=0\\
\text{weight vector} &  & w^{J}=\left(  2,1,\dots,1,0\right) \\
\text{statistics factor} &  & \sigma^{J}=1\\
\text{spin} &  & s^{J}=0.
\end{array}
\label{QnJ}%
\end{equation}
Due the Swieca et al \cite{KKS} the bound state of $N-1$ particles is to be
identified with the anti-particle. This means that the anti-particle
$\bar{\alpha}$ of a fundamental particle $\alpha$ of rank $1$ is a bound state
of rank $N-1$
\begin{equation}
\bar{\alpha}=(\rho)=(\rho_{1}\dots\rho_{N-1}),~\text{with }\rho_{1}<\dots
<\rho_{N-1},~\rho_{i}\neq\alpha\,. \label{bs}%
\end{equation}
The charge conjugation matrix is given by%
\begin{equation}
\mathbf{C}_{\beta\bar{\alpha}}=\mathbf{C}_{\beta(\rho_{1}\dots\rho_{N-1}%
)}=\mathbf{C}^{\bar{\alpha}\beta}=\epsilon_{\beta\rho_{1}\dots\rho_{N-1}}\,
\label{C}%
\end{equation}
with $\mathbf{C}_{\beta\bar{\alpha}}\mathbf{C}^{\bar{\alpha}\gamma}%
=\delta_{\beta}^{\gamma}$. In terms of fields this means $\bar{\psi}_{\beta
}=\mathbf{C}_{\beta(\rho)}\bar{\psi}^{(\rho)}=\mathbf{C}_{\beta(\rho)}%
\psi^{\rho_{1}}\dots\psi^{\rho_{N-1}}$.

For the Bethe ansatz the formulation of the Noether current given by%
\[
J_{\mu}^{\alpha(\rho)}=\bar{\psi}^{(\rho)}\gamma_{\mu}\psi^{\alpha}%
-\mathbf{C}^{\alpha(\rho)}\mathbf{C}_{(\sigma)\beta}\bar{\psi}^{(\sigma
)}\gamma_{\mu}\psi^{\beta}/N
\]
with $\mathbf{C}_{\alpha(\rho)}J_{\mu}^{\alpha(\rho)}=0$ is more convenient,
which means for the pseudo potentials
\begin{equation}
J_{a}=\mathbf{C}_{\beta(\rho)}\left(  T_{a}\right)  _{\alpha}^{\beta}%
J^{\alpha(\rho)}. \label{JJ}%
\end{equation}
Because the Bethe ansatz yields highest weight states we obtain the matrix
elements of the highest weight component$\ J(x)=J^{1\bar{N}}(x)=J^{1(12\dots
N-1)}(x)\,$. The form factor is given by (\ref{FK}) and (\ref{K}) with the
p-function for the operator $J(x)$ \cite{BFK1}%
\begin{equation}
p^{J}(\underline{\theta},\underline{\underline{z}})=e^{i\pi\frac{1}{N}n_{1}%
}\left(  {\textstyle\prod\nolimits_{i=1}^{n}} e^{-\frac{1}{2}\theta_{i}%
}\right)  \left(  {\textstyle\prod\nolimits_{i=1}^{n_{1}}} e^{\frac{1}{2}%
z_{i}^{(1)}}\right)  \left(  {\textstyle\prod\nolimits_{i=1}^{n_{N-1}}}
e^{\frac{1}{2}z_{i}^{(N-1)}}\right)  /\left(  {\textstyle\sum\nolimits_{i=1}%
^{n}} e^{-\theta_{i}}\right)  \label{pJ}%
\end{equation}
for $n=0\operatorname{mod}N$. The general weight formula of the Bethe states
(\ref{w}) implies that the numbers of integrations in (\ref{Kc}) satisfy
\begin{equation}
n_{j}=n\left(  1-j/N\right)  -1,~j=1,\dots,N-1. \label{nJ}%
\end{equation}
In particular the one particle and one anti-particle form factor is
\cite{BFK1}%
\begin{align}
F_{\alpha\bar{\beta}}^{J_{a}}(\theta,\omega)  &  =\left(  T_{a}\right)
_{\alpha\bar{\beta}}\frac{1}{\cosh\frac{1}{2}\left(  \theta-\omega\right)
}\bar{F}(\theta-\omega)/\bar{F}(i\pi)\label{FJ2}\\
F_{\alpha\bar{\beta}}^{J^{\gamma\bar{\delta}}}(\theta,\omega)  &  =\left(
\delta_{\alpha}^{\gamma}\delta_{\bar{\beta}}^{\bar{\delta}}-\mathbf{C}%
^{\gamma\bar{\delta}}\mathbf{C}_{\bar{\beta}\alpha}/N\right)  \frac{1}%
{\cosh\frac{1}{2}\left(  \theta-\omega\right)  }\bar{F}(\theta-\omega)/\bar
{F}(i\pi)\nonumber
\end{align}
where $\left(  T_{a}\right)  _{\alpha\bar{\beta}}=\mathbf{C}_{\delta\bar
{\beta}}\left(  T_{a}\right)  _{\alpha}^{\delta}$ and $\bar{F}(\theta)$
defined in (\ref{Fbar}) is the "minimal form factor function" for one particle
and one anti-particle.

\paragraph{Energy momentum $T^{\mu\nu}$:}

We write the energy momentum tensor in terms of an energy momentum potential%
\begin{equation}
T^{\mu\nu}(x)=R^{\mu\nu}(i\partial_{x})T(x),~~R^{\mu\nu}(P)=-P^{\mu}P^{\nu
}+g^{\mu\nu}P^{2} \label{Tmunu}%
\end{equation}
with
\begin{equation}%
\begin{array}
[c]{lcl}%
\text{charge} &  & Q^{T}=0\\
\text{weight vector} &  & w^{T}=\left(  0,\dots,0\right) \\
\text{statistics factor} &  & \sigma^{T}=1\\
\text{spin} &  & s^{T}=0.
\end{array}
\end{equation}
We propose the p-function of the potential
\[
p^{T}(\underline{\theta},\underline{\underline{z}})=\frac{\sum e^{z_{j}^{(1)}%
}}{\sum e^{\theta_{j}}}-\frac{\sum e^{-z_{j}^{(1)}}}{\sum e^{-\theta_{j}}%
}=p^{T_{+}}(\underline{\theta},\underline{z})+p^{T_{-}}(\underline{\theta
},\underline{z}).
\]
The general weight formula of Bethe states (\ref{w}) implies that the numbers
of integrations in (\ref{Kc}) satisfy
\begin{equation}
n_{j}=n\left(  1-j/N\right)  ,~j=1,\dots,N-1. \label{nT}%
\end{equation}
The one particle and one anti-particle form factors are \cite{BFK1}%
\begin{align}
F_{\alpha\bar{\beta}}^{T}(\theta,\omega)  &  =\mathbf{C}_{\alpha\bar{\beta}%
}\frac{-i}{\cosh\frac{1}{2}\left(  \theta-\omega\right)  }\frac{1}%
{\theta-\omega-i\pi}\bar{F}(\theta-\omega)/\bar{F}(i\pi)\label{FT2}\\
F_{\alpha\bar{\beta}}^{T^{\rho\sigma}}(\theta,\omega)  &  =4m^{2}%
\mathbf{C}_{\alpha\bar{\beta}}e^{\frac{1}{2}\left(  \rho+\sigma\right)
\left(  \theta+\omega+i\pi\right)  }\frac{\sinh\frac{1}{2}\left(
\theta-\omega-i\pi\right)  }{\theta-\omega-i\pi}\bar{F}(\theta-\omega)/\bar
{F}(i\pi),~\rho,\sigma=\pm\nonumber
\end{align}

\paragraph{The iso-scalar field $\phi(x)$}

with the quantum numbers%
\[%
\begin{array}
[c]{lcl}%
\text{charge} &  & Q^{\phi}=0\\
\text{weight vector} &  & w^{\phi}=\left(  0,\dots,0\right) \\
\text{statistics factor} &  & \sigma^{\phi}=e^{-i\eta}\\
\text{spin} &  & s^{\phi}=0,
\end{array}
\]
and the p-function
\begin{equation}
p^{\phi}(\underline{\theta},\underline{\underline{z}})=e^{i\frac{\pi}{N}n_{1}%
}\left(  {\textstyle\prod\nolimits_{i=1}^{n}} e^{-\left(  1-\frac{1}%
{N}\right)  \theta_{i}}\right)  \left(  {\textstyle\prod\nolimits_{i=1}%
^{n_{1}}} e^{z_{j}^{(1)}}\right)  \label{pphi}%
\end{equation}
for $n=0\operatorname{mod}N$. The general weight formula of Bethe states
(\ref{w}) implies that the numbers of integrations in (\ref{Kc}) satisfy
\begin{equation}
n_{j}=n\left(  1-j/N\right)  ,~j=1,\dots,N-1. \label{nphi}%
\end{equation}
The one particle and one anti-particle form factor is%
\[
F_{\alpha\bar{\delta}}^{\phi}(\theta,\omega)=\mathbf{C}_{\alpha\bar{\delta}%
}2i\left(  1-\sigma^{\phi}\right)  \frac{e^{-\left(  \frac{1}{2}-\frac{1}%
{N}\right)  \left(  \theta-\omega-i\pi\right)  }}{\theta-\omega-i\pi}%
\frac{\bar{F}(\theta-\omega)}{\bar{F}(i\pi)}%
\]
if we normalize the field by $\left\langle 0|\phi(x)|0\right\rangle =1$.

\paragraph{The fundamental field $\psi^{\alpha}(x)$}

of the chiral $SU(N)$ Gross-Neveu model with the quantum numbers
\begin{equation}%
\begin{array}
[c]{lcl}%
\text{charge} &  & Q^{\psi}=1\\
\text{weight vector} &  & w^{\psi}=\left(  1,0,\dots,0\right) \\
\text{statistics factor} &  & \sigma^{\psi}=e^{\left(  1-\frac{1}{N}\right)
i\pi}\\
\text{spin} &  & s^{\psi}=-\frac{1}{2}\left(  1-\frac{1}{N}\right)
\end{array}
\label{Qpsi}%
\end{equation}
The p-function of the highest weight component $\psi=\psi^{1}$ for
$n=1\operatorname{mod}N$ is \cite{BFK1}%
\begin{equation}
p^{\psi}(\underline{\theta},\underline{\underline{z}})=e^{\frac{1}{2}%
n_{1}i\eta}\left(  {\textstyle\prod\nolimits_{i=1}^{n}} e^{-\frac{1}{2}\left(
1-\frac{1}{N}\right)  \theta_{i}}\right)  \left(  {\textstyle\prod
\nolimits_{i=1}^{n_{1}}} e^{\frac{1}{2}z_{i}^{(1)}}\right)  \label{ppsi}%
\end{equation}
and the 1-particle matrix element is%
\begin{equation}
F_{\alpha}^{\psi}(\theta)=\delta_{\alpha}^{1}e^{-\frac{1}{2}\left(  1-\frac
{1}{N}\right)  \theta}\,. \label{F1psi}%
\end{equation}
The general weight formula of Bethe states (\ref{w}) with $w^{\psi}=\left(
1,0,\dots,0\right)  $ implies that the numbers of integrations in (\ref{Kc})
satisfy
\begin{equation}
n_{j}=\left(  n-1\right)  \left(  1-j/N\right)  ,~j=1,\dots,N-1. \label{npsi}%
\end{equation}

\paragraph{The field $\chi^{\bar{\alpha}}(x)$}

with the quantum numbers%
\[%
\begin{array}
[c]{lcl}%
\text{charge} &  & Q^{\chi}=N-1\\
\text{weight vector} &  & w^{\chi}=\left(  1,1,\dots1,0\right) \\
\text{statistics factor} &  & \sigma_{1}^{\chi}=e^{i\pi\left(  N-\frac{1}%
{N}\right)  }\\
\text{spin} &  & s^{\chi}=\frac{1}{2}\left(  1-\frac{1}{N}\right)  .
\end{array}
\]
The p-function of the highest weight component $\chi=\chi^{\bar{N}}$ for
$n=\left(  N-1\right)  \operatorname{mod}N$ is%
\begin{equation}
p^{\chi}(\underline{\theta},\underline{\underline{z}})=e^{\left(  n_{1}%
+\frac{1}{2}n_{N-1}\right)  i\eta}\left(  {\textstyle\prod\nolimits_{j=1}^{n}}
e^{-\left(  1-\frac{1}{2N}\right)  \theta_{j}}\right)  \left(
{\textstyle\prod\nolimits_{j=1}^{n_{1}}} e^{z_{j}^{(1)}}\right)  \left(
{\textstyle\prod\nolimits_{j=1}^{n_{N-1}}} e^{\frac{1}{2}z_{j}^{(N-1)}%
}\right)  /{\textstyle\sum} e^{-\theta_{i}} \label{pchi}%
\end{equation}
with $n_{j}=\left(  n+1\right)  \left(  1-j/N\right)  -1$ and the
1-anti-particle matrix element is (see \cite{BFK3})%
\begin{equation}
F_{\bar{\alpha}}^{\chi^{\bar{\beta}}}(\omega)=\delta_{\bar{\alpha}}%
^{\bar{\beta}}e^{\frac{1}{2}\left(  1-\frac{1}{N}\right)  \omega}.
\label{F1chi}%
\end{equation}

\subsection{Results}

As examples of the general formula (\ref{FFF}) we obtain:

\label{0}

\begin{enumerate}
\item Particle number $n=0\operatorname{mod}N$ and $k=0\operatorname{mod}N$%
\begin{align}
&  F_{\underline{\alpha}}^{J_{a}}(\underline{\theta}_{W})\overset
{W\rightarrow\infty}{\rightarrow}-2\eta W^{-1}f_{abc}F_{\underline{\hat
{\alpha}}}^{J_{b}}(\underline{\hat{\theta}})F_{\underline{\check{\alpha}}%
}^{J_{c}}(\underline{\check{\theta}})~,~~\text{see Theorem \ref{t1}}%
\label{FJ}\\
&  F_{\underline{\alpha}}^{\phi}(\underline{\theta}_{W})\overset
{W\rightarrow\infty}{\rightarrow}F_{\underline{\hat{\alpha}}}^{\phi
}(\underline{\hat{\theta}})F_{\underline{\check{\alpha}}}^{\phi}%
(\underline{\check{\theta}})~,~~\text{see Theorem \ref{t1a}}\label{Fphi}\\
&  F_{\underline{\alpha}}^{T}(\underline{\theta}_{W})\overset{W\rightarrow
\infty}{\rightarrow}2\eta W^{-2}F_{\underline{\hat{\alpha}}}^{J_{a}%
}(\underline{\hat{\theta}})F_{\underline{\check{\alpha}}}^{J_{a}}%
(\underline{\check{\theta}})~,~~\text{see Theorem \ref{t1b}.} \label{FT}%
\end{align}

\item Particle number $n=0\operatorname{mod}N$ and $k=1\operatorname{mod}N$%
\begin{align}
F_{\underline{\alpha}}^{J}(\underline{\theta}_{W})  &  \overset{W\rightarrow
\infty}{\rightarrow}c_{\psi\chi}^{J}(k,l,W)F_{\underline{\hat{\alpha}}}^{\psi
}(\underline{\hat{\theta}})F_{\underline{\check{\alpha}}}^{\chi}%
(\underline{\check{\theta}})~,~~\text{see Theorem \ref{t2}}\label{FJpsichi}\\
&  c_{\psi\chi}^{J}(k,l,W)=e^{i\pi l_{1}}d\,W^{\frac{1}{N^{2}}}e^{-\frac{1}%
{2}\left(  1-\frac{1}{N}\right)  W}\nonumber\\
F_{\underline{\alpha}}^{T}(\underline{\theta}_{W})  &  \overset{W\rightarrow
\infty}{\rightarrow}c_{\psi\chi}^{T}(k,l,W)\mathbf{C}_{\alpha\bar{\beta}%
}F_{\underline{\hat{\alpha}}}^{\psi^{\alpha}}(\underline{\hat{\theta}%
})F_{\underline{\check{\alpha}}}^{\chi^{\bar{\beta}}}(\underline{\check
{\theta}})~,~~\text{see Conjecture \ref{t2a}}\label{FTpsichi}\\
&  c_{\psi\chi}^{T}(k,l,W)=-ie^{i\pi l_{1}}d\,W^{\frac{1}{N^{2}}-1}%
e^{-\frac{1}{2}\left(  1-\frac{1}{N}\right)  W}\nonumber
\end{align}
with the constant $d=2\left(  2\pi\right)  ^{-\frac{1+N}{N^{2}}}%
e^{-i\pi\left(  N+\frac{1}{2N}\right)  }/\bar{F}(i\pi)$.

\item Particle number $n=1\operatorname{mod}N$ and $k=0\operatorname{mod}N$%
\begin{align}
&  F_{\underline{\alpha}}^{\psi^{\beta}}(\underline{\theta}_{W})\overset
{W\rightarrow\infty}{\rightarrow}i\eta W^{-1}\,\mathbf{C}_{\gamma\bar{\delta}%
}F_{\underline{\hat{\alpha}}}^{J^{\beta\bar{\delta}}}(\underline{\hat{\theta}%
})F_{\underline{\check{\alpha}}}^{\psi^{\gamma}}(\underline{\check{\theta}%
})\label{psiJpsi}\\
&  =2i\eta W^{-1}F_{\underline{\hat{\alpha}}}^{J_{a}}(\underline{\hat{\theta}%
})\left(  T_{a}\right)  _{\delta}^{\beta}F_{\underline{\check{\alpha}}}%
^{\psi^{\delta}}(\underline{\check{\theta}})~,~~\text{see Theorem \ref{t3}%
.}\nonumber
\end{align}

\item Particle number $n=1\operatorname{mod}N$ and $k=1\operatorname{mod}N$%
\begin{equation}
F_{\underline{\alpha}}^{\psi^{\alpha}}(\underline{\theta}_{W})\overset
{W\rightarrow\infty}{\rightarrow}e^{i\pi l_{1}}e^{-\frac{1}{2}\left(
1-\frac{1}{N}\right)  W}F_{\underline{\hat{\alpha}}}^{\psi^{\alpha}%
}(\underline{\hat{\theta}})F_{\underline{\check{\alpha}}}^{\phi}%
(\underline{\check{\theta}})~,~~\text{see Theorem \ref{t4}.} \label{psipsiphi}%
\end{equation}

\end{enumerate}

\section{Proofs}

\label{s4}We use the short notation $\underline{\theta}_{W}$ of Section
\ref{s3} and in addition $\underline{\underline{z}}_{W}=(\underline{z}%
_{W}^{(1)},\dots,\underline{z}_{W}^{(N-1)})$ where we shift $k_{j}$ of the
$z_{i}^{(j)}$ and define

\noindent\ $\underline{z}_{W}^{(j)}=(z_{1}^{(j)}+W,\dots,z_{k_{j}}%
^{(j)}+W,z_{k_{j}+1}^{(j)},\dots,z_{n_{j}}^{(j)})=(\underline{\hat{z}}%
^{(j)}+W,\underline{\check{z}}^{(j)}),~(j=1,\dots,N-1)$.

\noindent The choice of the $k_{j}$ integrations out of the $n_{j}$ ones in
(\ref{K}) is arbitrary therefore there is a factor of $\binom{n_{j}}{k_{j}}$
such that $\binom{n_{j}}{k_{j}}\frac{1}{n_{j}!}=\frac{1}{k_{j}!}\frac{1}%
{l_{j}!},~(l_{j}=n_{j}-k_{j})$ and there is the replacement
\[
\int d\underline{z}_{j}\dots\rightarrow\int d\underline{\hat{z}}_{j}\dots\int
d\underline{\check{z}}_{j}\dots
\]
The asymptotic behavior of the form factors given by (\ref{FK}) and (\ref{K})
with $\underline{\theta}=\underline{\theta}_{W}$ for $W\rightarrow\infty$ is
obtained from the asymptotic behavior of $F(\underline{\theta}_{W}),~\tilde
{h}(\underline{\theta}_{W},\underline{\underline{z}}_{W}),~\tilde{\Psi
}(\underline{\theta}_{W},\underline{\underline{z}}_{W})$ and the p-functions
(see Appendix \ref{sa0}). In the following, some equations are written for
simplicity up to constant factors. Constant factors in eq. (\ref{FFF}) are
finally obtained by form factor equation (iii).

\subsection{Theorem \ref{t1}}

\begin{theorem}
\label{t1}The form factor of the pseudo-potential of the current for particle
number $n=0\operatorname{mod}N$ and $k=0\operatorname{mod}N$ shows the cluster
behavior%
\[
F_{\underline{\alpha}}^{J^{\beta(\sigma)}}(\underline{\theta}_{W}%
)\overset{W\rightarrow\infty}{\rightarrow}i\eta\frac{1}{W}\mathbf{C}%
_{\gamma(\lambda)}\left(  F_{\underline{\hat{\alpha}}}^{J^{\beta(\lambda)}%
}(\underline{\hat{\theta}})F_{\underline{\check{\alpha}}}^{J^{\gamma(\sigma)}%
}(\underline{\check{\theta}})-F_{\underline{\hat{\alpha}}}^{J^{\gamma(\sigma
)}}(\underline{\hat{\theta}})F_{\underline{\check{\alpha}}}^{J^{\beta
(\lambda)}}(\underline{\check{\theta}})\right)
\]
which is equivalent to\footnote{For $SU(2)$ this result (\ref{FJ}) was
obtained previously by Smirnov \cite{Sm}.}%
\[
F_{\underline{\alpha}}^{J_{a}}(\underline{\theta}_{W})\overset{W\rightarrow
\infty}{\rightarrow}-2\eta\frac{1}{W}f_{abc}F_{\underline{\hat{\alpha}}%
}^{J_{b}}(\underline{\hat{\theta}})F_{\underline{\check{\alpha}}}^{J_{c}%
}(\underline{\check{\theta}})\,.
\]

\end{theorem}

\begin{proof}
We use the short notations of (\ref{FK}) $\dots$ (\ref{PHI}) and investigate
(for $J=J^{1\bar{N}}$)%
\[
F_{\underline{\alpha}}^{J}(\underline{\theta}_{W})=N_{n}^{J}F(\underline
{\theta}_{W})\int d\underline{\underline{z}}\tilde{h}\,(\underline{\theta}%
_{W},\underline{\underline{z}}_{W})p^{J}(\underline{\theta}_{W},\underline
{\underline{z}}_{W})\tilde{\Phi}_{\underline{\alpha}}(\underline{\theta}%
_{W},\underline{\underline{z}}_{W}).
\]
From the asymptotic behavior of $F(\underline{\theta}_{W})\tilde
{h}\,(\underline{\theta}_{W},\underline{\underline{z}}_{W})$ and
$p^{J}(\underline{\theta}_{W},\underline{z}_{W}^{(1)},\underline{z}%
_{W}^{(N-1)})$ in (\ref{Fha}), (\ref{pJa}) and (\ref{expad}) we derive for
$W\rightarrow\infty$ the exponential behavior%
\begin{equation}
F(\underline{\theta}_{W})\tilde{h}\,(\underline{\theta}_{W},\underline
{\underline{z}}_{W})p^{J}(\underline{\theta}_{W},\underline{\underline{z}}%
_{W})\varpropto\left(  e^{-\frac{1}{2}W}\right)  ^{\tilde{k}_{1}^{2}+\tilde
{k}_{N-1}^{2}+\sum_{j=1}^{N-2}\left(  \tilde{k}_{j}-\tilde{k}_{j+1}\right)
^{2}} \label{expJ}%
\end{equation}
where $\tilde{k}_{j}=k_{j}-k\left(  1-j/N\right)  $. For
$k=0\operatorname{mod}N$ the leading behavior $\left(  e^{-\frac{1}{2}%
W}\right)  ^{0}$ is obtained for $\tilde{k}_{j}=0$, therefore%
\[
k_{j}=k\left(  1-j/N\right)  ,~l_{j}=l\left(  1-j/N\right)  -1,~j=1,\dots
,N-1.
\]
For these values of $k_{j}$ and $l_{j}$ we obtain, more precisely, with
(\ref{Fha}), (\ref{pJa}) and (\ref{PHI0}) in leading order the asymptotic
behavior (up to a constant factor)%
\begin{multline}
F(\underline{\theta}_{W})\tilde{h}\,(\underline{\theta}_{W},\underline
{\underline{z}}_{W})p^{J}(\underline{\theta}_{W},\underline{\underline{z}}%
_{W})\tilde{\Phi}_{\underline{\alpha}}(\underline{\theta}_{W},\underline
{\underline{z}}_{W})\label{F0a}\\
\overset{W\rightarrow\infty}{\rightarrow}\left(  F(\underline{\hat{\theta}%
})\tilde{h}\,(\underline{\hat{\theta}},\underline{\underline{\hat{z}}}%
)\tilde{\Phi}_{\underline{\alpha}}(\underline{\hat{\theta}},\underline
{\underline{\hat{z}}})\right)  \left(  F(\underline{\check{\theta}})\tilde
{h}\,(\underline{\check{\theta}},\underline{\underline{\check{z}}}%
)p^{J}(\underline{\check{\theta}},\underline{\underline{\check{z}}}%
)\tilde{\Phi}_{\underline{\alpha}}(\underline{\check{\theta}},\underline
{\underline{\check{z}}})\right)  .
\end{multline}
The $\underline{\underline{\hat{z}}}$-integral vanishes because of Lemma
\ref{l1} and therefore in leading order
\[
F_{\underline{\alpha}}^{J}(\underline{\theta}_{W})\rightarrow0.
\]

\medskip

\noindent\textbf{Order }$\frac{1}{W}$: we have to apply the asymptotic
behavior of the h-function (\ref{ha}), (\ref{h0}), (\ref{h1}) and the Bethe
state (\ref{PHI1}) and (\ref{PHI1a}).

We present a complete proof of this $\frac{1}{W}$-term for $SU(2)$ and for
general $N$ the example of Appendix \ref{sex} for one particle and one
anti-particle. In addition we show consistency of the general clustering
formula with the form factor equation (iii) (see Remark \ref{rem2}).

We have to consider the 2 contributions:

\medskip

\noindent A) From the \textbf{h-function: }Note that because of Lemma \ref{l1}
in $\tilde{h}_{1}(\underline{\theta},\underline{z})$ of (\ref{h1}) only the
$\hat{z}_{j}$-dependent terms contribute. Therefore we get on the rhs of
(\ref{F0a}) from $\tilde{h}_{1}$ for $k_{1}=k\left(  1-1/N\right)
,~l_{1}=l\left(  1-1/N\right)  -1$%
\[
\left(  F(\underline{\hat{\theta}})\tilde{h}(\underline{\hat{\theta}%
},\underline{\hat{z}})\left(  -{\textstyle\sum} \hat{z}_{j}\right)
\tilde{\Phi}_{\underline{\hat{\alpha}}}(\underline{\hat{\theta}}%
,\underline{\hat{z}})\right)  \left(  F(\underline{\check{\theta}})\tilde
{h}(\underline{\check{\theta}},\underline{\check{z}})p^{J}(\underline
{\check{\theta}},\underline{\check{z}})\tilde{\Phi}_{\underline{\check{\alpha
}}}(\underline{\check{\theta}},\underline{\check{z}})\right)
\]
and (up to a constant factor)%
\begin{equation}
F_{\underline{\alpha}}^{J}(\underline{\theta}_{W})_{A}\rightarrow\frac{1}%
{W}\left(  F^{J}(\underline{\hat{\theta}})M_{1}^{2}\right)  _{\underline
{\hat{\alpha}}}\left(  F_{\underline{\check{\alpha}}}^{J}(\underline
{\check{\theta}})\right)  \label{FA}%
\end{equation}
where (\ref{Kz}) and the definition (\ref{K}) for $\mathcal{O}=J$ have been used.

\medskip

\noindent B) From the \textbf{Bethe state: }Again because of Lemma \ref{l1} we
may take in (\ref{PHI1a}) only the first term and write with $\tilde{\Phi
}_{\underline{\hat{\alpha}}}^{D_{j}}(\underline{\hat{\theta}},\underline
{\hat{z}})=\left(  \Omega{C}({\underline{\hat{\theta}}},\hat{z}_{k})\dots
{D}({\underline{\hat{\theta}}},\hat{z}_{j})\dots{C}({\underline{\hat{\theta}}%
},\hat{z}_{1})\right)  _{\underline{\hat{\alpha}}}$%
\[
\tilde{\Phi}_{\underline{\alpha}1}(\underline{\hat{\theta}},\underline
{\check{\theta}},\underline{\hat{z}},\underline{\check{z}})\rightarrow
{\textstyle\sum\limits_{j}} \left(  \tilde{\Phi}_{\underline{\hat{\alpha}}%
}^{D_{j}}(\underline{\hat{\theta}},\underline{\hat{z}})\right)  \left(
\tilde{\Phi}(\underline{\check{\theta}},\underline{\check{z}})M_{1}%
^{2}\right)  _{\underline{\check{\alpha}}}%
\]
and we get%
\[
\sum_{j}\left(  F(\underline{\hat{\theta}})\tilde{h}(\underline{\hat{\theta}%
},\underline{\hat{z}})\tilde{\Phi}^{D_{j}}(\underline{\hat{\theta}}%
,\underline{\hat{z}})\right)  _{\underline{\hat{\alpha}}}\left(
F(\underline{\check{\theta}})\tilde{h}(\underline{\check{\theta}}%
,\underline{\check{z}})p^{J}(\underline{\check{\theta}},\underline{\check{z}%
})\tilde{\Phi}(\underline{\check{\theta}},\underline{\check{z}})M_{1}%
^{2}\right)  _{\underline{\check{\alpha}}}%
\]
and (up to a constant factor)
\begin{equation}
F_{\underline{\alpha}}^{J}(\underline{\theta}_{W})_{B}\rightarrow\frac{1}%
{W}\left(  F^{J}(\underline{\hat{\theta}})\right)  _{\underline{\hat{\alpha}}%
}\left(  F^{J}(\underline{\hat{\theta}})M_{1}^{2}\right)  _{\underline
{\check{\alpha}}}. \label{FB}%
\end{equation}
where (\ref{Kbb}) has been used. The final result is
\[
F_{\underline{\alpha}}^{J}(\underline{\theta}_{W})\rightarrow i\eta\frac{1}%
{W}\left(  \left(  F^{J}(\underline{\hat{\theta}})M_{1}^{2}\right)
_{\underline{\hat{\alpha}}}\left(  F_{\underline{\check{\alpha}}}%
^{J}(\underline{\hat{\theta}})\right)  -\left(  F_{\underline{\hat{\alpha}}%
}^{J}(\underline{\hat{\theta}})\right)  \left(  F^{J}(\underline{\check
{\theta}})M_{1}^{2}\right)  _{\underline{\check{\alpha}}}\right)  \,
\]
which is for $SU(2)$ the component $\left(  \beta,(\sigma)\right)  =(1,1)$ of
(\ref{FJ}) because $F_{\underline{\alpha}}^{J}(\underline{\theta
})=F_{\underline{\alpha}}^{J^{11}}(\underline{\theta})$ and $\left(
F^{J}(\underline{\theta})M_{1}^{2}\right)  _{\underline{\alpha}}%
=F_{\underline{\alpha}}^{J^{12}}(\underline{\theta})+F_{\underline{\alpha}%
}^{J^{21}}(\underline{\theta})$. The other components are obtained by
$SU(2)$-transformations. The constant factor is calculated below and the minus
sign is due to $SU(2)$ invariance. In terms of the components $J_{a}$
(\ref{JJ}) this can be written as in (\ref{FJ}) (see (\ref{Equi})).
\end{proof}

\subparagraph{Calculation of the functions $c_{JJ}^{J}(k,l,W):$}

defined by%
\[
F_{\underline{\alpha}}^{J^{\beta(\sigma)}}(\underline{\theta}_{W}%
)\overset{W\rightarrow\infty}{\rightarrow}c_{JJ}^{J}(k,l,W)\left(
\mathbf{C}_{\gamma(\lambda)}F_{\underline{\hat{\alpha}}}^{J^{\beta(\lambda)}%
}(\underline{\hat{\theta}})F_{\underline{\check{\alpha}}}^{J^{\gamma(\sigma)}%
}(\underline{\check{\theta}})-\mathbf{C}_{\gamma(\lambda)}F_{\underline
{\hat{\alpha}}}^{J^{\gamma(\sigma)}}(\underline{\hat{\theta}})F_{\underline
{\check{\alpha}}}^{J^{\beta(\lambda)}}(\underline{\check{\theta}})\right)
\]
for general $N$. Here and in the following we use the short notation%
\begin{equation}
S_{\alpha\underline{\alpha}}^{\underline{\alpha}^{\prime}\alpha^{\prime}%
}(\theta,\underline{\theta})=S_{\gamma_{n}\alpha_{n}}^{\alpha_{n}^{\prime
}\alpha^{\prime}}(\theta-\theta_{n})\dots S_{\alpha\alpha_{1}}^{\alpha
_{1}^{\prime}\gamma_{2}}(\theta-\theta_{1})\,. \label{sn}%
\end{equation}
We also use the satistics factor $\dot{\sigma}_{\alpha}^{\mathcal{O}}$, which
is related to the \textquotedblleft physical\textquotedblright\ statistics
by\footnote{See eqs. (27) and (28) in \cite{BFK1}.}
\begin{equation}
\dot{\sigma}_{\alpha}^{\mathcal{O}}=\sigma_{\alpha}^{\mathcal{O}%
}(-1)^{(N-1)+(1-1/N)(n-Q^{\mathcal{O}})} \label{st}%
\end{equation}
where $Q^{\mathcal{O}}$ is the charge of $\mathcal{O}$.

We apply the general procedure of Appendix \ref{ac}: Using $a(W)\rightarrow
e^{-i\pi\left(  1-\frac{1}{N}\right)  }$ of (\ref{Sa}) and $\sigma_{1}%
^{J}=1,~Q^{J}=0$ we check (\ref{tocheck}) and (\ref{tocheck1}) for this case%
\begin{gather*}
\dot{\sigma}_{1}^{J}(n)S_{1\underline{\check{\alpha}}}^{\underline
{\check{\alpha}}^{\prime}1}(\theta+W,\underline{\check{\theta}})\overset
{W\rightarrow\infty}{\rightarrow}(-1)^{(N-1)+(1-1/N)(k+l)}(a(W))^{l}%
1_{\underline{\check{\alpha}}}^{\underline{\check{\alpha}}^{\prime}%
}\rightarrow\dot{\sigma}_{1}^{J}(k)1_{\underline{\check{\alpha}}}%
^{\underline{\check{\alpha}}^{\prime}}\,\\
\dot{\sigma}_{1}^{J}(n)S_{\underline{\hat{\alpha}}\bar{1}}^{\bar{1}\hat
{\alpha}^{\prime}}(\underline{\hat{\theta}}+W,\omega)\overset{W\rightarrow
\infty}{\rightarrow}(-1)^{(N-1)+(1-1/N)(k+l)}(a(W))^{k(N-1)}1_{\underline
{\hat{\alpha}}}^{\underline{\hat{\alpha}}^{\prime}}\rightarrow(-1)^{\left(
N-1\right)  k}\dot{\sigma}_{1}^{J}(l)1_{\underline{\hat{\alpha}}}%
^{\underline{\hat{\alpha}}^{\prime}}%
\end{gather*}
Therefore, as proofed in Appendix \ref{ac}, $c_{JJ}^{J}(k,l,W)$ is independent
of $k$ and $l$, because $(-1)^{\left(  N-1\right)  k}=1$ for
$k=0\operatorname{mod}N$. It is convenient to consider the special case
$c_{JJ}^{J}(N,N,W)$:

1) We take the bound states $\bar{1}=(\hat{\alpha}_{2}\dots\hat{\alpha}_{N})$
and $\bar{N}=(\check{\alpha}_{1}\dots\check{\alpha}_{N-1})$ and calculate for
$\underline{\theta}_{W}=(\hat{\theta}+W,\hat{\omega}+W,\check{\omega}%
,\check{\theta})$%
\begin{multline}
\operatorname*{Res}_{\hat{\theta}=i\pi+\hat{\omega}}F_{1\bar{1}\bar{N}%
1}^{J^{1\bar{N}}}(\underline{\theta}_{W})=2i\,\mathbf{C}_{1\bar{1}}%
\,F_{\bar{N}1}^{J^{1\bar{N}}}(\check{\omega},\check{\theta})\left(
1-\dot{\sigma}_{\bar{1}}^{J}(2N)S_{\bar{1},\bar{N}1}^{\bar{N}1,\bar{1}}%
(\hat{\omega}+W;\check{\omega},\check{\theta})\right) \label{3J}\\
\overset{W\rightarrow\infty}{\rightarrow}-2i\,\mathbf{C}_{1\bar{1}}%
\,i\eta\frac{1}{W}F_{\bar{N}1}^{J^{1\bar{N}}}(\check{\omega},\check{\theta}).
\end{multline}
It was used that (\ref{S}), (\ref{Sb}), (\ref{Sa}) including $1/W$ terms and
$a(\theta)a(-\theta)=1$ imply%
\begin{multline*}
\dot{\sigma}_{1}^{J}(2N)S_{\bar{1},\bar{N}1}^{\bar{N}1,\bar{1}}(\hat{\omega
}+W;\check{\omega},\check{\theta})\\
\overset{W\rightarrow\infty}{\rightarrow}(-1)^{(N-1)+(1-1/N)2N}\left(
a(W)\tilde{b}(W)\right)  \left(  (-1)^{N-1}a(-W)\right)  \rightarrow
1+i\eta\frac{1}{W}.
\end{multline*}

2) Taking first $W\rightarrow\infty$ and then the $\operatorname*{Res}$ means%
\begin{multline}
\operatorname*{Res}_{\hat{\theta}=i\pi+\hat{\omega}}\left(  F_{1\bar{1}\bar
{N}1}^{J^{1\bar{N}}}(\underline{\theta}_{W})\overset{W\rightarrow\infty
}{\rightarrow}c_{JJ}^{J}(N,N,W)\mathbf{C}_{\gamma(\lambda)}\left(  F_{1\bar
{1}}^{J^{1(\lambda)}}(\underline{\hat{\theta}})F_{\bar{N}1}^{J^{\gamma\bar{N}%
}}(\underline{\check{\theta}})-F_{1\bar{1}}^{J^{\gamma\bar{N}}}(\underline
{\hat{\theta}})F_{\bar{N}1}^{J^{1(\lambda)}}(\underline{\check{\theta}%
})\right)  \right) \\
=c_{JJ}^{J}(N,N,W)\left(  -2i\right)  \mathbf{C}_{1\bar{1}}F_{\bar{N}%
1}^{J^{1\bar{N}}}(\check{\omega},\check{\theta}). \label{4J}%
\end{multline}
where (\ref{FJ2}) was used. As result we obtain from (\ref{3J}) and
(\ref{4J})
\[
c_{JJ}^{J}(k,l,W)=i\eta\frac{1}{W}.
\]

\begin{remark}
\label{rem2}Note that this also proves consistency of the clustering formula
(\ref{FJ}) for general $N$ with the form factor equation (iii).
\end{remark}

\subparagraph{Equivalence:}

We prove that%
\begin{equation}
F_{\underline{\alpha}}^{J_{a}}(\underline{\theta}_{W})\rightarrow-2\eta
\frac{1}{W}f_{abc}F_{\underline{\hat{\alpha}}}^{J_{b}}(\underline{\hat{\theta
}})F_{\underline{\check{\alpha}}}^{J_{c}}(\underline{\check{\theta}})
\label{FJabc}%
\end{equation}
is equivalent to
\[
F_{\underline{\alpha}}^{J^{\alpha(\rho)}}(\underline{\theta}_{W}%
)\rightarrow\frac{1}{W}i\eta\mathbf{C}_{\gamma(\sigma)}\left(  F_{\underline
{\hat{\alpha}}}^{J^{\alpha(\sigma)}}(\underline{\hat{\theta}})F_{\underline
{\check{\alpha}}}^{J^{\gamma(\rho)}}(\underline{\check{\theta}})-F_{\underline
{\hat{\alpha}}}^{J^{\gamma(\rho)}}(\underline{\hat{\theta}})F_{\underline
{\check{\alpha}}}^{J^{\alpha(\sigma)}}(\underline{\check{\theta}})\right)
\,.
\]
We have the general relations \cite{Ha,Wiki}%
\begin{equation}
\left[  T_{a},T_{b}\right]  _{\alpha}^{\beta}=if_{abc}\left(  T_{c}\right)
_{\alpha}^{\beta}~,~~\left(  T_{b}\right)  _{\alpha}^{\beta}\left(
T_{b}\right)  _{\gamma}^{\delta}=\frac{1}{2}\left(  \delta_{\alpha}^{\delta
}\delta_{\gamma}^{\beta}-\frac{1}{N}\delta_{\alpha}^{\beta}\delta_{\gamma
}^{\delta}\right)  \,. \label{TT}%
\end{equation}
By (\ref{JJ}) and (\ref{FJabc}) we obtain for $W\rightarrow\infty$%
\begin{align}
F_{\underline{\alpha}}^{J_{a}}(\underline{\theta}_{W})  &  =\mathbf{C}%
_{\gamma(\rho)}\left(  T_{a}\right)  _{\alpha}^{\gamma}F_{\underline{\alpha}%
}^{J^{\alpha(\rho)}}(\underline{\theta}_{W})\label{Equi}\\
&  \rightarrow-2\eta\frac{1}{W}f_{abc}F_{\underline{\hat{\alpha}}}^{J_{b}%
}(\underline{\hat{\theta}})F_{\underline{\check{\alpha}}}^{J_{c}}%
(\underline{\check{\theta}})\nonumber\\
&  =i\eta\frac{1}{W}\mathbf{C}_{\gamma(\rho)}\left(  T_{a}\right)  _{\alpha
}^{\gamma}\left(  \mathbf{C}_{\gamma^{\prime}(\rho)^{\prime}}\left(
F_{\underline{\hat{\alpha}}}^{J^{\alpha(\rho)^{\prime}}}\right)  \left(
F_{\underline{\check{\alpha}}}^{J^{\gamma^{\prime}(\rho)}}\right)
-\mathbf{C}_{\gamma^{\prime}(\rho)^{\prime}}\left(  F_{\underline{\hat{\alpha
}}}^{J^{\gamma^{\prime}(\rho)}}\right)  \left(  F_{\underline{\check{\alpha}}%
}^{J^{\alpha(\rho)^{\prime}}}\right)  \right) \nonumber
\end{align}
where the relations (\ref{TT}) have been used. This proves the equivalency.

\subsection{Theorem \ref{t1a}}

\begin{theorem}
\label{t1a}The form factor of the field $\phi(x)$ for particle number
$n=0\operatorname{mod}N$ and $k=0\operatorname{mod}N$ shows the cluster
behavior
\[
F_{\underline{\alpha}}^{\phi}(\underline{\theta}_{W})\overset{W\rightarrow
\infty}{\rightarrow}F_{\underline{\hat{\alpha}}}^{\phi}(\underline{\hat
{\theta}})F_{\underline{\check{\alpha}}}^{\phi}(\underline{\check{\theta}%
})\,.
\]

\end{theorem}

\begin{remark}
Note that this is the typical behavior of an exponential of a bosonic field
(see \cite{BK5}).
\end{remark}

\begin{proof}
We investigate%
\[
F_{\underline{\alpha}}^{\phi}(\underline{\theta}_{W})=N_{n}^{\phi}%
F(\underline{\theta}_{W})\int d\underline{\underline{z}}\tilde{h}%
\,(\underline{\theta}_{W},\underline{\underline{z}}_{W})p^{\phi}%
(\underline{\theta}_{W},\underline{z}_{W})\tilde{\Phi}_{\underline{\alpha}%
}(\underline{\theta}_{W},\underline{\underline{z}}_{W})
\]
From the asymptotic behavior of $F(\underline{\theta}_{W})\tilde
{h}\,(\underline{\theta}_{W},\underline{\underline{z}}_{W})$ and $p^{\phi
}(\underline{\theta}_{W},\underline{z}_{W})$ in (\ref{Fha}), (\ref{pphia}) and
(\ref{expsc}) we derive for $W\rightarrow\infty$ the exponential behavior%
\begin{equation}
F(\underline{\theta}_{W})\tilde{h}\,(\underline{\theta}_{W},\underline
{\underline{z}}_{W})p^{\phi}(\underline{\theta}_{W},\underline{\underline{z}%
}_{W})\varpropto\left(  e^{-\frac{1}{2}W}\right)  ^{\tilde{k}_{1}^{2}%
+\tilde{k}_{N-1}^{2}+\sum_{j=1}^{N-2}\left(  \tilde{k}_{j}-\tilde{k}%
_{j+1}\right)  ^{2}-\tilde{k}_{1}} \label{exphi}%
\end{equation}
where $\tilde{k}_{j}=k_{j}-k\left(  1-j/N\right)  $. For
$k=0\operatorname{mod}N$ the leading behavior $\left(  e^{-\frac{1}{2}%
W}\right)  ^{0}$ is obtained for $\tilde{k}_{j}=0\Rightarrow$%
\[
k_{j}=k\left(  1-j/N\right)  ,~l_{j}=l\left(  1-j/N\right)  ,~j=1,\dots,N-1
\]
For these values of $k_{j}$ and $l_{j}$ we obtain, more precisely, with
(\ref{Fha}), (\ref{pphia}) and (\ref{PHI0}) in leading order the asymptotic
behavior (up to a constant factor)%
\begin{multline*}
F(\underline{\theta}_{W})p^{\phi}(\underline{\theta}_{W},\underline
{\underline{z}}_{W})\tilde{\Phi}_{\underline{\alpha}}(\underline{\theta}%
_{W},\underline{\underline{z}}_{W})\\
\rightarrow\left(  F(\underline{\hat{\theta}})\tilde{h}(\underline{\hat
{\theta}},\underline{\hat{z}})p^{\phi}(\underline{\hat{\theta}},\underline
{\hat{z}})\tilde{\Phi}_{\underline{\hat{\alpha}}}(\underline{\hat{\theta}%
},\underline{\underline{\hat{z}}})\right)  \left(  F(\underline{\check{\theta
}})\tilde{h}(\underline{\check{\theta}},\underline{\check{z}})p^{\phi
}(\underline{\check{\theta}},\underline{\check{z}})\tilde{\Phi}_{\underline
{\check{\alpha}}}(\underline{\check{\theta}},\underline{\underline{\check{z}}%
})\right)
\end{multline*}
such that%
\[
F_{\underline{\alpha}}^{\phi}(\underline{\theta}_{W})\rightarrow
F_{\underline{\hat{\alpha}}}^{\phi}(\underline{\hat{\theta}})F_{\underline
{\check{\alpha}}}^{\phi}(\underline{\check{\theta}})
\]
The constant factor is again calculated using the form factor equation (iii).
\end{proof}

\subparagraph{Calculation of the function $c_{\phi\phi}^{\phi}(k,l,W):$}

defined by%
\[
F_{\underline{\alpha}}^{\phi}(\underline{\theta}_{W})\rightarrow c_{\phi\phi
}^{\phi}(k,l,W)F_{\underline{\hat{\alpha}}}^{\phi}(\underline{\hat{\theta}%
})F_{\underline{\check{\alpha}}}^{\phi}(\underline{\check{\theta}})\,.
\]
We apply the general procedure of Appendix \ref{ac}: Using $a(W)\rightarrow
e^{-i\pi\left(  1-\frac{1}{N}\right)  }$ of (\ref{Sa}) and $\sigma^{\phi
}=e^{i\eta},~Q^{\phi}=0$ we check (\ref{tocheck}) and (\ref{tocheck1}) for
this case%
\begin{align*}
\dot{\sigma}_{1}^{\phi}(n)S_{1\underline{\check{\alpha}}}^{\underline
{\check{\alpha}}^{\prime}1}(\theta+W,\underline{\check{\theta}})  &
=e^{i\eta}(-1)^{(N-1)+(1-1/N)(k+l)}(a(W))^{l}1_{\underline{\check{\alpha}}%
}^{\underline{\check{\alpha}}^{\prime}}\rightarrow\dot{\sigma}_{1}^{\phi
}(k)1_{\underline{\check{\alpha}}}^{\underline{\check{\alpha}}^{\prime}}\,\\
\dot{\sigma}_{1}^{\phi}(n)S_{\underline{\hat{\alpha}}\bar{1}}^{\bar{1}%
\hat{\alpha}^{\prime}}(\underline{\hat{\theta}}+W,\omega)  &  =e^{i\eta
}(-1)^{(N-1)+(1-1/N)(k+l)}(a(W))^{(N-1)k}1_{\underline{\hat{\alpha}}%
}^{\underline{\hat{\alpha}}^{\prime}}\rightarrow\dot{\sigma}_{1}^{\phi
}(l)1_{\underline{\hat{\alpha}}}^{\underline{\hat{\alpha}}^{\prime}}%
\end{align*}
Therefore, as proofed in Appendix \ref{ac}, $c_{\phi\phi}^{\phi}(k,l,W)$ is
independent of $k$ and $l$, because $(-1)^{\left(  N-1\right)  k}=1$ for
$k=0\operatorname{mod}N$. The special case $c_{\phi\phi}^{\phi}(k,0,W)$ is
obtained by the form factor equation (v) with $s^{\phi}=0$ and (\ref{Fphi})
for $\underline{\check{\alpha}}=\emptyset$%
\begin{align*}
F_{\underline{\hat{\alpha}}\emptyset}^{\phi}(\underline{\theta}_{W})  &
\rightarrow e^{Ws^{\phi}}F_{\underline{\hat{\alpha}}\emptyset}^{\phi
}(\underline{\theta})=F_{\underline{\hat{\alpha}}}^{\phi}(\underline{\theta
})\\
F_{\underline{\hat{\alpha}}\emptyset}^{\phi}(\underline{\theta}_{W})  &
\rightarrow c_{\phi\phi}^{\phi}(k,0,W)F_{\underline{\hat{\alpha}}}^{\phi
}(\underline{\theta})F_{\emptyset}^{\phi}%
\end{align*}
which implies
\[
c_{\phi\phi}^{\phi}(k,l,W)=1\,
\]
if we normalize the field $\phi(x)$ by $F_{\emptyset}^{\phi}=\left\langle
0|\phi(x)|0\right\rangle =1$.

\subsection{Theorem \ref{t1b}}

\begin{theorem}
\label{t1b}The form factor of the energy momentum potential for particle
number $n=0\operatorname{mod}N$ and $k=0\operatorname{mod}N$ satisfies%
\begin{equation}
F_{\underline{\alpha}}^{T}(\underline{\theta}_{W})=O(W^{-2})~\text{for
}W\rightarrow\infty\,. \label{FTN0}%
\end{equation}

\end{theorem}

More precisely

\begin{conjecture}
\label{cj1}The cluster behavior of form factor of $T$ for
$k=0\operatorname{mod}N$ reads as%
\begin{equation}
F_{\underline{\alpha}}^{T}(\underline{\theta}_{W})\overset{W\rightarrow\infty
}{\rightarrow}2\eta W^{-2}F_{\underline{\hat{\alpha}}}^{J_{a}}(\underline
{\hat{\theta}})F_{\underline{\check{\alpha}}}^{J_{a}}(\underline{\check
{\theta}})=\eta W^{-2}\mathbf{C}_{\alpha\bar{\delta}}\mathbf{C}_{\beta
\bar{\gamma}}F_{\underline{\hat{\alpha}}}^{J^{\alpha\bar{\gamma}}}%
(\underline{\hat{\theta}})F_{\underline{\check{\alpha}}}^{J^{\beta\bar{\delta
}}}(\underline{\check{\theta}}) \label{FTN}%
\end{equation}

\end{conjecture}

We have no general proof of this conjecture. The problem is that the expansion
for large $W$ of the integrand in the contour integral representation in
(\ref{K}) must not be interchanged with the integration, this is only allowed
up to the $1/W$-term.

However, we have checked consistency with the form factor equation (iii),
which also yields the function $c_{JJ}^{T}(k,l,W)=\eta W^{-2}$.

\begin{proof}
To prove (\ref{FTN0}) we investigate for $n=k+l,~m=k_{1}+l_{1},~m=n/2$%
\[
F_{\underline{\alpha}}^{T}(\underline{\theta}_{W})=N_{n}^{T}\int
d\underline{z}F(\underline{\theta}_{W})\tilde{h}(\underline{\theta}%
_{W},\underline{z}_{W})p^{T}(\underline{\theta}_{W},\underline{z}_{W}%
)\Psi_{\underline{\alpha}}(\underline{\theta}_{W},\underline{z}_{W})\,.
\]
From the asymptotic behavior of $F(\underline{\theta}_{W})\tilde
{h}\,(\underline{\theta}_{W},\underline{\underline{z}}_{W})$ and
$p^{T}(\underline{\theta}_{W},\underline{z}_{W})$ in (\ref{Fha}), (\ref{pTa})
and (\ref{expsc}) we derive for $W\rightarrow\infty$ the exponential behavior%
\begin{equation}
F(\underline{\theta}_{W})\tilde{h}\,(\underline{\theta}_{W},\underline
{\underline{z}}_{W})p^{T}(\underline{\theta}_{W},\underline{\underline{z}}%
_{W})\varpropto\left(  e^{-\frac{1}{2}W}\right)  ^{\tilde{k}_{1}^{2}+\tilde
{k}_{N-1}^{2}+\sum_{j=1}^{N-2}\left(  \tilde{k}_{j}-\tilde{k}_{j+1}\right)
^{2}} \label{expT}%
\end{equation}
where $\tilde{k}_{j}=k_{j}-k\left(  1-j/N\right)  $. For
$k=0\operatorname{mod}N$ the leading behavior $\left(  e^{-\frac{1}{2}%
W}\right)  ^{0}$ is obtained for $\tilde{k}_{j}=0$. Therefore by (\ref{nT})
\[
k_{j}=k\left(  1-j/N\right)  ,~l_{j}=l\left(  1-j/N\right)  ,~j=1,\dots,N-1.
\]
For these values of $k_{j}$ and $l_{j}$ we obtain, more precisely, with
(\ref{Fha}), (\ref{pTa}) and (\ref{PHI0}) in leading order the asymptotic
behavior (up to a constant factor)%
\begin{multline*}
F(\underline{\theta}_{W})\tilde{h}\,(\underline{\theta}_{W},\underline
{\underline{z}}_{W})p^{T}(\underline{\theta}_{W},\underline{\underline{z}}%
_{W})\tilde{\Phi}_{\underline{\alpha}}(\underline{\theta}_{W},\underline
{\underline{z}}_{W})\\
\overset{W\rightarrow\infty}{\rightarrow}\left(  F(\underline{\hat{\theta}%
})\tilde{h}\,(\underline{\hat{\theta}},\underline{\underline{\hat{z}}}%
)\tilde{\Phi}_{\underline{\alpha}}(\underline{\hat{\theta}},\underline
{\underline{\hat{z}}})\right)  \left(  F(\underline{\check{\theta}})\tilde
{h}\,(\underline{\check{\theta}},\underline{\underline{\check{z}}})\tilde
{\Psi}_{\underline{\alpha}}(\underline{\check{\theta}},\underline
{\underline{\check{z}}})\right)  \left(  p^{T_{+}}(\underline{\hat{\theta}%
},\underline{\hat{z}})+p^{T-}(\underline{\check{\theta}},\underline{\check{z}%
})\right)
\end{multline*}
However, this means that in leading order
\[
F_{\underline{\alpha}}^{T}(\underline{\theta}_{W})\rightarrow0
\]
because of Lemma \ref{l1}.\medskip

\noindent\textbf{Order }$\frac{1}{W}$: similarly, as in the proof of Theorem
\ref{t1} we discuss the contribution from $\tilde{h}_{1}(\underline{\theta
},\underline{z})$ of (\ref{h1}), however, for $k_{1}=k\left(  1-1/N\right)
,~l_{1}=l\left(  1-1/N\right)  $ there are no the $\hat{z}_{j}$-dependent
terms and and therefore this contribution vanishes by Lemma \ref{l1}.

From the Bethe state $\tilde{\Psi}_{\underline{\alpha}1}(\underline
{\hat{\theta}},\underline{\check{\theta}},\underline{\hat{z}},\underline
{\check{z}})$ of (\ref{PHI1a}) and (\ref{pTa}) we obtain contributions of the
type%
\[
\int d\underline{z}\,\tilde{h}(\underline{\theta},\underline{z})\left(
\tilde{\Psi}(\underline{\theta},\underline{z})M_{1}^{2}\right)  _{\underline
{\alpha}}~\text{and }\int d\underline{z}\,\tilde{h}(\underline{\theta
},\underline{z})p^{T_{\pm}}(\underline{\theta},\underline{z})\left(
\tilde{\Psi}(\underline{\theta},\underline{z})M_{1}^{2}\right)  _{\underline
{\alpha}}%
\]
where both are $=0$, the first one because of Lemma \ref{l1} and the second
one because $T(x)$ is an iso-scalar operator. Therefore there are no
contributions of order $W^{-1}$ and
\[
F_{\underline{\alpha}}^{T}(\underline{\theta}_{W})\rightarrow O(W^{-2}%
)~\text{for }W\rightarrow\infty\,.
\]

\end{proof}

\subparagraph{Calculation of the functions $c_{JJ}^{T}(k,l,W):$}

defined by%
\[
F_{\underline{\alpha}}^{T}(\underline{\theta}_{W})\overset{W\rightarrow\infty
}{\rightarrow}c_{JJ}^{T}(k,l,W)\mathbf{C}_{\alpha\bar{\delta}}\mathbf{C}%
_{\beta\bar{\gamma}}F_{\underline{\hat{\alpha}}}^{J^{\alpha\bar{\gamma}}%
}F_{\underline{\check{\alpha}}}^{J^{\beta\bar{\delta}}}(\underline
{\check{\theta}})\,.
\]
We apply the general procedure of Appendix \ref{ac}: Using $a(W)\rightarrow
e^{-i\pi\left(  1-\frac{1}{N}\right)  }$ of (\ref{Sa}) and $\sigma_{1}%
^{T}=1,~Q^{T}=0$ we check (\ref{tocheck}) and (\ref{tocheck1}) for this case%
\begin{gather*}
\dot{\sigma}_{1}^{T}(n)S_{1\underline{\check{\alpha}}}^{\underline
{\check{\alpha}}^{\prime}1}(\theta,\underline{\check{\theta}})\overset
{W\rightarrow\infty}{\rightarrow}(-1)^{(N-1)+(1-1/N)(k+l)}(a(W))^{l}%
1_{\underline{\check{\alpha}}}^{\underline{\check{\alpha}}^{\prime}%
}\rightarrow\dot{\sigma}_{1}^{T}(k)1_{\underline{\check{\alpha}}}%
^{\underline{\check{\alpha}}^{\prime}}\,\\
\dot{\sigma}_{1}^{T}(n)S_{\underline{\hat{\alpha}}\bar{1}}^{\bar{1}\hat
{\alpha}^{\prime}}(\underline{\hat{\theta}}+W,\omega)\overset{W\rightarrow
\infty}{\rightarrow}(-1)^{(N-1)+(1-1/N)(k+l)}(a(W))^{k(N-1)}1_{\underline
{\hat{\alpha}}}^{\underline{\hat{\alpha}}^{\prime}}\rightarrow(-1)^{\left(
N-1\right)  k}\dot{\sigma}_{1}^{J}(l)1_{\underline{\hat{\alpha}}}%
^{\underline{\hat{\alpha}}^{\prime}}%
\end{gather*}
Therefore, as proofed in Appendix \ref{ac}, $c_{JJ}^{T}(k,l,W)$ is independent
of $k$ and $l$, because $(-1)^{\left(  N-1\right)  k}=1$ for
$k=0\operatorname{mod}N$. It is convenient to consider the special case
$c_{JJ}^{T}(N,N,W)$:

1) We take the bound states $\bar{1}=(\hat{\alpha}_{2}\dots\hat{\alpha}_{N})$
and $\bar{1}=(\check{\alpha}_{1}\dots\check{\alpha}_{N-1})$ and calculate for
$\underline{\theta}_{W}=(\hat{\theta}+W,\hat{\omega}+W,\check{\omega}%
,\check{\theta})$%
\begin{multline*}
\operatorname*{Res}_{\hat{\theta}=i\pi+\hat{\omega}}F_{1\bar{1}\bar{1}1}%
^{T}(\underline{\theta}_{W})=2i\,\mathbf{C}_{1\bar{1}}\,F_{\bar{1}1}%
^{T}(\check{\omega},\check{\theta})\left(  1-\dot{\sigma}_{\bar{1}}%
^{T}(2N)S_{\bar{1},\bar{1}1}^{\bar{1}1,\bar{1}}(\hat{\omega}+W;\check{\omega
},\check{\theta})\right)  \\
\overset{W\rightarrow\infty}{\rightarrow}-2i\,W^{-2}i\eta\left(  1-1/N\right)
\left(  \check{\theta}-\hat{\omega}+i\pi\right)  F_{\bar{1}1}^{T}%
(\check{\omega},\check{\theta})
\end{multline*}
It was used that $\mathbf{C}_{1\bar{1}}=1$ and that (\ref{S}), (\ref{Sb}) and
(\ref{Sa}) imply%
\begin{multline*}
\dot{\sigma}_{\bar{1}}^{T}(2N)S_{\bar{1},1\bar{1}}^{\bar{1}1,\bar{1}}%
(\hat{\omega}+W;\check{\omega},\check{\theta})\\
=(-1)^{(N-1)+(1-1/N)2N}\left(  a(\hat{\omega}+W-\check{\omega})\right)
\left(  (-1)^{N-1}a(i\pi-\left(  \hat{\omega}+W-\check{\theta}\right)
)\right)  \\
\overset{W\rightarrow\infty}{\rightarrow}\exp\left(  -i\eta\left(
1-1/N\right)  \left(  \left(  \hat{\omega}+W-\check{\omega}\right)
^{-1}+\left(  i\pi-\left(  \hat{\omega}+W-\check{\theta}\right)  \right)
^{-1}\right)  \right)  \\
\overset{W\rightarrow\infty}{\rightarrow}1+W^{-2}i\eta\left(  1-1/N\right)
\left(  \check{\theta}-\hat{\omega}+i\pi\right)  .
\end{multline*}

2) Taking first $W\rightarrow\infty$ and then the $\operatorname*{Res}$ means%
\begin{multline*}
\operatorname*{Res}_{\hat{\theta}=i\pi+\hat{\omega}}\left(  F_{1\bar{1}\bar
{1}1}^{T}(\underline{\theta}_{W})\overset{W\rightarrow\infty}{\rightarrow
}c_{JJ}^{T}(N,N,W)\mathbf{C}_{\alpha\bar{\delta}}\mathbf{C}_{\beta\bar{\gamma
}}F_{1\bar{1}}^{J^{\alpha\bar{\gamma}}}(\hat{\theta},\hat{\omega})F_{\bar{1}%
1}^{J^{\beta\bar{\delta}}}(\check{\omega},\check{\theta})\right) \\
=-2ic_{JJ}^{T}(N,N,W)F_{\bar{1}1}^{J^{1\bar{1}}}(\check{\omega},\check{\theta
})
\end{multline*}
where (\ref{FJ2}) has been used. The particle anti-particle form factors
(\ref{FJ2}) and (\ref{FT2}) satisfy applying the form factor equation (ii)%
\begin{equation}
F_{\bar{1}1}^{J^{1\bar{1}}}(\theta,\omega)=\left(  1-1/N\right)  i\left(
\theta-\omega+i\pi\right)  F_{\bar{1}1}^{T}(\theta,\omega) \label{FJT}%
\end{equation}
therefore%
\[
c_{JJ}^{T}(k,l,W)=\eta W^{-2}%
\]
which supports (\ref{FT}).

\begin{remark}
Repeating the last discussion for the more general case $k=N,~l=LN,~\underline
{\theta}_{W}=(\hat{\omega}+W,\hat{\theta}+W,\underline{\check{\theta}%
},\underline{\check{\omega}})$ and $\check{\alpha}=\left(  \alpha\dots
\alpha\bar{\alpha}\dots\bar{\alpha}\right)  $ for a fixed $\alpha=1,\dots,N$
we obtain as an generalization of (\ref{FJT}) the interesting relation of
energy momentum and current form factors%
\begin{align}
\left(  1-\frac{1}{N}\right)  \sum_{j=1}^{L}i\left(  \check{\theta}_{j}%
-\check{\omega}_{j}-i\pi\right)  \,F_{\alpha\dots\alpha\bar{\alpha}\dots
\bar{\alpha}}^{T}(\underline{\check{\theta}},\underline{\check{\omega}})  &
=\mathbf{C}_{\alpha\bar{\alpha}}\left(  F_{\alpha\dots\alpha\bar{\alpha}%
\dots\bar{\alpha}}^{J^{\alpha\bar{\alpha}}}(\underline{\check{\theta}%
},\underline{\check{\omega}})\right) \label{FJTg}\\
&  =2\left(  T_{a}\right)  _{\alpha}^{\alpha}F_{\alpha\dots\alpha\bar{\alpha
}\dots\bar{\alpha}}^{J_{a}}(\underline{\check{\theta}},\underline
{\check{\omega}}).\nonumber
\end{align}

\end{remark}

\subparagraph{Equivalence:}

Using the general relations (\ref{TT}), (\ref{JJ}) and $\mathbf{C}_{\alpha
\bar{\gamma}}F_{\underline{\hat{\alpha}}}^{J^{\alpha\bar{\gamma}}}=0$ we
obtain%
\[
2F_{\underline{\hat{\alpha}}}^{J_{a}}(\underline{\hat{\theta}})F_{\underline
{\check{\alpha}}}^{J_{a}}(\underline{\check{\theta}})=2\left(  T_{a}\right)
_{\alpha\bar{\gamma}}F_{\underline{\hat{\alpha}}}^{J^{\alpha\bar{\gamma}}%
}(\underline{\hat{\theta}})\left(  T_{a}\right)  _{\beta\bar{\delta}%
}F_{\underline{\check{\alpha}}}^{J^{\beta\bar{\delta}}}(\underline
{\check{\theta}})=\mathbf{C}_{\beta\bar{\gamma}}\mathbf{C}_{\alpha\bar{\delta
}}\,F_{\underline{\hat{\alpha}}}^{J^{\alpha\bar{\gamma}}}(\underline
{\hat{\theta}})F_{\underline{\check{\alpha}}}^{J^{\beta\bar{\delta}}%
}(\underline{\check{\theta}}).
\]
with $\left(  T_{a}\right)  _{\alpha\bar{\beta}}=\mathbf{C}_{\delta\bar{\beta
}}\left(  T_{a}\right)  _{\alpha}^{\delta}$.

\subsection{Theorem \ref{t2}}

\begin{theorem}
\label{t2}The cluster behavior of the form factor of the pseudo-potential of
the current for particle number $n=0\operatorname{mod}N$ and
$k=1\operatorname{mod}N$ reads as%
\begin{align*}
&  F_{\underline{\alpha}}^{J}(\underline{\theta}_{W})\overset{W\rightarrow
\infty}{\rightarrow}c_{\psi\chi}^{J}(k,l,W)F_{\underline{\hat{\alpha}}}^{\psi
}(\underline{\hat{\theta}})F_{\underline{\check{\alpha}}}^{\chi}%
(\underline{\check{\theta}})\\
c_{\psi\chi}^{J}(k,l,W)  &  =e^{i\pi l_{1}}2\left(  2\pi\right)  ^{-\frac
{1+N}{N^{2}}}e^{-i\pi\left(  N+\frac{1}{2N}\right)  }/\bar{F}(i\pi
)\,W^{\frac{1}{N^{2}}}e^{-\frac{1}{2}\left(  1-\frac{1}{N}\right)  W}%
\end{align*}
with $l_{1}=\left(  l+1\right)  \left(  1-1/N\right)  -1$.
\end{theorem}

\begin{proof}
We investigate%
\[
F_{\underline{\alpha}}^{J}(\underline{\theta}_{W})=N_{n}^{J}F(\underline
{\theta}_{W})\int d\underline{\underline{z}}\tilde{h}\,(\underline{\theta}%
_{W},\underline{\underline{z}}_{W})p^{J}(\underline{\theta}_{W},\underline
{\underline{z}}_{W})\tilde{\Phi}_{\underline{\alpha}}(\underline{\theta}%
_{W},\underline{\underline{z}}_{W})
\]
The exponential behavior of the integrand is again given by (\ref{expJ}). For
$k=1\operatorname{mod}N$ the leading asymptotic behavior $\left(  e^{-\frac
{1}{2}W}\right)  ^{1-\frac{1}{N}}$ is obtained for $\tilde{k}_{j}=j/N-1$ which
implies%
\[
k_{j}=\left(  k-1\right)  \left(  1-j/N\right)  ,~~l_{j}=\left(  l+1\right)
\left(  1-j/N\right)  -1,~j=1,\dots,N-1
\]
For these values of $k_{j}$ and $l_{j}$ we obtain, more precisely%
\begin{align*}
&  F(\underline{\theta}_{W})\tilde{h}\,(\underline{\theta}_{W},\underline
{\underline{z}}_{W})p^{J}(\underline{\theta}_{W},\underline{\underline{z}}%
_{W})\tilde{\Phi}_{\underline{\alpha}}(\underline{\theta}_{W},\underline
{\underline{z}}_{W})\\
&  \rightarrow W^{\frac{1}{N^{2}}}e^{-\frac{1}{2}\left(  1-\frac{1}{N}\right)
W}\left(  F(\underline{\hat{\theta}})\tilde{h}(\underline{\hat{\theta}%
},\underline{\hat{z}})p^{\psi}(\underline{\hat{\theta}},\underline{\hat{z}%
})\tilde{\Psi}_{\underline{\hat{\alpha}}}(\underline{\hat{\theta}}%
,\underline{\hat{z}})\right)  \left(  F(\underline{\check{\theta}})\tilde
{h}(\underline{\check{\theta}},\underline{\check{z}})p^{\chi}(\underline
{\check{\theta}},\underline{\check{z}})\tilde{\Phi}_{\underline{\check{\alpha
}}}(\underline{\check{\theta}},\underline{\check{z}})\right)  ,
\end{align*}
which implies (up to const.)%
\[
F_{\underline{\alpha}}^{J}(\underline{\theta}_{W})\rightarrow W^{\frac
{1}{N^{2}}}e^{-\frac{1}{2}\left(  1-\frac{1}{N}\right)  W}F_{\underline
{\hat{\alpha}}}^{\psi}(\underline{\hat{\theta}})F_{\underline{\check{\alpha}}%
}^{\chi}(\underline{\check{\theta}}).
\]

\end{proof}

\subparagraph{Calculation of the function $c_{\psi\chi}^{J}(k,l,W)$}

defined by%
\[
F_{\underline{\alpha}}^{J}(\underline{\theta}_{W})\overset{W\rightarrow\infty
}{\rightarrow}c_{\psi\chi}^{J}(k,l,W)F_{\underline{\hat{\alpha}}}^{\psi
}(\underline{\hat{\theta}})F_{\underline{\check{\alpha}}}^{\chi}%
(\underline{\check{\theta}}).
\]
We apply the procedure of Appendix \ref{ac}: Using $a(W)\rightarrow
e^{-i\pi\left(  1-\frac{1}{N}\right)  }$ of (\ref{Sa}) we check (\ref{tocheck}%
) and (\ref{tocheck1}) with $\sigma^{J}=1,~Q^{J}=0,~\sigma^{\psi}%
=e^{i\pi\left(  1-\frac{1}{N}\right)  },~Q^{\psi}=1$ and $\sigma_{1}^{\chi
}=e^{i\pi\left(  N-\frac{1}{N}\right)  },~Q^{\chi}=N-1$%
\begin{gather*}
\dot{\sigma}_{1}^{J}(n)S_{1\underline{\check{\alpha}}}^{\underline
{\check{\alpha}}^{\prime}1}(\theta+W,\underline{\check{\theta}})\overset
{W\rightarrow\infty}{\rightarrow}(-1)^{(N-1)+(1-1/N)(k+l)}(a(W))^{l}%
1_{\underline{\check{\alpha}}}^{\underline{\check{\alpha}}^{\prime}%
}\rightarrow\dot{\sigma}_{1}^{\psi}(k)1_{\underline{\check{\alpha}}%
}^{\underline{\check{\alpha}}^{\prime}}\\
\dot{\sigma}_{1}^{J}(n)S_{\hat{\alpha}\bar{1}}^{\bar{1}\hat{\alpha}^{\prime}%
}(\underline{\hat{\theta}}+W,\omega)\overset{W\rightarrow\infty}{\rightarrow
}(-1)^{(N-1)+(1-1/N)(k+l)}(a(W))^{(N-1)k}1_{\underline{\hat{\alpha}}%
}^{\underline{\hat{\alpha}}^{\prime}}\rightarrow(-1)^{\left(  N-1\right)
k}\dot{\sigma}_{1}^{\chi}(l)1_{\hat{\alpha}}^{\hat{\alpha}^{\prime}}%
\end{gather*}
Therefore $c_{\psi\chi}^{J}(k,l,W)$ is independent of $k$ and for
$k=1\operatorname{mod}N$ (see \ref{ac})%
\[
c_{\psi\chi}^{J}(k,l,W)=c_{\psi\chi}^{J}(k_{0},l_{0},W)(-1)^{\left(
N-1\right)  \left(  l-l_{0}\right)  /N}%
\]
The special case $c_{\psi\chi}^{J}(1,N-1,W)$ is calculated by the following
example, which implies%
\begin{equation}
c_{\psi\chi}^{J}(k,l,W)=e^{i\pi l_{1}}2\left(  2\pi\right)  ^{-\frac
{1+N}{N^{2}}}e^{-i\pi\left(  N+\frac{1}{2N}\right)  }/\bar{F}(i\pi
)\,W^{\frac{1}{N^{2}}}e^{-\frac{1}{2}\left(  1-\frac{1}{N}\right)  W}
\label{cJ}%
\end{equation}
because $l_{1}=\left(  l+1\right)  \left(  1-1/N\right)  -1$.

\paragraph{Example:}

The particle anti-particle of (\ref{FJ2}) and asymptotic behavior of the
particle anti-particle minimal form factor function (\ref{Fba}) imply%
\[
F_{1\bar{N}}^{J}(\underline{\theta}_{W})\overset{W\rightarrow\infty
}{\rightarrow}2e^{-\frac{1}{2}(\theta-\omega+W)}\left(  \left(  2\pi\right)
^{-1-\frac{1}{N}}W^{\frac{1}{N}}e^{\frac{1}{2}W}e^{\frac{1}{2}\left(
\theta-\omega-i\pi\right)  }\right)  ^{\frac{1}{N}}/\bar{F}(i\pi).
\]
The asymptotic relation (\ref{FJpsichi}), (\ref{F1psi}) and (\ref{F1chi})
give
\[
F_{1\bar{N}}^{J}(\underline{\theta}_{W})\overset{W\rightarrow\infty
}{\rightarrow}c_{\psi\chi}^{J}(1,N-1,W)F_{1}^{\psi}(\theta)F_{\bar{N}}^{\chi
}(0)=c_{\psi\chi}^{J}(1,N-1,W)e^{-\frac{1}{2}\left(  1-\frac{1}{N}\right)
\theta}e^{\frac{1}{2}\left(  1-\frac{1}{N}\right)  \omega}%
\]
which means%
\[
c_{\psi\chi}^{J}(1,N-1,W)=2e^{-\frac{1}{2}W}\left(  \left(  2\pi\right)
^{-1-\frac{1}{N}}W^{\frac{1}{N}}e^{\frac{1}{2}W}e^{-\frac{1}{2}i\pi}\right)
^{\frac{1}{N}}/\bar{F}(i\pi)
\]
and (\ref{cJ}).

\subsection{Conjecture \ref{t2a}}

\begin{conjecture}
\label{t2a}The form factor of the energy momentum potential for particle
number $n=0\operatorname{mod}N$ and $k=1\operatorname{mod}N$ shows the cluster
behavior (\ref{FTpsichi})%
\[
F_{\underline{\alpha}}^{T}(\underline{\theta}_{W})\overset{W\rightarrow\infty
}{\rightarrow}c_{\psi\chi}^{T}(k,l,W)\mathbf{C}_{\alpha\bar{\beta}%
}F_{\underline{\hat{\alpha}}}^{\psi^{\alpha}}(\underline{\hat{\theta}%
})F_{\underline{\check{\alpha}}}^{\chi^{\bar{\beta}}}(\underline{\check
{\theta}}).
\]

\end{conjecture}

We have no general proof of this conjecture. The problem is the same as in
Conjecture \ref{cj1}, that the expansion for large $W$ of the integrand in the
multiple contour integral representation in (\ref{K}) must not be interchanged
with the integration. However, the relation (\ref{FJTg}) implies the cluster
relation (\ref{FTpsichi}) and we have again checked consistency with the form
factor equation (iii), which also yields the function $c_{\psi\chi}%
^{T}(k,l,W)$ of (\ref{FTpsichi}).

\subparagraph{Calculation of the function $c_{\psi\chi}^{T}(k,l,W):$}

In the same way as above for $c_{\psi\chi}^{J}$ we prove that $c_{\psi\chi
}^{T}(k,l,W)$ is independent of $k$ and for $k=1\operatorname{mod}N$%
\[
c_{\psi\chi}^{T}(k,l,W)=c_{\psi\chi}^{T}(k_{0},l_{0},W)(-1)^{\left(
N-1\right)  \left(  l-l_{0}\right)  /N}%
\]
The special case $c_{\psi\chi}^{T}(1,N-1,W)$ is calculated by the following
example, which implies%
\begin{equation}
c_{\psi\chi}^{T}(k,l,W)=ie^{i\pi l_{1}}\,W^{\frac{1}{N^{2}}-1}e^{-\frac{1}%
{2}\left(  1-\frac{1}{N}\right)  W}2\left(  2\pi\right)  ^{-\frac{1+N}{N^{2}}%
}e^{-i\pi\left(  N+\frac{1}{2N}\right)  }/\bar{F}(i\pi)\, \label{cT}%
\end{equation}

\subparagraph{Example:}

The particle anti-particle of (\ref{FT2}) and asymptotic behavior of the
particle anti-particle minimal form factor function (\ref{Fba}) imply%
\[
F_{1\bar{N}}^{T}(\underline{\theta}_{W})\overset{W\rightarrow\infty
}{\rightarrow}2e^{-\frac{1}{2}(\theta-\omega+W)}\frac{1}{iW}\left(  \left(
2\pi\right)  ^{-1-\frac{1}{N}}W^{\frac{1}{N}}e^{\frac{1}{2}W}e^{\frac{1}%
{2}\left(  \theta-\omega-i\pi\right)  }\right)  ^{\frac{1}{N}}/\bar{F}(i\pi).
\]
The asymptotic relation (\ref{FTpsichi}), (\ref{F1psi}) and (\ref{F1chi})
give
\[
F_{1\bar{N}}^{T}(\underline{\theta}_{W})\overset{W\rightarrow\infty
}{\rightarrow}c_{\psi\chi}^{T}(1,N-1,W)F_{1}^{\psi}(\theta)F_{\bar{N}}^{\chi
}(\omega)=c_{\psi\chi}^{T}(1,N-1,W)e^{-\frac{1}{2}\left(  1-\frac{1}%
{N}\right)  \theta}e^{\frac{1}{2}\left(  1-\frac{1}{N}\right)  \omega}%
\]
which means%
\[
c_{\psi\chi}^{T}(1,N-1,W)=\frac{1}{iW}2e^{-\frac{1}{2}W}\left(  \left(
2\pi\right)  ^{-1-\frac{1}{N}}W^{\frac{1}{N}}e^{\frac{1}{2}W}e^{-\frac{1}%
{2}i\pi}\right)  ^{\frac{1}{N}}/\bar{F}(i\pi)
\]
and (\ref{cT}).

\subsection{Theorem \ref{t3}}

\begin{theorem}
\label{t3}The cluster behavior of form factor of the fundamental field for the
number particles $n=1\operatorname{mod}N$ and $k=0\operatorname{mod}N$ reads
as%
\begin{equation}
F_{\underline{\alpha}}^{\psi^{\beta}}(\underline{\theta}_{W})\overset
{W\rightarrow\infty}{\rightarrow}\frac{1}{W}i\eta\mathbf{C}_{\gamma\bar
{\delta}}F_{\underline{\hat{\alpha}}}^{J^{\beta\bar{\delta}}}(\underline
{\hat{\theta}})F_{\underline{\check{\alpha}}}^{\psi^{\gamma}}(\underline
{\check{\theta}})=\frac{1}{W}2i\eta F_{\underline{\hat{\alpha}}}^{J_{a}%
}(\underline{\hat{\theta}})\left(  T_{a}\right)  _{\delta}^{\beta
}F_{\underline{\check{\alpha}}}^{\psi^{\delta}}(\underline{\check{\theta}}).
\label{Fpsi}%
\end{equation}

\end{theorem}

\begin{proof}
We investigate%
\[
F_{\underline{\alpha}}^{\psi}(\underline{\theta}_{W})=N_{n}^{\psi}%
F(\underline{\theta}_{W})\int d\underline{\underline{z}}\tilde{h}%
\,(\underline{\theta}_{W},\underline{\underline{z}}_{W})p^{\psi}%
(\underline{\theta}_{W},\underline{z}_{W})\tilde{\Phi}_{\underline{\alpha}%
}(\underline{\theta}_{W},\underline{\underline{z}}_{W})
\]
From the asymptotic behavior of $F(\underline{\theta}_{W},\underline{\omega
}_{W})\tilde{h}\,(\underline{\theta}_{W},\underline{\omega}_{W},\underline
{\underline{z}}_{W})$ and $p^{\psi}(\underline{\theta}_{W},\underline{z}_{W})$
in (\ref{Fha}), (\ref{ppsia}) and (\ref{expvec}) we derive for $W\rightarrow
\infty$ the exponential behavior\textbf{ }
\begin{equation}
F(\underline{\theta}_{W})\tilde{h}\,(\underline{\theta}_{W},\underline
{\underline{z}}_{W})p^{\psi}(\underline{\theta}_{W},\underline{\underline{z}%
}_{W})\varpropto\left(  e^{-\frac{1}{2}W}\right)  ^{\tilde{k}_{1}^{2}%
+\tilde{k}_{N-1}^{2}+\sum_{j=1}^{N-2}\left(  \tilde{k}_{j}-\tilde{k}%
_{j+1}\right)  ^{2}}~,~~\tilde{k}_{j}=k_{j}-k\left(  1-j/N\right)
\label{exppsi}%
\end{equation}
For $k=0\operatorname{mod}N$ and $l=1\operatorname{mod}N$ the leading
asymptotic behavior $\varpropto\left(  e^{-\frac{1}{2}W}\right)  ^{0}$ is
obtained for $\tilde{k}_{j}=0$ i.e. $k_{j}=k\left(  1-j/N\right)  $ and
$l_{j}=\left(  l-1\right)  \left(  1-j/N\right)  $ by (\ref{npsi}), which
implies (up to const.)%
\begin{multline}
F(\underline{\theta}_{W})\tilde{h}(\underline{\theta}_{W},\underline
{\underline{z}}_{W})p^{\psi}(\underline{\theta}_{W},\underline{z}_{W}%
)\tilde{\Phi}_{\underline{\alpha}}(\underline{\theta}_{W},\underline
{\underline{z}}_{W})\label{Fpsi0}\\
\rightarrow\left(  F(\underline{\hat{\theta}})\tilde{h}(\underline{\hat
{\theta}},\underline{\underline{\hat{z}}})\tilde{\Phi}_{\underline{\hat
{\alpha}}}(\underline{\hat{\theta}},\underline{\underline{\hat{z}}})\right)
\left(  F(\underline{\check{\theta}})\tilde{h}(\underline{\check{\theta}%
},\underline{\underline{\check{z}}})p^{\psi}(\underline{\check{\theta}%
},\underline{\check{z}})\tilde{\Phi}_{\underline{\check{\alpha}}}%
(\underline{\check{\theta}},\underline{\underline{\check{z}}})\right)
\end{multline}
However, this means that in leading order
\[
F_{\underline{\alpha}}^{\psi}(\underline{\theta}_{W})\rightarrow0
\]
because of Lemma \ref{l1}. The proof of the $\frac{1}{W}$ contribution is
similar to that one of Theorem \ref{t1}.\medskip

\noindent\textbf{Order }$\frac{1}{W}$: we have to apply the asymptotic
behavior of the h-function (\ref{ha}) and the Bethe state (\ref{PHIa}).

The result for the contribution of $h_{1}$ is (up to a constant)%
\begin{equation}
F_{\underline{\alpha},h_{1}}^{\psi^{1}}(\underline{\theta}_{W})\rightarrow
i\eta W^{-1}\left(  \mathbf{C}_{1\bar{1}}F_{\underline{\hat{\alpha}}%
}^{J^{1\bar{1}}}(\underline{\hat{\theta}})\right)  F_{\underline{\check
{\alpha}}}^{\psi^{1}}(\underline{\check{\theta}}) \label{Fh1}%
\end{equation}
and the result for the contribution of from\textbf{ }$\Phi_{1}$ is%
\[
F_{\underline{\alpha},\Phi_{1}}^{\psi^{1}}(\underline{\theta}_{W})\rightarrow
i\eta W^{-1}\left(  \mathbf{C}_{\gamma\bar{\delta}}^{(1)}F_{\underline
{\hat{\alpha}}}^{J^{1\bar{\delta}}}(\underline{\hat{\theta}})F_{\underline
{\check{\alpha}}}^{\psi^{\gamma}}(\underline{\check{\theta}})\right)
\]
Because $\mathbf{C}_{\gamma\bar{\delta}}^{(1)}+\mathbf{C}_{1\bar{1}}%
\delta_{\gamma}^{1}\delta_{\bar{\delta}}^{\bar{1}}=\mathbf{C}_{\gamma
\bar{\delta}}$ the claim (\ref{psiJpsi}) is proved.
\end{proof}

\subparagraph{Calculation of the function $c_{J\psi}^{\psi}(k,l,W):$}

defined by%
\[
F_{\underline{\alpha}}^{\psi^{\beta}}(\underline{\theta}_{W})\overset
{W\rightarrow\infty}{\rightarrow}c_{J\psi}^{\psi}(k,l,W)\mathbf{C}_{\gamma
\bar{\delta}}F_{\underline{\hat{\alpha}}}^{J^{\beta\bar{\delta}}}%
(\underline{\hat{\theta}})F_{\underline{\check{\alpha}}}^{\psi^{\gamma}%
}(\underline{\check{\theta}})
\]
We apply the procedure of Appendix \ref{ac}: Using $a(W)\rightarrow
e^{-i\pi\left(  1-\frac{1}{N}\right)  }$ of (\ref{Sa}) we check (\ref{tocheck}%
) and (\ref{tocheck1}) with $\sigma_{1}^{\psi}=e^{\left(  1-\frac{1}%
{N}\right)  i\pi}$ and $\sigma_{1}^{J}=1$ for $k=0\operatorname{mod}N$%
\begin{multline*}
\dot{\sigma}_{1}^{\psi}(n)S_{1\underline{\check{\alpha}}}^{\underline
{\check{\alpha}}^{\prime}1}(\theta+W,\underline{\check{\theta}})\overset
{W\rightarrow\infty}{\rightarrow}e^{\left(  1-\frac{1}{N}\right)  i\pi
}(-1)^{(N-1)+(1-1/N)(k+l-1)}(a(W))^{l}1_{\underline{\check{\alpha}}%
}^{\underline{\check{\alpha}}^{\prime}}\\
\rightarrow(-1)^{(N-1)+(1-1/N)k}1_{\underline{\check{\alpha}}}^{\underline
{\check{\alpha}}^{\prime}}=\dot{\sigma}_{1}^{J}(k)1_{\underline{\check{\alpha
}}}^{\underline{\check{\alpha}}^{\prime}}\,.
\end{multline*}%
\begin{multline*}
\dot{\sigma}_{1}^{\psi}(n)S_{\underline{\hat{\alpha}}\bar{1}}^{\bar{1}%
\hat{\alpha}^{\prime}}(\underline{\hat{\theta}}+W,\omega)\overset
{W\rightarrow\infty}{\rightarrow}e^{\left(  1-\frac{1}{N}\right)  i\pi
}(-1)^{(N-1)+(1-1/N)(k+l)}(a(W))^{(N-1)k}1_{\underline{\hat{\alpha}}%
}^{\underline{\hat{\alpha}}^{\prime}}\\
\overset{W\rightarrow\infty}{\rightarrow}(-1)^{\left(  N-1\right)  k}%
\dot{\sigma}_{1}^{\psi}(l)1_{\underline{\hat{\alpha}}}^{\underline{\hat
{\alpha}}^{\prime}}\,.
\end{multline*}
Therefore $c_{J\psi}^{\psi}(k,l,W)$ is independent of $k$ and $l$. It is
convenient to consider the special case $c_{J\psi}^{\psi}(N,1,W):$

1) We take the bound states $\bar{1}=(\hat{\alpha}_{2}\dots\hat{\alpha}_{N})$
and calculate for $\underline{\theta}_{W}=(\hat{\theta}+W,\hat{\omega
}+W,\check{\theta})$%
\begin{multline}
\operatorname*{Res}_{\hat{\theta}=i\pi+\hat{\omega}}F_{1\bar{1}1}^{\psi
}(\underline{\theta}_{W})=2i\,\mathbf{C}_{1\bar{1}}\,F_{1}^{\psi}%
(\check{\theta})\left(  1-\dot{\sigma}_{\bar{1}}^{\psi}(N+1)S_{\bar{1}%
,1}^{1,\bar{1}}(\hat{\omega}+W,\check{\theta})\right) \\
\overset{W\rightarrow\infty}{\rightarrow}-2i\,\mathbf{C}_{1\bar{1}}%
\,i\eta\left(  1-1/N\right)  \frac{1}{W}F_{1}^{\psi}(\check{\theta}).
\end{multline}
It was used that $\sigma_{1}^{\psi}=e^{i\pi\left(  1-\frac{1}{N}\right)
},~\sigma_{1}^{\psi}\sigma_{\bar{1}}^{\psi}=(-1)^{(N-1)}$ (see \cite{BFK1}),
(\ref{Sb}), (\ref{Sa}) and $a(\theta)a(-\theta)=1$ imply%
\begin{multline*}
\dot{\sigma}_{\bar{1}}^{\psi}(N+1)S_{\bar{1},1}^{1,\bar{1}}(\hat{\omega
}+W,\check{\theta})\overset{W\rightarrow\infty}{\rightarrow}e^{-i\pi\left(
1-\frac{1}{N}\right)  }(-1)^{(1-1/N)N}(-1)^{(N-1)}a(-W)\\
\rightarrow e^{-i\pi\left(  1-\frac{1}{N}\right)  }(-1)^{(1-1/N)N}%
(-1)^{(N-1)}e^{i\pi\left(  1-\frac{1}{N}\right)  }e^{i\eta\left(  1-\frac
{1}{N}\right)  \frac{1}{W}}\rightarrow1+i\eta\left(  1-1/N\right)  \frac{1}{W}%
\end{multline*}

2) Taking first $W\rightarrow\infty$ and then the $\operatorname*{Res}$ means%
\begin{multline*}
\operatorname*{Res}_{\hat{\theta}=i\pi+\hat{\omega}}\left(  F_{1\bar{1}%
1}^{\psi}(\underline{\theta}_{W})\overset{W\rightarrow\infty}{\rightarrow
}c_{J\psi}^{\psi}(N,1,W)\left(  \mathbf{C}_{\gamma\bar{\delta}}F_{1\bar{1}%
}^{J^{1\bar{\delta}}}(\omega,\theta)F_{1}^{\psi^{\gamma}}(\check{\theta
})\right)  \right) \\
=c_{J\psi}^{\psi}(N,1,W)\left(  1-1/N\right)  \left(  -2i\right)
\mathbf{C}_{1\bar{1}}F_{1}^{\psi^{\gamma}}(\check{\theta}).
\end{multline*}
As result we obtain $c_{J\psi}^{\psi}(k,l,W)=i\eta\frac{1}{W}$ which proves
(\ref{psiJpsi}).

\subparagraph{Equivalence:}

Using the general relations (\ref{TT}) and (\ref{JJ}) we obtain%
\begin{align*}
&  2F_{\underline{\hat{\alpha}}}^{J_{a}}(\underline{\hat{\theta}})\left(
T_{a}\right)  _{\delta}^{\beta}F_{\underline{\check{\alpha}}}^{\psi^{\delta}%
}(\underline{\check{\theta}})\\
&  =2\mathbf{C}_{\gamma\bar{\epsilon}}\left(  T_{a}\right)  _{\alpha}^{\gamma
}F_{\underline{\hat{\alpha}}}^{J^{\alpha\bar{\epsilon}}}(\underline
{\hat{\theta}})\left(  T_{a}\right)  _{\delta}^{\beta}F_{\underline
{\check{\alpha}}}^{\psi^{\delta}}(\underline{\check{\theta}})\\
&  =2\mathbf{C}_{\gamma\bar{\epsilon}}\left(  \frac{1}{2}\left(
\delta_{\alpha}^{\beta}\delta_{\delta}^{\gamma}-1/N\,\delta_{\alpha}^{\gamma
}\delta_{\delta}^{\beta}\right)  \right)  F_{\underline{\hat{\alpha}}%
}^{J^{\alpha\bar{\epsilon}}}(\underline{\hat{\theta}})F_{\underline
{\check{\alpha}}}^{\psi^{\delta}}(\underline{\check{\theta}})\\
&  =\mathbf{C}_{\gamma\bar{\epsilon}}F_{\underline{\hat{\alpha}}}%
^{J^{\beta\bar{\epsilon}}}(\underline{\hat{\theta}})F_{\underline
{\check{\alpha}}}^{\psi^{\gamma}}(\underline{\check{\theta}})
\end{align*}
because of $\mathbf{C}_{\alpha\bar{\epsilon}}F_{\underline{\hat{\alpha}}%
}^{J^{\alpha\bar{\epsilon}}}(\underline{\hat{\theta}})=0$.

\subsection{Theorem \ref{t4}}

\begin{theorem}
\label{t4}The cluster behavior of form factor of the fundamental field for
particle number $n=1\operatorname{mod}N$ and $k=1\operatorname{mod}N$ reads as%
\[
F_{\underline{\alpha}}^{\psi^{\alpha}}(\underline{\theta}_{W})\overset
{W\rightarrow\infty}{\rightarrow}(-1)^{l_{1}}e^{-\frac{1}{2}\left(  1-\frac
{1}{N}\right)  W}F_{\underline{\hat{\alpha}}}^{\psi^{\alpha}}(\underline
{\hat{\theta}})F_{\underline{\check{\alpha}}}^{\phi}(\underline{\check{\theta
}})\,.
\]

\end{theorem}

\begin{proof}
We investigate%
\[
F_{\underline{\alpha}}^{\psi}(\underline{\theta}_{W})=N_{n}^{\psi}%
F(\underline{\theta}_{W})\int d\underline{\underline{z}}\tilde{h}%
\,(\underline{\theta}_{W},\underline{\underline{z}}_{W})p^{\psi}%
(\underline{\theta}_{W},\underline{\underline{z}}_{W})\tilde{\Psi}%
_{\underline{\alpha}}(\underline{\theta}_{W},\underline{\underline{z}}_{W})
\]
and obtain as above the exponential behavior\textbf{ }(\ref{exppsi}). The
leading behavior $\varpropto\left(  e^{-\frac{1}{2}W}\right)  ^{1-\frac{1}{N}%
}$ is obtained for $\tilde{k}_{j}=j/N-1$ which means%
\[
k_{j}=\left(  k-1\right)  \left(  1-j/N\right)  ,~l_{j}=l\left(  1-j/N\right)
~,~~\text{for }~j=1,\dots,N-1
\]
and (up to const.)%
\begin{align*}
&  F(\underline{\theta}_{W})p^{\psi}(\underline{\theta}_{W},\underline
{\underline{z}}_{W})\tilde{\Psi}_{\underline{\alpha}}(\underline{\theta}%
_{W},\underline{\underline{z}}_{W})\\
&  \rightarrow\left(  e^{-\frac{1}{2}W}\right)  ^{1-\frac{1}{N}}\left(
F(\underline{\hat{\theta}})\tilde{h}(\underline{\hat{\theta}},\underline
{\hat{z}})p^{\psi}(\underline{\hat{\theta}},\underline{\hat{z}})\tilde{\Psi
}_{\underline{\hat{\alpha}}}(\underline{\hat{\theta}},\underline{\hat{z}%
})\right)  \left(  F(\underline{\check{\theta}})\tilde{h}(\underline
{\check{\theta}},\underline{\check{z}})p^{\phi}(\underline{\check{\theta}%
},\underline{\check{z}})\tilde{\Psi}_{\underline{\check{\alpha}}}%
(\underline{\check{\theta}},\underline{\check{z}})\right)  .
\end{align*}
proving (\ref{psipsiphi}).
\end{proof}

\subparagraph{Calculation of the function $c_{\psi\phi}^{\psi}(k,l,W):$}

defined by%
\begin{equation}
F_{\underline{\alpha}}^{\psi}(\underline{\theta}_{W})\overset{W\rightarrow
\infty}{\rightarrow}c_{\psi\phi}^{\psi}(k,l,W)F_{\underline{\hat{\alpha}}%
}^{\psi}(\underline{\hat{\theta}})F_{\underline{\check{\alpha}}}^{\phi
}(\underline{\check{\theta}}). \label{Fpsipsiphix}%
\end{equation}
We apply the procedure of Appendix \ref{ac}: Using $a(W)\overset
{W\rightarrow\infty}{\rightarrow}e^{-i\pi\left(  1-\frac{1}{N}\right)  }$ of
(\ref{Sa}) we check (\ref{tocheck}) and (\ref{tocheck1}) with $\sigma
_{1}^{\psi}=e^{i\pi\left(  1-\frac{1}{N}\right)  },~Q^{\psi}=1$ and
$\sigma_{1}^{\phi}=e^{-i\eta},~Q^{\phi}=0$%
\begin{align*}
\dot{\sigma}_{1}^{\psi}(n)S_{1\underline{\check{\alpha}}}^{\underline
{\check{\alpha}}^{\prime}1}(\theta+W,\underline{\check{\theta}})  &
\rightarrow e^{i\pi\left(  1-\frac{1}{N}\right)  }(-1)^{(N-1)+(1-1/N)(k+l-1)}%
(a(W))^{l}1_{\underline{\check{\alpha}}}^{\underline{\check{\alpha}}^{\prime}%
}\rightarrow\dot{\sigma}_{1}^{\psi}(k)1_{\underline{\check{\alpha}}%
}^{\underline{\check{\alpha}}^{\prime}}\\
\dot{\sigma}_{1}^{\psi}(n)S_{\hat{\alpha}\bar{1}}^{\bar{1}\hat{\alpha}%
^{\prime}}(\underline{\hat{\theta}}+W,\omega)  &  \rightarrow e^{i\pi\left(
1-\frac{1}{N}\right)  }(-1)^{(N-1)+(1-1/N)(k+l-1)}(a(W))^{(N-1)k}%
1_{\underline{\hat{\alpha}}}^{\underline{\hat{\alpha}}^{\prime}}\\
&  \rightarrow(-1)^{\left(  N-1\right)  k}\dot{\sigma}_{1}^{\phi}%
(l)1_{\hat{\alpha}}^{\hat{\alpha}^{\prime}}%
\end{align*}
Therefore (\ref{c1}) and (\ref{c2}) imply that $c_{\psi\phi}^{\psi}(k,l,W)$ is
independent of $k$ and (\ref{c3}) and (\ref{c4}) that for
$k=1\operatorname{mod}N$%
\[
c_{\psi\phi}^{\psi}(k,l,W)=c_{\psi\phi}^{\psi}(k_{0},l_{0},W)(-1)^{\left(
N-1\right)  \left(  l-l_{0}\right)  /N}%
\]
The special case $c_{\psi\chi}^{\psi}(1,0,W)$ is obtained by
(\ref{Fpsipsiphix}) for $\underline{\check{\alpha}}=\emptyset$ and the form
factor equation (v) with spin $s^{\psi}=-\frac{1}{2}\left(  1-\frac{1}%
{N}\right)  $%
\begin{align*}
F_{\underline{\alpha}}^{\psi}(\underline{\theta}_{W})  &  \rightarrow
c_{\psi\chi}^{J}(1,0,W)F_{\underline{\hat{\alpha}}}^{\psi}(\underline
{\hat{\theta}})F_{\emptyset}^{\phi}\\
F_{\underline{\alpha}}^{\psi}(\underline{\theta}_{W})  &  \rightarrow
e^{s^{\psi}W}=e^{-\frac{1}{2}\left(  1-\frac{1}{N}\right)  W}F_{\underline
{\hat{\alpha}}}^{\psi}(\underline{\hat{\theta}})
\end{align*}
if we normalize the field $\phi(x)$ by $F_{\emptyset}^{\phi}=\left\langle
0|\phi(x)|0\right\rangle =1$, this gives the result%
\[
c_{\psi\phi}^{\psi}(k,l,W)=(-1)^{l_{1}}e^{-\frac{1}{2}\left(  1-\frac{1}%
{N}\right)  W}%
\]
because $l_{1}=l(1-1/N)$.

\section{Summary\label{s6}}

In this article we investigate the rapidity clustering of exact multi-particle
form factors of the $SU(N)$ chiral Gross-Neveu model. For some examples of
local fields, in particular, the Noether current, the energy momentum tensor,
the fundamental spinor field etc, we explicitly demonstrate the clustering or
factorization phenomena. In a forthcoming paper we will consider the form
factor of the Noether current in a special form, in order to connect the
asymptotic clustering with Bjorken scattering.

\paragraph{Acknowledgment:}

It is our pleasure to thank G. Savvidy for helpful discussion on Bjorken
scattering. One of authors (M.K.) thanks Felix von Oppen for hospitality at
the Institut f\"{u}r Theoretische Physik, Freie Universit\"{a}t Berlin. H.B. was
partially supported by the Armenian State Committee of Science in the
framework of the research projects 18T-1C-340,191T-008 and 20RF-142. A.F.
acknowledges support from DAAD (Deutscher Akademischer Austauschdienst) and
CNPq (Conselho Nacional de Desenvolvimento Cient\'{\i}fico e Tecnol\'{o}gico).

\appendix

\section*{Appendix}

\addcontentsline{toc}{part}{Appendix}

\renewcommand{\theequation}{\mbox{\Alph{section}.\arabic{equation}}} \setcounter{equation}{0}

\section{Some lemmata}

\begin{lemma}
\label{l1}For $n=0\operatorname{mod}N,~n_{j}=n\left(  1-\frac{j}{N}\right)  $
and $p(\underline{\theta},\underline{\underline{z}})=$ independent of
$z_{i}^{(j)}$ the K-function vanishes%
\begin{equation}
K_{\underline{\alpha}}(\underline{\theta})=\int d\underline{\underline{z}%
}\tilde{h}\,(\underline{\theta},\underline{\underline{z}})\tilde{\Phi
}_{\underline{\alpha}}(\underline{\theta},\underline{\underline{z}})=0
\label{F0}%
\end{equation}

\end{lemma}

For $SU(2)$ the proof of this lemma is quite analog to that for the
Sine-Gordon model in \cite{BK}. For general $N$ we present an example (see
Proposition \ref{p1}).

\begin{lemma}
\label{l3}For $SU(2)$ and $m=n/2$%
\begin{equation}
K_{\underline{\alpha}}(\underline{\theta})=\frac{1}{m!}\int_{\mathcal{C}%
_{\underline{\theta}}}dz_{1}\cdots\int_{\mathcal{C}_{\underline{\theta}}%
}dz_{m}\,\tilde{h}(\underline{\theta},{\underline{z}})\left(  -\sum_{i=1}%
^{m}z_{i}\right)  \,\tilde{\Psi}_{\underline{\alpha}}(\underline{\theta
},{\underline{z}})=\left(  -1\right)  ^{m}8\pi^{5}i\left(  K^{J}%
(\underline{\theta})M_{1}^{2}\right)  _{\underline{\alpha}} \label{Kz}%
\end{equation}
which is a non-highest weight K-function and%
\begin{equation}
K_{\underline{\alpha}}(\underline{\theta})=i\pi\frac{1}{m!}\int_{\mathcal{C}%
_{\underline{\theta}}}dz_{1}\cdots\int_{\mathcal{C}_{\underline{\theta}}%
}dz_{m}\,\tilde{h}(\underline{\theta},{\underline{z}})\sum_{j=1}^{m}%
\,\tilde{\Phi}_{\underline{\alpha}}^{D_{j}}(\underline{\theta},{\underline{z}%
})=-\left(  -1\right)  ^{m}8\pi^{5}iK_{\underline{\alpha}}^{J}(\underline
{\theta}) \label{Kbb}%
\end{equation}
where $\tilde{\Phi}_{\underline{\hat{\alpha}}}^{D_{j}}(\underline{\hat{\theta
}},\underline{\hat{z}})=\left(  \Omega{C}({\underline{\hat{\theta}}},\hat
{z}_{k})\dots{D}({\underline{\hat{\theta}}},\hat{z}_{j})\dots{C}%
({\underline{\hat{\theta}}},\hat{z}_{1})\right)  _{\underline{\hat{\alpha}}}$.
\end{lemma}

For $SU(2)$ the proofs are similar to the one of Lemma \ref{l1}. For general
$N$ see the proofs of Propositions \ref{p2} and \ref{p3}.

\setcounter{equation}{0}

\section{Examples of particle anti-particle form factors}

\label{sex}

\subsection{Bound states - Anti-particles}

The following is taken from \cite{BFK0,BFK1,BFK3}\footnote{See also
\cite{BKZ2} for $U(N)$ Bethe ansatz.}.

\subsubsection{Bound state S-matrix}

The S-matrix of a particle and an anti-particle (which is a bound state of
$N-1$ particles (\ref{bs}) \cite{KKS}) is
\begin{equation}
S_{\bar{\alpha}\beta}^{\delta\bar{\gamma}}(\theta)=(-1)^{N-1}\left(
\delta_{\bar{\alpha}}^{\bar{\gamma}}\delta_{\beta}^{\delta}\,b(\pi
i-\theta)+\mathbf{C}^{\delta\bar{\gamma}}\mathbf{C}_{\bar{\alpha}\beta}\,c(\pi
i-\theta)\right)  \label{Sb}%
\end{equation}
where the charge conjugation matrices are given by (\ref{C}).

\subsubsection{Bound state form factors}

The general form factor formula for $n$ particles $\underline{\alpha}$ and
$\bar{n}$ anti-particles $\underline{\bar{\delta}}$ is%
\[
F_{\underline{\alpha}\underline{\bar{\delta}}}(\underline{\theta}%
,\underline{\omega})=N_{n\bar{n}}F(\underline{\theta},\underline{\omega
})K_{\underline{\alpha}\underline{\bar{\delta}}}(\underline{\theta}%
,\underline{\omega})
\]
where
\begin{align*}
F(\underline{\theta},\underline{\omega})  &  =\left(  \prod_{1\leq i<j\leq
n}F(\theta_{ij})\right)  \left(  \prod_{1\leq i\leq n}\prod_{1\leq j\leq
\bar{n}}\bar{F}(\theta_{i}-\omega_{j})\right)  \left(  \prod_{1\leq
i<j\leq\bar{n}}F(\omega_{ij})\right) \\
K_{\underline{\alpha}\underline{\bar{\delta}}}(\underline{\theta}%
,\underline{\omega})  &  =\int_{\mathcal{C}_{\underline{\theta}}%
,\underline{\omega}}\underline{dz}\tilde{h}(\underline{\theta},\underline
{z})p(\underline{\theta},\underline{\omega},\underline{z})\tilde{\Psi
}_{\underline{\alpha}\underline{\bar{\delta}}}(\underline{\theta}%
,\underline{\omega},\underline{z})\\
\tilde{h}(\underline{\theta},\underline{z})  &  =\prod_{i=1}^{n}\prod
_{j=1}^{m}\tilde{\phi}(\theta_{i}-z_{j})\prod_{1\leq i<j\leq m}\tau(z_{ij})\\
\tilde{\phi}(\theta)  &  =\Gamma\left(  -\frac{\theta}{2\pi i}\right)
\Gamma\left(  1-\frac{1}{N}+\frac{\theta}{2\pi i}\right)
\end{align*}
with $\int_{\mathcal{C}_{\underline{\theta}},\underline{\omega}}\underline
{dz}=\frac{1}{m!}\int_{\mathcal{C}_{\underline{\theta}},\underline{\omega}%
}dz_{1}\dots\int_{\mathcal{C}_{\underline{\theta}},\underline{\omega}}dz_{m}$.
The minimal F-function for a particle and an anti-particle $\bar{F}\left(
\theta\right)  $ is defined in (\ref{Fbar}) and satisfies (\ref{Fbb}) and the
asymptotic behavior (\ref{Fba}). The 0-level Bethe ansatz state writes in
terms of the basic states as%
\[
\tilde{\Psi}_{\underline{\alpha}\underline{\bar{\delta}}}(\underline{\theta
},\underline{\omega},\underline{z})=L_{\underline{\beta}\underline
{\bar{\epsilon}}}(\underline{z},\underline{\omega})\tilde{\Phi}_{\underline
{\alpha}\underline{\bar{\delta}}}^{\underline{\beta}\underline{\bar{\epsilon}%
}}(\underline{\theta},\underline{\omega},\underline{z})\,
\]
The function $L_{\underline{\beta}(\sigma)}(\underline{z},\underline{\omega
}),~(\sigma)=(1,\sigma_{2},\dots,\sigma_{N-1})$ is given by the 1-level
off-shell Bethe ansatz, etc. The final formula is%
\begin{align}
K_{\underline{\alpha}\underline{\bar{\delta}}}^{\mathcal{O}}(\underline
{\theta},\underline{\omega})  &  =\int d\underline{z}_{1}\dots\int
d\underline{z}_{N-1}\tilde{h}\,(\underline{\theta},\underline{\omega
},\underline{\underline{z}})p^{\mathcal{O}}(\underline{\theta},\underline
{\omega},\underline{\underline{z}})\tilde{\Phi}_{\underline{\alpha}%
\underline{\bar{\delta}}}(\underline{\theta},\underline{\omega},\underline
{\underline{z}})\,\label{Knn}\\
\tilde{h}\,(\underline{\theta},\underline{\omega},\underline{\underline{z}})
&  =\prod_{j=0}^{N-2}\tilde{h}(\underline{z}_{j},\underline{z}_{j+1}%
)\prod_{i=1}^{\bar{n}}\prod_{j=1}^{n_{N-1}}\tilde{\chi}(\omega_{i}%
-z_{j}^{(N-1)})\label{hnn}\\
\tilde{\chi}(\omega)  &  =\Gamma\left(  \frac{1}{2}+\frac{\omega}{2\pi
i}\right)  \Gamma\left(  \frac{1}{2}-\frac{1}{N}-\frac{\omega}{2\pi i}\right)
. \label{chi}%
\end{align}
The complete Bethe ansatz state is%
\begin{equation}
\tilde{\Phi}_{\underline{\alpha}\underline{\bar{\delta}}}(\underline{\theta
},\underline{\omega},\underline{\underline{z}})=\tilde{\Phi}^{(N-2)}%
\,_{\underline{\alpha}_{N-2}\underline{\bar{\delta}}_{N-2}}^{\underline
{\alpha}_{N-1}\underline{\bar{\delta}}_{N-1}}(\underline{z}_{N-2}%
,\underline{\omega},\underline{z}_{N-1})\dots\tilde{\Phi}^{(1)}\,_{\underline
{\alpha}_{1}\underline{\bar{\delta}}_{1}}^{\underline{\alpha}_{2}%
\underline{\bar{\delta}}_{2}}(\underline{z}_{1},\underline{\omega}%
,\underline{z}_{2})\tilde{\Phi}_{\underline{\alpha}\underline{\bar{\delta}}%
}^{\underline{\alpha}_{1}\underline{\bar{\delta}}_{1}}(\underline{\theta
},\underline{\omega},\underline{z}_{1}) \label{Bnn}%
\end{equation}
where $\underline{\alpha}_{N-1}=(N,\dots,N)$ and $\underline{\bar{\delta}%
}_{N-1}=(\bar{N},\dots,\bar{N})$ consists of highest weight bound states
$\bar{N}=(1,2,\dots,N-1).$ The state of level $j$ is given by monodromy
matrices as
\[
\tilde{\Phi}^{(j)}\,_{\underline{\alpha}_{j}\underline{\bar{\delta}}_{j}%
}^{\underline{\alpha}_{j+1}\underline{\bar{\delta}}_{j+1}}(\underline{z}%
_{j},\underline{\omega},\underline{z}_{j+1})=\tilde{T}^{(j)}\,_{\underline
{\alpha}_{j}\underline{\bar{\delta}}_{j},\underline{j+1}}^{\underline{\alpha
}_{j+1},\underline{j+1}\underline{\bar{\delta}}_{j+1}}(\underline{z}%
_{j},\underline{\omega},\underline{z}_{j+1})=%
\begin{array}
[c]{c}%
\unitlength3mm\begin{picture}(9,9) \thicklines \put(9,6){\oval(18.4,8.4)[lb]} \put(9,6){\oval(18,8)[lb]} \put(4,1){\line(0,1){3}} \put(3.8,1){\line(0,1){3}} \put(3.8,0){$\underline\alpha_j$} \put(8,1){\line(0,1){5}} \put(7.8,1){\line(0,1){5}} \put(7.8,-.2){$\underline{\bar\delta}_j$} \put(-.5,7){$\underline\alpha_{j+1}$} \put(2.8,4.6){$\underline{j+1}$} \put(7.,7){$\underline{\bar\delta}_{j+1}$} \put(9.2,1.8){$\underline{j+ 1}$} \put(4.4,3){$\underline z_j$} \put(6.6,3){$\underline\omega$} \put(.4,3.5){$\underline z_{j+1}$} \end{picture}
\end{array}
\]
If there are $n$ particles and $\bar{n}$ anti-particles the $SU(N)$ weights
are \cite{BKZ2,BFK3}%
\begin{align}
w  &  =\left(  n-n_{1},n_{1}-n_{2},\dots,n_{N-2}-n_{N-1},n_{N-1}-\bar
{n}\right)  +\bar{n}(1,\dots,1)\label{wb}\\
&  =w^{\emph{O}}+L(1,\dots,1)\nonumber
\end{align}
where $n_{1}=m,n_{2},\dots$ are the numbers of $C$ operators in the various
levels of the nesting, $w^{\emph{O}}$ is the weight vector of the operator
$\mathcal{O}$ and $L=0,1,2,\dots$.

\subsection{Lemma \ref{l1} for general $N$ and $n=\bar{n}=1$}

\begin{proposition}
\label{p1}The K-function given by (\ref{K}) with p-function $=1$
\begin{equation}
K_{\alpha\bar{\delta}}(\theta,\omega)=\int d\underline{\underline{z}}\tilde
{h}\,(\underline{\theta},\underline{\omega},\underline{\underline{z}}%
)\tilde{\Phi}_{\alpha\bar{\delta}}(\theta,\omega,\underline{\underline{z}})=0
\end{equation}
for $n=\bar{n}=1$.
\end{proposition}

\begin{proof}
The weight formula (\ref{wb}) implies that $n_{j}=1$ for $j=1,\dots,N-1$ and
the L-function of level $j$ is%
\[
L_{\beta\bar{\gamma}}^{(j)}(z,\omega)=\int_{\mathcal{C}}du\,\tilde{\phi
}\left(  z-u\right)  L_{\beta^{\prime}\bar{\gamma}^{\prime}}^{(j+1)}%
(u,\omega)\left(  T_{\beta\bar{\gamma},j+1}^{\beta^{\prime},j+1\bar{\gamma
}^{\prime}}(z,\omega,u)\right)  =\mathbf{C}_{\beta\bar{\gamma}}^{(j)}%
L_{ca}^{(j)}(z,\omega)
\]
where (\ref{r1}) - (\ref{r4}) have been used. For $j=0$%
\[
K_{\alpha\bar{\delta}}(\underline{\theta},\underline{\omega})=\mathbf{C}%
_{\alpha\bar{\delta}}L_{ca}^{(0)}(\theta,\omega)=0
\]
by (\ref{r1a}).
\end{proof}

\subsection{Theorem \ref{t1} for general $N$ and $n=\bar{n}=2,~k=\bar{k}=1$}

We consider form factors of the pseudo potential $J(x)$ for particles and
anti-particles. Formula (\ref{wb}) means, generalizing (\ref{nJ})%
\begin{equation}
n_{j}=n\left(  1-j/N\right)  +\bar{n}j/N-1\,. \label{nJb}%
\end{equation}
and the p-function is \cite{BFK3}%
\begin{equation}
p^{J}(\underline{\theta},\underline{\omega},\underline{z},\underline
{z}^{(N-1)})=\frac{\left(  \prod e^{\frac{1}{2}z_{i}^{(1)}}\right)  \left(
\prod e^{\frac{1}{2}z_{i}^{(N-1)}}\right)  \left(  \prod e^{-\frac{1}{2}%
\theta_{i}}\right)  \left(  \prod e^{-\frac{1}{2}\omega_{i}}\right)  }{\sum
e^{-\theta_{i}}+\sum e^{-\omega_{i}}} \label{pJb}%
\end{equation}
with the asymptotic behavior%
\begin{equation}
p^{J}(\underline{\theta}_{W},\underline{\omega}_{W},\underline{\underline{z}%
}_{W})\rightarrow e^{-\frac{1}{2}W\left(  k-k_{1}-k_{N-1}\right)  }\left(
{\textstyle\prod} e^{-\frac{1}{2}\hat{\theta}}\right)  \left(
{\textstyle\prod} e^{\frac{1}{2}\hat{z}^{(1)}}\right)  \left(
{\textstyle\prod} e^{\frac{1}{2}\hat{z}^{(N-1)}}\right)  \left(
{\textstyle\prod} e^{-\frac{1}{2}\hat{\omega}}\right)  p^{J}(\underline
{\check{\theta}},\underline{\underline{\check{z}}}). \label{pJab}%
\end{equation}
In particular for $n=\bar{n}=2$ and $k=\bar{k}=1$ we prove the proposition:

\begin{proposition}
\label{p2}The form factor of the current for $n=\bar{n}=2$ and $k=\bar{k}=1$
satisfies the clustering formula (\ref{FJ}) in the form%
\begin{multline}
F_{\underline{\alpha}\underline{\bar{\delta}}}^{J^{\beta\underline
{\bar{\epsilon}}}}(\underline{\theta}_{W},\underline{\omega}_{W}%
)\overset{W\rightarrow\infty}{\rightarrow}i\eta W^{-1}\left(  \mathbf{C}%
_{\gamma\bar{\kappa}}F_{\hat{\alpha}\hat{\bar{\delta}}}^{J^{\beta\bar{\kappa}%
}}(\hat{\theta},\hat{\omega})F_{\check{\alpha}\check{\bar{\delta}}}%
^{J^{\gamma\underline{\bar{\epsilon}}}}(\check{\theta},\check{\omega
})-\mathbf{C}_{\gamma\bar{\kappa}}F_{\hat{\alpha}\hat{\bar{\delta}}%
}^{J^{\gamma\underline{\bar{\epsilon}}}}(\hat{\theta},\hat{\omega}%
))F_{\check{\alpha}\check{\bar{\delta}}}^{J^{\beta\bar{\kappa}}}(\check
{\theta},\check{\omega})\right) \label{Fexa}\\
=-2\eta W^{-1}f_{abc}F_{\hat{\alpha}\hat{\bar{\delta}}}^{J_{b}}(\hat{\theta
},\hat{\omega})F_{\check{\alpha}\check{\bar{\delta}}}^{J_{c}}(\check{\theta
},\check{\omega}).
\end{multline}

\end{proposition}

\begin{proof}
The exponential behavior (\ref{expJ}) implies for $n=\bar{n}=2$ and $k=\bar
{k}=1$ that $k_{j}=1$ and $l_{j}=0$ for $j=1,\dots,N-1.$ We investigate for
$J=J^{1\bar{N}}$ ($\bar{N}=$ bound state $(1\dots N-1)$)%
\begin{equation}
K_{\underline{\alpha}\underline{\bar{\delta}}}^{J}(\underline{\theta
},\underline{\omega})=\int d\underline{\underline{z}}\tilde{h}\,(\underline
{\theta},\underline{\omega},\underline{\underline{z}})p^{J}(\underline{\theta
},\underline{\omega},\underline{\underline{z}})\tilde{\Phi}_{\underline
{\alpha}\underline{\bar{\delta}}}(\underline{\theta},\underline{\omega
},\underline{\underline{z}})\, \label{KJ}%
\end{equation}
We have proved in theorem \ref{t1} that in leading order
\[
F_{\underline{\alpha}}^{J}(\underline{\theta}_{W})\rightarrow0
\]

\noindent\textbf{Order }$\frac{1}{W}$: we have to apply the asymptotic
behavior of the h-function (\ref{ha}) and the Bethe state (\ref{PHIa}).

The \textbf{result for the contribution of }$h_{1}$ is%
\[
F_{\underline{\alpha}\underline{\bar{\delta}},h_{1}}^{J^{1\bar{N}}}%
(\underline{\theta}_{W},\underline{\omega}_{W})\rightarrow i\eta W^{-1}\left(
\mathbf{C}_{1\bar{1}}F_{\hat{\alpha}\hat{\bar{\delta}}}^{J^{1\bar{1}}}%
(\hat{\theta},\hat{\omega})-\mathbf{C}_{N\bar{N}}F_{\hat{\alpha}\hat
{\bar{\delta}}}^{J^{N\bar{N}}}(\hat{\theta},\hat{\omega})\right)
F_{\check{\alpha}\check{\bar{\delta}}}^{J^{1\bar{N}}}(\check{\theta}%
,\check{\omega})
\]
and the \textbf{result for the contribution of from }$\Phi_{1}$ is%
\begin{align*}
F_{\underline{\alpha}\underline{\bar{\delta}},\Phi_{1}}^{J^{1\bar{N}}%
}(\underline{\theta}_{W},\underline{\omega}_{W})  &  \rightarrow i\eta
W^{-1}\left(  \mathbf{C}_{\delta\bar{\kappa}}^{(1)}F_{\hat{\alpha}\hat
{\bar{\delta}}}^{J^{1\bar{\kappa}}}(\hat{\theta},\hat{\omega})F_{\check
{\alpha}\check{\bar{\delta}}}^{J^{\delta\bar{N}}}(\check{\theta},\check
{\omega})-\mathbf{C}_{\delta\bar{\kappa}}F_{\hat{\alpha}\hat{\bar{\delta}}%
}^{J^{\delta\bar{N}}}(\hat{\theta},\hat{\omega})F_{\check{\alpha}\check
{\bar{\delta}}}^{J^{1\bar{\kappa}}}(\check{\theta},\check{\omega})\right. \\
&  +\left.  \mathbf{C}_{N\bar{N}}F_{\hat{\alpha}\hat{\bar{\delta}}}%
^{J^{N\bar{N}}}(\hat{\theta},\hat{\omega})F_{\check{\alpha}\check{\bar{\delta
}}}^{J^{1\bar{N}}}(\check{\theta},\check{\omega})\right)  .
\end{align*}
Because $\mathbf{C}_{\delta\bar{\kappa}}^{(1)}+\mathbf{C}_{1\bar{1}}%
\delta_{\delta}^{1}\delta_{\bar{\kappa}}^{\bar{1}}=\mathbf{C}_{\delta
\bar{\kappa}}$ (see (\ref{r4a})) the relation (\ref{Fexa}) is proved.
\end{proof}

\subsection{Theorem \ref{t3} for general $N$ and $n=2,~\bar{n}=1,~k=\bar{k}%
=1$}

We consider form factors of the fundamental field $\psi(x)$ for particles and
anti-particles. Formula (\ref{wb}) means, generalizing (\ref{npsi})%
\begin{equation}
n_{j}=\left(  n-1\right)  \left(  1-j/N\right)  +\bar{n}j/N,~j=1,\dots,N-1.
\label{npsib}%
\end{equation}
and the p-function is
\begin{equation}
p^{\psi}(\underline{\theta},\underline{\omega},\underline{z})=e^{\frac{1}%
{2}n_{1}i\eta}e^{i\pi\bar{n}\left(  1-\frac{2}{N}\right)  }\left(
{\textstyle\prod\nolimits_{i=1}^{n}} e^{-\frac{1}{2}\left(  1-\frac{1}%
{N}\right)  \theta_{i}}\right)  {\textstyle\prod\nolimits_{i=1}^{\bar{n}}}
\left(  e^{-\frac{1}{2}\frac{\omega}{N}}\right)  \left(  {\textstyle\prod
\nolimits_{i=1}^{n_{1}}} e^{\frac{1}{2}z_{i}}\right)  \label{ppsib}%
\end{equation}
with the asymptotic behavior%
\begin{equation}
p^{\psi}(\underline{\theta}_{W},\underline{\omega}_{W},\underline{z}%
_{W})\rightarrow e^{-\frac{1}{2}W\left(  \left(  1-\frac{1}{N}\right)
k+\frac{1}{N}\bar{k}-k_{1}\right)  }p^{\psi}(\underline{\hat{\theta}%
},\underline{\hat{\omega}},\underline{\hat{z}})p^{\psi}(\underline
{\check{\theta}},\underline{\check{\omega}},\underline{\check{z}}).
\label{ppsiab}%
\end{equation}
In particular for $n=\bar{n}=2$ and $k=\bar{k}=1$ we prove the proposition:

\begin{proposition}
\label{p3}The form factor of the current for $n=2,~\bar{n}=1$ and $k=\bar
{k}=1$ satisfies the clustering formula (\ref{psiJpsi}) in the form%
\begin{equation}
F_{\underline{\alpha}\bar{\delta}}^{\psi^{\beta}}(\underline{\theta}%
_{W},\omega_{W})\rightarrow i\eta W^{-1}\left(  \mathbf{C}_{\gamma\bar{\kappa
}}F_{\hat{\alpha}\hat{\bar{\delta}}}^{J^{\beta\bar{\kappa}}}(\hat{\theta}%
,\hat{\omega})F_{\check{\alpha}}^{\psi^{\gamma}}(\check{\theta})\right)  .
\label{FpsiJpsi}%
\end{equation}

\end{proposition}

\begin{proof}
The exponential behavior (\ref{exppsi}) implies for $n=\bar{n}=2$ and
$k=\bar{k}=1$ that $k_{j}=1$ and $l_{j}=0$ for $j=1,\dots,N-1.$ We investigate
for $\psi=\psi^{1}$
\begin{equation}
K_{\underline{\alpha}\bar{\delta}}^{\psi}(\underline{\theta},\omega)=\int
d\underline{\underline{z}}\tilde{h}\,(\underline{\theta},\underline{\omega
},\underline{\underline{z}})p^{\psi}(\underline{\theta},\underline{\omega
},\underline{\underline{z}})\tilde{\Phi}_{\underline{\alpha}\bar{\delta}%
}(\underline{\theta},\omega,\underline{\underline{z}})\, \label{Kpsi}%
\end{equation}
We have proved in theorem \ref{t3} that in leading order
\[
F_{\underline{\alpha}}^{\psi}(\underline{\theta}_{W})\rightarrow0
\]

\noindent\textbf{Order }$\frac{1}{W}$: we have to apply the asymptotic
behavior of the h-function (\ref{ha}) and the Bethe state (\ref{PHIa}).

The \textbf{result for the contribution of }$h_{1}$ is%
\[
F_{\underline{\alpha}\bar{\delta},h_{1}}^{\psi^{1}}(\underline{\theta}%
_{W},\omega_{W})\rightarrow i\eta W^{-1}\left(  \mathbf{C}_{1\bar{1}}%
F_{\hat{\alpha}{\hat{\bar{\delta}}}}^{J^{1\bar{1}}}(\hat{\theta},\hat{\omega
})\right)  F_{\check{\alpha}}^{\psi^{1}}(\check{\theta})
\]
and the \textbf{result for the contribution of from }$\Phi_{1}$ is%
\[
F_{\underline{\alpha}\bar{\delta},\tilde{\Phi}_{1}}^{\psi^{1}}(\underline
{\theta}_{W},\omega_{W})\rightarrow i\eta W^{-1}\left(  \mathbf{C}_{\gamma
\bar{\kappa}}^{(1)}F_{\hat{\alpha}{\hat{\bar{\delta}}}}^{J^{1\bar{\kappa}}%
}(\hat{\theta},\hat{\omega})\right)  F_{\check{\alpha}}^{\psi^{\gamma}}%
(\check{\theta})
\]
Because $\mathbf{C}_{\delta\bar{\kappa}}^{(1)}+\mathbf{C}_{1\bar{1}}%
\delta_{\delta}^{1}\delta_{\bar{\kappa}}^{\bar{1}}=\mathbf{C}_{\delta
\bar{\kappa}}$ (see (\ref{r4a})) the relation (\ref{FpsiJpsi}) is proved.
\end{proof}

\subsection{Formulas}

\begin{definition}
We define (for $0\leq j\leq N-2$) iteratively%
\begin{align}
L_{xy}^{(j)}(z,\omega)  &  =\int du\tilde{\phi}(z-u)L_{ca}^{(j+1)}%
(u,\omega)\tilde{x}(z-u)\tilde{y}(\omega-z)\label{d1}\\
L_{uxy}^{(j)}(z,\omega)  &  =\int du\tilde{\phi}(z-u)L_{ca}^{(j+1)}%
(u,\omega)u\tilde{x}(z-u)\tilde{y}(\omega-z) \label{d2}%
\end{align}
with%
\begin{align*}
\tilde{x}(z),\tilde{y}(z)  &  =\tilde{a}=1,~\tilde{b}(z)=\frac{z}{z-i\eta
},~\tilde{c}(z)=\frac{-i\eta}{z-i\eta},~\tilde{d}(z)=\frac{-i\eta}{i\pi
-z},~\eta=\frac{2\pi}{N}\\
L_{ca}^{(N-1)}(z,\omega)  &  =(-1)^{N-1}\tilde{\chi}_{N-1}(\omega-z)
\end{align*}

\end{definition}

\begin{proposition}
~

\begin{enumerate}
\item If $\tilde{\chi}_{N-1}(\omega)=\tilde{\chi}(\omega)=\Gamma\left(
\frac{1}{2}+\frac{\omega}{2\pi i}\right)  \Gamma\left(  \frac{1}{2}-\frac
{1}{N}-\frac{\omega}{2\pi i}\right)  $ then
\[
L_{ca}^{(j)}(z,\omega)=(-1)^{N-1}c_{N-2}\cdots c_{j}\tilde{\chi}_{j}%
(\omega-z)
\]
with%
\begin{align}
\tilde{\chi}_{j}(\omega)  &  =\Gamma\left(  -\tfrac{1}{2}+j/N-\omega/(2\pi
i)\right)  \Gamma\left(  \tfrac{1}{2}+\omega/(2\pi i)\right) \label{r1}\\
c_{j}  &  =4\pi^{2}\frac{\Gamma\left(  1-\frac{1}{N}\right)  \Gamma\left(
\frac{j+1}{N}\right)  }{\Gamma\left(  \frac{1}{N}j\right)  }%
,~0<j<N-1\nonumber\\
c_{N-2}\cdots c_{j}  &  =(4\pi^{2})^{N-1-j}\frac{\left(  \Gamma\left(
1-\frac{1}{N}\right)  \right)  ^{N-j}}{\Gamma\left(  \frac{1}{N}j\right)
}\nonumber
\end{align}%
\begin{equation}
c_{0}=0\Rightarrow L_{ca}^{(0)}(z,\omega)=0. \label{r1a}%
\end{equation}

\item
\begin{align}
L_{bd}^{(j)}(z,\omega)  &  =L_{ca}^{(j)}(z,\omega)/(N-j-1)\label{r2}\\
L_{aa}^{(j)}(z,\omega)  &  =L_{ca}^{(j)}(z,\omega)\left(  1+N\left(
z-\omega-i\pi\right)  /(2i\pi j)\right) \label{r2a}\\
K_{aa}(\theta,\omega)  &  =L_{aa}^{(0)}(\theta,\omega)=(-1)^{N-1}c_{N-2}\cdots
c_{1}\frac{4\pi^{4}}{\sin\frac{\pi}{N}}\frac{1}{\cosh\frac{1}{2}\left(
\theta-\omega\right)  }. \label{K0}%
\end{align}

\item
\begin{align}
L_{uca}^{(j)}(z,\omega)  &  =\frac{1}{j}\left(  \left(  1+j\right)
z-\omega-i\pi\right)  L_{ca}^{(j)}(z,\omega)\label{r3}\\
L_{ubd}^{(j)}(z,\omega)  &  =-\left(  \frac{1}{j}\left(  z-\omega-i\pi\right)
+\frac{1}{N-j-1}\left(  i\pi-\omega\right)  \right)  L_{ca}^{(j)}(z,\omega)
\label{r3a}%
\end{align}
in particular
\begin{align}
L_{uca}^{(0)}(\theta,\omega)  &  =\frac{2i\pi}{N}K_{aa}(\theta,\omega
)\label{r3b}\\
L_{ubd}^{(0)}(\theta,\omega)  &  =-\frac{2i\pi}{N}K_{aa}(\theta,\omega).
\label{r3c}%
\end{align}

\item If $L_{\beta^{\prime}(\mu^{\prime})}^{(N-1)}(z,\omega)=\mathbf{C}%
_{\beta(\mu)}^{(N-1)}L_{ca}^{(j)}(z,\omega)=\delta_{\beta}^{N}\delta_{(\mu
)}^{(1..N-1)}(-1)^{N-1}\tilde{\chi}(\omega-z)$ then%
\begin{equation}
L_{\beta(\mu)}^{(j)}(z,\omega)=\int_{\mathcal{C}}du\tilde{\phi}\left(
z-u\right)  L_{\beta^{\prime}(\mu^{\prime})}^{(j+1)}(u,\omega)\left(
T_{\beta(\mu),j+1}^{\beta^{\prime},j+1(\mu^{\prime})}(z,\omega,u)\right)
=\mathbf{C}_{\beta(\mu)}^{(j)}L_{ca}^{(j)}(z,\omega) \label{r4}%
\end{equation}
where
\begin{equation}
\mathbf{C}_{\beta(\mu)}^{(j)}=\mathbf{C}_{\beta(\mu)}~\text{for }%
\beta>j~\text{else }=0. \label{r4a}%
\end{equation}

\item
\begin{align}
L_{u\beta(\mu)}^{(j)}(z,\omega)  &  =\int_{\mathcal{C}}du\tilde{\phi}\left(
z-u\right)  L_{u\beta^{\prime}(\mu^{\prime})}^{(j+1)}(u,\omega)u\left(
T_{\beta(\mu),j+1}^{\beta^{\prime},j+1(\mu^{\prime})}(z,\omega,u)\right)
\label{r5}\\
&  =\frac{1}{j}\left(  \left(  N\delta_{\beta}^{N}\delta_{(\mu)}^{\bar{N}%
}-\mathbf{C}_{\beta(\mu)}^{(j)}\right)  \right)  zL_{ca}^{(j)}(z,\omega
)+const.L_{ca}^{(j)}(z,\omega)\nonumber
\end{align}
in particular%
\[
L_{u\beta(\mu)}^{(N-1)}(z,\omega)=\delta_{\beta}^{N}\delta_{(\mu)}^{\bar{N}%
}zL_{ca}^{(N-1)}(z,\omega).
\]

\item for $j+1<\alpha^{\prime}<N$%
\begin{equation}
\int_{\mathcal{C}}du\tilde{\phi}\left(  z-u\right)  L_{aa}^{(j+1)}%
(u,\omega)\left(  T_{\alpha(\rho)j+1}^{\alpha^{\prime},j+1\bar{N}}%
(z,\omega,u)\right)  =\delta_{\alpha}^{\alpha^{\prime}}\delta_{(\rho)}%
^{\bar{N}}L_{aa}^{(j)}(z,\omega). \label{r6}%
\end{equation}

\end{enumerate}
\end{proposition}

\begin{proof}
We use%
\begin{multline*}
\frac{1}{2\pi i}\left(  \int_{\mathcal{C}_{a}}+\int_{\mathcal{C}_{b}}\right)
dz\Gamma(a-z)\Gamma(b-z)\Gamma(c+z)\Gamma(d+z)\\
=-\frac{\Gamma\left(  c+a\right)  \Gamma\left(  d+a\right)  \Gamma\left(
c+b\right)  \Gamma\left(  d+b\right)  }{\Gamma\left(  c+d+a+b\right)  }%
\end{multline*}
where $\mathcal{C}_{a}$ encircles the poles of $\Gamma(a-z)$ clockwise.

\noindent1. With $\tilde{\phi}(z)\tilde{c}(z)=-\frac{1}{N}\Gamma\left(
-\frac{z}{2\pi i}\right)  \Gamma\left(  -\frac{1}{N}+\frac{z}{2\pi i}\right)
$ and $\tilde{\chi}_{j+1}(\omega)$ of (\ref{r1}) follows%
\begin{align*}
&  \int_{\mathcal{C}_{\underline{\theta}}}du\tilde{\phi}(z-u)\tilde
{c}(z-u)\tilde{\chi}_{j+1}(\omega-u)\\
&  =-\frac{1}{N}\int_{\mathcal{C}_{\underline{\theta}}}du\Gamma\left(
-\tfrac{z-u}{2\pi i}\right)  \Gamma\left(  -\tfrac{1}{N}+\tfrac{z-u}{2\pi
i}\right)  \Gamma\left(  -\tfrac{1}{2}+\tfrac{j+1}{N}-\tfrac{\omega-u}{2\pi
i}\right)  \Gamma\left(  \tfrac{1}{2}+\tfrac{\omega-u}{2\pi i}\right) \\
&  =-\frac{4}{N}\pi^{2}\Gamma\left(  -\frac{1}{N}\right)  \frac{\Gamma\left(
\frac{j+1}{N}\right)  }{\Gamma\left(  \frac{1}{N}j\right)  }\Gamma\left(
\frac{1}{2}-\frac{1}{2}\frac{z-\omega}{i\pi}\right)  \Gamma\left(  -\frac
{1}{2}+\frac{1}{N}j+\frac{1}{2}\frac{z-\omega}{i\pi}\right) \\
&  =c_{j}\tilde{\chi}_{j}(\omega-u)
\end{align*}
and iterating this result $\Rightarrow$ 1.

\noindent2. With $\tilde{\phi}(z)\tilde{b}(z)=-\Gamma\left(  1-\frac{z}{2\pi
i}\right)  \Gamma\left(  -\frac{1}{N}+\frac{z}{2\pi i}\right)  $ and

$\tilde{d}(\omega)\tilde{\chi}_{j+1}(\omega)=\frac{1}{N}\Gamma\left(
-\frac{1}{2}+\frac{j+1}{N}-\frac{\omega}{2\pi i}\right)  \Gamma\left(
-\frac{1}{2}+\frac{\omega}{2\pi i}\right)  \Rightarrow$%
\begin{align*}
&  \int_{\mathcal{C}_{\underline{\theta}}}du\tilde{\phi}(z-u)\tilde
{b}(z-u)\tilde{d}(\omega-u)\tilde{\chi}_{j+1}(\omega-u)\\
&  =-\frac{1}{N}\int_{\mathcal{C}_{\underline{\theta}}}du\Gamma\left(
-\tfrac{z-u}{2\pi i}\right)  \Gamma\left(  -\tfrac{1}{N}+\tfrac{z-u}{2\pi
i}\right)  \Gamma\left(  -\tfrac{1}{2}+\tfrac{j+1}{N}-\tfrac{\omega-u}{2\pi
i}\right)  \Gamma\left(  -\tfrac{1}{2}+\tfrac{\omega-u}{2\pi i}\right) \\
&  =\frac{1}{N-j-1}c_{j}\tilde{\chi}_{j}(\omega-u)
\end{align*}
and%
\begin{align*}
&  \int_{\mathcal{C}_{\underline{\theta}}}du\tilde{\phi}(z-u)\tilde{\chi
}_{j+1}(\omega-u)\\
&  =\int_{\mathcal{C}_{\underline{\theta}}}du\Gamma\left(  -\tfrac{z-u}{2\pi
i}\right)  \Gamma\left(  1-\tfrac{1}{N}+\tfrac{z-u}{2\pi i}\right)
\Gamma\left(  -\tfrac{1}{2}+\tfrac{j+1}{N}-\tfrac{\omega-u}{2\pi i}\right)
\Gamma\left(  \tfrac{1}{2}+\tfrac{\omega-u}{2\pi i}\right) \\
&  =c_{j}\tilde{\chi}_{j}(\omega-u)\left(  1+\frac{N}{2i\pi j}\left(
z-\omega-i\pi\right)  \right)
\end{align*}
and $\tilde{\chi}_{0}(\omega)=\Gamma\left(  -\frac{1}{2}-\frac{\omega}{2\pi
i}\right)  \Gamma\left(  \frac{1}{2}+\frac{\omega}{2\pi i}\right)
=\frac{-2i\pi^{2}}{\left(  i\pi+\omega\right)  \cosh\frac{1}{2}\omega
}\Rightarrow$ 2.

\noindent3. With $\tilde{c}(z-u)u=\frac{2i\pi}{N}+\tilde{c}(z-u)\left(
z-\frac{2i\pi}{N}\right)  \Rightarrow$%
\begin{align*}
L_{uca}^{(j)}(z,\omega)  &  =2i\pi/N\,L_{aa}^{(j)}(z,\omega)+\left(
z-2i\pi/N\right)  L_{ca}^{(j)}(z,\omega)\\
&  =\left(  2i\pi\left(  1+N\left(  z-\omega-i\pi\right)  /(2i\pi j)\right)
/N+\left(  z-2i\pi/N\right)  \right)  L_{ca}^{(j)}(z,\omega)\\
&  =\left(  \left(  1+j\right)  z-\omega-i\pi\right)  /j\,L_{ca}%
^{(j)}(z,\omega).
\end{align*}
and using $\tilde{b}=1-\tilde{c},~\tilde{d}(\omega-u)u=-\frac{2i\pi}{N}%
-\tilde{d}(\omega-u)\left(  i\pi-\omega\right)  \Rightarrow$%
\begin{align*}
L_{ubd}^{(j)}(z,\omega)  &  =-2i\pi/N\,L_{ba}^{(j)}-\left(  i\pi
-\omega\right)  L_{bd}^{(j)}\\
&  =-\left(  2i\pi/N\,L_{aa}^{(j)}-L_{ca}^{(j)}\right)  -\left(  i\pi
-\omega\right)  /(N-j-1)\,L_{ca}^{(j)}\\
&  =-\left(  \left(  z-\omega-i\pi\right)  /j+\left(  i\pi-\omega\right)
/(N-j-1)\right)  L_{ca}^{(j)}(z,\omega).
\end{align*}
By (\ref{r2a})
\begin{align*}
L_{uca}^{(0)}(z,\omega)  &  =\lim_{j\rightarrow0}\frac{\frac{1}{j}\left(
\left(  1+j\right)  z-\omega-i\pi\right)  }{1+\frac{N}{2i\pi j}\left(
z-\omega-i\pi\right)  }L_{aa}^{(0)}(z,\omega)=\frac{2i\pi}{N}K_{aa}%
(z,\omega)\\
L_{ubd}^{(0)}(z,\omega)  &  =\lim_{j\rightarrow0}\frac{-\left(  \frac{1}%
{j}\left(  z-\omega-i\pi\right)  +\frac{1}{N-j-1}\left(  i\pi-\omega\right)
\right)  }{1+\frac{N}{2i\pi j}\left(  z-\omega-i\pi\right)  }L_{aa}%
^{(0)}(z,\omega)=-\frac{2i\pi}{N}K_{aa}(z,\omega).
\end{align*}
$\Rightarrow$ 3.

\noindent4.
\begin{align*}
&  \mathbf{C}_{\beta^{\prime}(\mu^{\prime})}^{(j+1)}T_{\beta(\mu),j+1}%
^{\beta^{\prime},j+1(\mu^{\prime})}\\
&  =\mathbf{C}_{\beta^{\prime}(\mu^{\prime})}^{(j+1)}\left(  \tilde{b}%
\delta_{j+1}^{\beta^{\prime}}\delta_{\beta}^{j+1}\delta_{(\mu)}^{(\mu^{\prime
})}+\tilde{c}\delta_{\beta}^{\beta^{\prime}}\delta_{(\mu)}^{(\mu^{\prime}%
)}+\tilde{b}\tilde{d}\delta_{\beta}^{j+1}\mathbf{C}^{\beta^{\prime}%
(\mu^{\prime})}\mathbf{C}_{(\mu)j+1}+\tilde{c}\tilde{d}\delta_{\beta}%
^{\beta^{\prime}}\mathbf{C}^{j+1(\mu^{\prime})}\mathbf{C}_{(\mu)j+1}\right) \\
&  =\tilde{b}\delta_{\beta}^{j+1}\mathbf{C}_{j+1(\mu)}^{(j+1)}+\tilde
{c}\mathbf{C}_{\beta(\mu)}^{(j+1)}+\tilde{b}\tilde{d}\mathbf{C}_{\beta
(\mu^{\prime})}^{(j+1)}\mathbf{C}^{j+1(\mu^{\prime})}\mathbf{C}_{(\mu
)j+1}+\tilde{c}\tilde{d}\delta_{\beta}^{j+1}\mathbf{C}_{(\mu)j+1}^{(j+1)}%
\end{align*}
$\mathbf{C}_{(\mu)j+1}^{(j+1)}=0$ and $\mathbf{C}_{\beta(\mu^{\prime}%
)}^{(j+1)}\mathbf{C}^{j+1(\mu^{\prime})}\mathbf{C}_{(\mu)j+1}=(N-j-1)\delta
_{\beta}^{j+1}\mathbf{C}_{j+1(\mu)}\Rightarrow$%
\begin{align*}
L_{\beta(\mu)}^{(j)}(z,\omega)  &  =\int_{\mathcal{C}}du\tilde{\phi}\left(
z-u\right)  L_{ca}^{(j+1)}(u,\omega)\mathbf{C}_{\beta^{\prime}(\mu^{\prime}%
)}^{(j+1)}\left(  T_{\beta(\mu),j+1}^{\beta^{\prime},j+1(\mu^{\prime}%
)}(z,\omega,u\right) \\
&  =L_{ca}^{(j)}\mathbf{C}_{\beta(\mu)}^{(j+1)}+L_{bd}^{(j)}(N-j-1)\delta
_{\beta}^{j+1}\mathbf{C}_{j+1(\mu)}\\
&  =L_{ca}^{(j)}\left(  \mathbf{C}_{\beta(\mu)}^{(j+1)}+\delta_{\beta}%
^{j+1}\mathbf{C}_{j+1(\mu)}\right)  =L_{ca}^{(j)}\mathbf{C}_{\beta(\mu)}^{(j)}%
\end{align*}

\noindent5. By induction: let $L_{u\beta(\mu)}^{(N-1)}(z,\omega)=\delta
_{\beta}^{N}\delta_{(\mu)}^{\bar{N}}zL_{ca}^{(N-1)}(z,\omega)$ and

$L_{u\beta^{\prime}(\mu^{\prime})}^{(j+1)}(u,\omega)=\frac{1}{j+1}\left(
\left(  N\delta_{\beta^{\prime}}^{N}\delta_{(\mu^{\prime})}^{(1..N-1)}%
-\mathbf{C}_{\beta^{\prime}(\mu^{\prime})}^{(j+1)}\right)  \right)
L_{ca}^{(j+1)}\Rightarrow$%
\begin{multline*}
L_{u\beta(\mu)}^{(j)}(z,\omega)=\int_{\mathcal{C}}du\tilde{\phi}\left(
z-u\right)  L_{u\beta^{\prime}(\mu^{\prime})}^{(j+1)}(u,\omega)u\left(
T_{\beta(\mu),j+1}^{\beta^{\prime},j+1(\mu^{\prime})}(z,\omega,u\right) \\
=\frac{1}{j+1}\int_{\mathcal{C}}du\tilde{\phi}\left(  z-u\right)  \left(
\left(  N\delta_{\beta^{\prime}}^{N}\delta_{(\mu^{\prime})}^{(1..N-1)}%
-\mathbf{C}_{\beta^{\prime}(\mu^{\prime})}^{(j+1)}\right)  L_{ca}%
^{(j+1)}(u,\omega)\right) \\
\times u\left(  \tilde{c}\delta_{\beta}^{\beta^{\prime}}\delta_{(\mu)}%
^{(\mu^{\prime})}+\tilde{b}\tilde{d}\delta_{\beta}^{j+1}\mathbf{C}%
^{\beta^{\prime}(\mu^{\prime})}\mathbf{C}_{(\mu)j+1}\right) \\
=\frac{1}{j+1}\left(  L_{uca}^{(j)}(z,\omega)\left(  N\delta_{\beta}^{N}%
\delta_{(\mu)}^{(1..N-1)}-\mathbf{C}_{\beta(\mu)}^{(j+1)}\right)
+L_{ubd}^{(j)}(z,\omega)\delta_{\beta}^{j+1}\left(  j+1\right)  \mathbf{C}%
_{j+1(\mu)}\right) \\
=\frac{L_{ca}^{(j)}(z,\omega)}{j+1}\left(  \left(  \frac{1}{j}\left(
1+j\right)  \left(  N\delta_{\beta}^{N}\delta_{(\mu)}^{(1..N-1)}%
-\mathbf{C}_{\beta(\mu)}^{(j+1)}\right)  -\frac{1}{j}\delta_{\beta}%
^{j+1}\left(  j+1\right)  \mathbf{C}_{j+1(\mu)}\right)  z+const\right) \\
=\left(  N\delta_{\beta}^{N}\delta_{(\mu)}^{(1..N-1)}-\mathbf{C}_{\beta(\mu
)}^{(j)}\right)  /j\,zL_{ca}^{(j)}(z,\omega)+constL_{ca}^{(j)}.
\end{multline*}
$\Rightarrow$ 5. because

$( N\delta_{\beta^{\prime}}^{N}\delta_{(\mu^{\prime})}^{(1..N-1)}%
\mathbf{C}^{\beta^{\prime}(\mu^{\prime})}-\mathbf{C}_{\beta^{\prime}%
(\mu^{\prime})}^{(j+1)}\mathbf{C}^{\beta^{\prime}(\mu^{\prime})})
\mathbf{C}_{(\mu)j+1}$\newline$~\hfill=\left(  N\mathbf{C}_{j+1(\mu
)}-(N-j-1)\mathbf{C}_{j+1(\mu)}\right)  =\left(  j+1\right)  \mathbf{C}%
_{j+1(\mu)}$.

\noindent6. $T_{\alpha(\rho)j+1}^{\alpha^{\prime},j+1\bar{N}}=\tilde{c}%
\delta_{\alpha}^{\alpha^{\prime}}\delta_{(\rho)}^{\bar{N}}$ holds for
$j+1<\alpha^{\prime}<N$ and with (\ref{r2a}) for $j+1$ follows%
\begin{align*}
&  \int_{\mathcal{C}}du\tilde{\phi}\left(  z-u\right)  \tilde{c}\left(
z-u\right)  L_{aa}^{(j+1)}(u,\omega)\\
&  =\int_{\mathcal{C}}du\tilde{\phi}\left(  z-u\right)  \tilde{c}\left(
z-u\right)  L_{ca}^{(j+1)}(u,\omega)\left(  1+N\left(  u-\omega-i\pi\right)
/(2i\pi\left(  j+1\right)  )\right) \\
&  =L_{ca}^{(j)}(z,\omega)\left(  1+N\left(  -\omega-i\pi\right)
/(2i\pi\left(  j+1\right)  )\right)  +L_{uca}^{(j)}N/(2i\pi\left(  j+1\right)
)\\
&  =L_{ca}^{(j)}(z,\omega)\left(  1+\frac{N\left(  -\omega-i\pi\right)
}{2i\pi\left(  j+1\right)  }+\frac{N}{2i\pi\left(  j+1\right)  }\frac{1}%
{j}\left(  \left(  1+j\right)  z-\omega-i\pi\right)  \right)  =L_{aa}%
^{(j)}(z,\omega)
\end{align*}
by (\ref{r2a}) $\Rightarrow$ 6.
\end{proof}

\setcounter{equation}{0}

\section{The functions $c(k,l,W)$}

\label{ac}

The functions $c_{\hat{\mathcal{O}}\mathcal{\check{O}}}^{\mathcal{O}}(k,l,W)$
in (\ref{FFF})
\[
F_{\underline{\alpha}}^{\mathcal{O}}(\underline{\theta}_{W})\rightarrow
c_{\hat{\mathcal{O}}\mathcal{\check{O}}}^{\mathcal{O}}(k,l,W)F_{\underline
{\hat{\alpha}}}^{\hat{\mathcal{O}}}(\underline{\hat{\theta}})F_{\underline
{\check{\alpha}}}^{\mathcal{\check{O}}}(\underline{\check{\theta}})
\]
are calculated using the form factor equation (iii) (see (\ref{1.14})), by
taking for $F_{\underline{\alpha}}^{\mathcal{O}}(\underline{\theta}_{W})$
first the $\operatorname*{Res}$ and then the limit $W\rightarrow\infty$ or
exchanging the procedures. We use two special cases of the form factor
equation (iii) :

\begin{enumerate}
\item[I] As in (\ref{bs}) we take the bound state $\left(  \alpha_{1}%
\dots\alpha_{N-1}\right)  =(2\dots N)=\bar{1}$ with rapidity $\omega$ and
(iii) reads as%
\begin{equation}
\operatorname*{Res}_{\omega=i\pi+\theta}F_{\bar{1}1\underline{\hat{\alpha}%
}\underline{\check{\alpha}}}^{\mathcal{O}}(\omega,\theta,\underline
{\hat{\theta}},\underline{\check{\theta}})=2i\,\mathbf{C}_{\bar{1}%
1}\,F_{\underline{\hat{\alpha}}^{\prime}\underline{\check{\alpha}}^{\prime}%
}^{\mathcal{O}}(\underline{\hat{\theta}},\underline{\check{\theta}})\left(
1_{\underline{\hat{\alpha}}\underline{\check{\alpha}}}^{\underline{\hat
{\alpha}}^{\prime}\underline{\check{\alpha}}^{\prime}}-\dot{\sigma}%
_{1}^{\mathcal{O}}(n)S_{1\underline{\hat{\alpha}}}^{\underline{\hat{\alpha}%
}^{\prime}1}(\theta,\underline{\hat{\theta}})S_{1\underline{\check{\alpha}}%
}^{\underline{\check{\alpha}}^{\prime}1}(\theta,\underline{\check{\theta}%
})\right)  \label{iii1}%
\end{equation}
where we use the short notation of (\ref{sn}) and (\ref{st}) for the
statistics factor $\dot{\sigma}_{1}^{\mathcal{O}}(n)$.

\item[II] The form factor equation (i)%
\[
F_{\underline{\hat{\alpha}}\bar{\beta}\gamma\underline{\check{\alpha}}}^{\psi
}(\underline{\hat{\theta}},\omega,\theta,\underline{\check{\theta}}%
)=F_{\bar{\beta}^{\prime}\gamma^{\prime}\underline{\hat{\alpha}}^{\prime
\prime}\underline{\check{\alpha}}}^{\psi}(\omega,\theta,\underline{\hat
{\theta}},\underline{\check{\theta}})S_{\underline{\hat{\alpha}}^{\prime
}\gamma}^{\gamma^{\prime}\underline{\hat{\alpha}}^{\prime\prime}}%
(\underline{\hat{\theta}},\theta)S_{\underline{\hat{\alpha}}\bar{\beta}}%
^{\bar{\beta}^{\prime}\underline{\hat{\alpha}}^{\prime}}(\underline
{\hat{\theta}},\omega)
\]
implies that%
\begin{align}
\operatorname*{Res}_{\omega=i\pi+\theta}F_{\underline{\hat{\alpha}}\bar
{1}1\underline{\check{\alpha}}}^{\mathcal{O}}(\underline{\hat{\theta}}%
,\omega,\theta,\underline{\check{\theta}})  &  =2i\mathbf{C}_{\bar{1}%
1}F_{\underline{\hat{\alpha}}^{\prime}\underline{\check{\alpha}}^{\prime}%
}^{\mathcal{O}}(\underline{\hat{\theta}},\underline{\check{\theta}%
})\label{iii2}\\
&  \times\left(  (-1)^{\left(  N-1\right)  k}1_{\underline{\hat{\alpha}%
}\underline{\check{\alpha}}}^{\underline{\hat{\alpha}}^{\prime}\underline
{\check{\alpha}}^{\prime}}-\dot{\sigma}_{1}^{\mathcal{O}}(n)S_{\underline
{\hat{\alpha}}\bar{1}}^{\bar{1}\underline{\hat{\alpha}}^{\prime}}%
(\underline{\hat{\theta}},\omega)S_{1\underline{\check{\alpha}}}%
^{\underline{\check{\alpha}}^{\prime}1}(\theta,\underline{\check{\theta}%
})\right) \nonumber
\end{align}
where $k=\left\vert \underline{\hat{\alpha}}\right\vert $. It has been used
that crossing \cite{BFK1} $S_{\alpha\bar{1}}^{\bar{1}\beta}(\theta
)=(-1)^{\left(  N-1\right)  }S_{1\alpha}^{\beta1}(i\pi-\theta)$ implies
\[
S_{\alpha^{\prime}1}^{1\alpha^{\prime\prime}}(\theta)S_{\alpha\bar{1}}%
^{\bar{1}\alpha^{\prime}}(\theta-i\pi)=(-1)^{\left(  N-1\right)  }%
S_{\alpha^{\prime}1}^{1\alpha^{\prime\prime}}(\theta)S_{1\alpha}%
^{\alpha^{\prime}1}(-\theta)=(-1)^{\left(  N-1\right)  }1_{\alpha}%
^{\alpha^{\prime\prime}}.
\]

\end{enumerate}

We consider 4 procedures:

\begin{enumerate}
\item Let $\underline{\theta}_{W}=(\omega+W,\theta+W,\underline{\hat{\theta}%
}+W,\underline{\check{\theta}}),~k=N+\left\vert \underline{\hat{\alpha}%
}\right\vert >N,~l=\left\vert \underline{\check{\alpha}}\right\vert $ then by
(\ref{iii1}) and (\ref{FFF})
\begin{align}
\operatorname*{Res}_{\omega=i\pi+\theta}F_{\bar{1}1\underline{\hat{\alpha}%
}\underline{\check{\alpha}}}^{\mathcal{O}}(\underline{\theta}_{W})  &
=2i\mathbf{C}_{\bar{1}1}F_{\underline{\hat{\alpha}}^{\prime}\underline
{\check{\alpha}}^{\prime}}^{\mathcal{O}}(\underline{\hat{\theta}}%
+W,\underline{\check{\theta}})\left(  1_{\underline{\hat{\alpha}}%
\underline{\check{\alpha}}}^{\underline{\hat{\alpha}}^{\prime}\underline
{\check{\alpha}}^{\prime}}-\dot{\sigma}_{1}^{\mathcal{O}}(n)S_{1\underline
{\hat{\alpha}}}^{\underline{\hat{\alpha}}^{\prime}1}(\theta,\underline
{\hat{\theta}})S_{1\underline{\check{\alpha}}}^{\underline{\check{\alpha}%
}^{\prime}1}(\theta+W,\underline{\check{\theta}})\right) \label{c1}\\
&  \overset{W\rightarrow\infty}{\rightarrow}2ic_{\hat{\mathcal{O}%
}\mathcal{\check{O}}}^{\mathcal{O}}(k-N,l,W)F_{\underline{\hat{\alpha}%
}^{\prime}}^{\hat{\mathcal{O}}}(\underline{\hat{\theta}})\left(
1_{\underline{\hat{\alpha}}}^{\underline{\hat{\alpha}}^{\prime}}-\dot{\sigma
}_{1}^{\hat{\mathcal{O}}}(k)S_{1\underline{\hat{\alpha}}}^{\underline
{\hat{\alpha}}^{\prime}1}(\theta,\underline{\hat{\theta}})\right)
F_{\underline{\check{\alpha}}}^{\mathcal{\check{O}}}(\underline{\check{\theta
}})\nonumber
\end{align}
if
\begin{equation}
\dot{\sigma}_{1}^{\mathcal{O}}(n)S_{1\underline{\check{\alpha}}}%
^{\underline{\check{\alpha}}^{\prime}1}(\theta+W,\underline{\check{\theta}%
})\overset{W\rightarrow\infty}{\rightarrow}\dot{\sigma}_{1}^{\hat{\mathcal{O}%
}}(k)1_{\underline{\check{\alpha}}}^{\underline{\check{\alpha}}^{\prime}}\,.
\label{tocheck}%
\end{equation}

\item Inverting the procedures%
\begin{multline}
\operatorname*{Res}_{\omega=i\pi+\theta}\left\{  F_{\bar{1}1\underline
{\hat{\alpha}}\underline{\check{\alpha}}}^{\mathcal{O}}(\underline{\theta}%
_{W})\overset{W\rightarrow\infty}{\rightarrow}c_{\hat{\mathcal{O}%
}\mathcal{\check{O}}}^{\mathcal{O}}\left(  k,l,W\right)  F_{\bar{1}%
1\underline{\hat{\alpha}}}^{\hat{\mathcal{O}}}(\omega,\theta,\underline
{\hat{\theta}})F_{\underline{\check{\alpha}}}^{\mathcal{\check{O}}}%
(\underline{\check{\theta}})\right\} \label{c2}\\
=c_{\hat{\mathcal{O}}\mathcal{\check{O}}}^{\mathcal{O}}\left(  k,l,W\right)
2iF_{\underline{\hat{\alpha}}^{\prime}}^{\hat{\mathcal{O}}}(\underline
{\hat{\theta}})\mathbf{C}_{\bar{1}1}\left(  1_{\underline{\hat{\alpha}}%
}^{\underline{\hat{\alpha}}^{\prime}}-\dot{\sigma}_{1}^{\hat{\mathcal{O}}%
}(k)S_{1\underline{\hat{\alpha}}}^{\underline{\hat{\alpha}}^{\prime}1}\right)
F_{\underline{\check{\alpha}}}^{\mathcal{\check{O}}}(\underline{\check{\theta
}})
\end{multline}

\item Let $\underline{\theta}_{W}=(\underline{\hat{\theta}}+W,\omega
,\theta,\underline{\check{\theta}}),~k=\left\vert \underline{\hat{\alpha}%
}\right\vert ,~l=N+\left\vert \underline{\check{\alpha}}\right\vert >N$ then
by (\ref{iii2}) and (\ref{FFF})
\begin{multline}
\operatorname*{Res}_{\omega=i\pi+\theta}F_{\underline{\hat{\alpha}}\bar
{1}1\underline{\check{\alpha}}}^{\mathcal{O}}(\underline{\theta}%
_{W})=2i\,\mathbf{C}_{1\bar{1}}\,F_{\underline{\hat{\alpha}}^{\prime
}\underline{\check{\alpha}}^{\prime}}^{\mathcal{O}}(\underline{\hat{\theta}%
}+W,\underline{\check{\theta}})\\
\times\left(  (-1)^{\left(  N-1\right)  k}1_{\hat{\alpha}}^{\hat{\alpha
}^{\prime}}1_{\underline{\check{\alpha}}}^{\underline{\check{\alpha}}^{\prime
}}-\dot{\sigma}_{1}^{\mathcal{O}}(n)S_{\underline{\hat{\alpha}}\bar{1}}%
^{\bar{1}\hat{\alpha}^{\prime}}(\underline{\hat{\theta}}+W,\omega
)S_{1\underline{\check{\alpha}}}^{\underline{\check{\alpha}}^{\prime}1}%
(\theta,\underline{\check{\theta}})\right) \\
\overset{W\rightarrow\infty}{\rightarrow}2i\mathbf{C}_{\bar{1}1}%
c_{\hat{\mathcal{O}}\mathcal{\check{O}}}^{\mathcal{O}}(k,l-N,W)F_{\hat{\alpha
}}^{\hat{\mathcal{O}}}(\hat{\theta})F_{\underline{\check{\alpha}}^{\prime}%
}^{\mathcal{\check{O}}}(\underline{\check{\theta}})\,(-1)^{\left(  N-1\right)
k}\left(  1_{\underline{\check{\alpha}}}^{\underline{\check{\alpha}}^{\prime}%
}-\dot{\sigma}_{1}^{\mathcal{\check{O}}}(l)S_{1\underline{\check{\alpha}}%
}^{\underline{\check{\alpha}}^{\prime}1}(\theta,\underline{\check{\theta}%
})\right)  \label{c3}%
\end{multline}
if
\begin{equation}
\dot{\sigma}_{1}^{\mathcal{O}}(n)S_{\underline{\hat{\alpha}}\bar{1}}^{\bar
{1}\hat{\alpha}^{\prime}}(\underline{\hat{\theta}}+W,\omega)\overset
{W\rightarrow\infty}{\rightarrow}(-1)^{\left(  N-1\right)  k}\dot{\sigma}%
_{1}^{\mathcal{\check{O}}}(l)1_{\underline{\hat{\alpha}}}^{\underline
{\hat{\alpha}}^{\prime}}\,. \label{tocheck1}%
\end{equation}

\item Taking first $W\rightarrow\infty$ and then the $\operatorname*{Res}$
means%
\begin{multline}
\operatorname*{Res}_{\omega=i\pi+\theta}\left\{  F_{\underline{\hat{\alpha}%
}\bar{1}1\underline{\check{\alpha}}}^{\mathcal{O}}(\underline{\theta}%
_{W})\overset{W\rightarrow\infty}{\rightarrow}c_{\hat{\mathcal{O}%
}\mathcal{\check{O}}}^{\mathcal{O}}(k,l,W)F_{\hat{\alpha}}^{\hat{\mathcal{O}}%
}(\hat{\theta})F_{\bar{1}1\underline{\check{\alpha}}}^{\mathcal{\check{O}}%
}(\omega,\theta,\underline{\check{\theta}})\right\} \label{c4}\\
=c_{\hat{\mathcal{O}}\mathcal{\check{O}}}^{\mathcal{O}}(k,l,W)F_{\hat{\alpha}%
}^{\hat{\mathcal{O}}}(\hat{\theta})2i\mathbf{C}_{\bar{1}1}F_{\underline
{\check{\alpha}}^{\prime}}^{\mathcal{\check{O}}}(\underline{\check{\theta}%
})\left(  1_{\underline{\check{\alpha}}}^{\underline{\check{\alpha}}^{\prime}%
}-\dot{\sigma}_{1}^{\mathcal{\check{O}}}(l)S_{1\underline{\check{\alpha}}%
}^{\underline{\check{\alpha}}^{\prime}1}(\theta,\underline{\check{\theta}%
})\right)  .
\end{multline}
1. and 2. prove that $c_{\hat{\mathcal{O}}\mathcal{\check{O}}}^{\mathcal{O}%
}\left(  k,l,W\right)  $ is independent of $k$%
\[
c_{\hat{\mathcal{O}}\mathcal{\check{O}}}^{\mathcal{O}}(k-N,l,W)=c_{\hat
{\mathcal{O}}\mathcal{\check{O}}}^{\mathcal{O}}\left(  k,l,W\right)  .
\]
3. and 4. imply that $c_{\hat{\mathcal{O}}\mathcal{\check{O}}}^{\mathcal{O}%
}\left(  k,l,W\right)  $ depends on $l$ as
\begin{equation}
c_{\hat{\mathcal{O}}\mathcal{\check{O}}}^{\mathcal{O}}(k,l-N,W)=c_{\hat
{\mathcal{O}}\mathcal{\check{O}}}^{\mathcal{O}}(k,l,W)(-1)^{\left(
N-1\right)  k}\nonumber
\end{equation}
which means that $c_{\hat{\mathcal{O}}\mathcal{\check{O}}}^{\mathcal{O}%
}\left(  k,l,W\right)  $ is independent of $l$ if $k=0\operatorname{mod}N$ and
in general
\[
c_{\hat{\mathcal{O}}\mathcal{\check{O}}}^{\mathcal{O}}(k,l,W)=c_{\hat
{\mathcal{O}}\mathcal{\check{O}}}^{\mathcal{O}}(k_{0},l_{0},W)(-1)^{\left(
N-1\right)  k\left(  l-l_{0}\right)  /N}%
\]
where $\left(  l-l_{0}\right)  =0\operatorname{mod}N$ and $c_{\hat
{\mathcal{O}}\mathcal{\check{O}}}^{\mathcal{O}}(k_{0},l_{0},W)$ is obtained by
a simple example.
\end{enumerate}

\setcounter{equation}{0}

\section{Form factor equations}

\label{sf}

The co-vector valued function $F_{1\dots n}^{\mathcal{O}}({\underline{\theta}%
})$ is meromorphic in all variables $\theta_{1},\dots,\theta_{n}$ and
satisfies the following relations \cite{KW,Sm}:

\begin{itemize}
\item[(i)] The Watson's equations describe the symmetry property under the
permutation of both, the variables $\theta_{i},\theta_{j}$ and the spaces
$i,j=i+1$ at the same time
\begin{equation}
F_{\dots ij\dots}^{\mathcal{O}}(\dots,\theta_{i},\theta_{j},\dots)=F_{\dots
ji\dots}^{\mathcal{O}}(\dots,\theta_{j},\theta_{i},\dots)\,S_{ij}(\theta_{ij})
\label{1.10}%
\end{equation}
for all possible arrangements of the $\theta$'s.

\item[(ii)] The crossing relation implies a periodicity property under the
cyclic permutation of the rapidity variables and spaces
\begin{multline}
^{~\text{out,}\bar{1}}\langle\,p_{1}\,|\,\mathcal{O}(0)\,|\,p_{2},\dots
,p_{n}\,\rangle_{2\dots n}^{\text{in,conn.}}\\
=F_{1\ldots n}^{\mathcal{O}}(\theta_{1}+i\pi,\theta_{2},\dots,\theta_{n}%
)\dot{\sigma}_{1}^{\mathcal{O}}\mathbf{C}^{\bar{1}1}=F_{2\ldots n1}%
^{\mathcal{O}}(\theta_{2},\dots,\theta_{n},\theta_{1}-i\pi)\mathbf{C}%
^{1\bar{1}}. \label{1.12}%
\end{multline}
The components of the vector $\dot{\sigma}_{1}^{\mathcal{O}}$ are given by
$\dot{\sigma}_{\alpha}^{\mathcal{O}}=\sigma_{\alpha}^{\mathcal{O}%
}(-1)^{(N-1)+(1-1/N)(n-Q^{\mathcal{O}})}$ \cite{BFK1}, where the statistics
factor $\sigma_{\alpha}^{\mathcal{O}}$ is determined by the space-like
commutation rule of the operator $\mathcal{O}$ and the field which creates the
particle $\alpha$. The charge conjugation matrix $\mathbf{C}^{\bar{1}1}$ is
given by (\ref{C}).

\item[(iii)] There are poles determined by one-particle states in each
sub-channel given by a subset of particles of the state. In particular the
function $F_{\underline{\alpha}}^{\mathcal{O}}({\underline{\theta}})$ has a
pole at $\theta_{12}=i\pi$ such that
\begin{equation}
\operatorname*{Res}_{\theta_{12}=i\pi}F_{1\dots n}^{\mathcal{O}}(\theta
_{1},\dots,\theta_{n})=2i\,\mathbf{C}_{12}\,F_{3\dots n}^{\mathcal{O}}%
(\theta_{3},\dots,\theta_{n})\left(  \mathbf{1}-\dot{\sigma}^{\mathcal{O}%
}S_{2n}\dots S_{23}\right)  \,. \label{1.14}%
\end{equation}

\item[(iv)] If there are also bound states in the model the function
$F_{\underline{\alpha}}^{\mathcal{O}}({\underline{\theta}})$ has additional
poles. If for instance the particles 1 and 2 form a bound state (12), there is
a pole at $\theta_{12}=i\eta$ such that
\begin{equation}
\operatorname*{Res}_{\theta_{12}=i\eta}F_{12\dots n}^{\mathcal{O}}(\theta
_{1},\theta_{2},\dots,\theta_{n})\,=F_{(12)\dots n}^{\mathcal{O}}%
(\theta_{(12)},\dots,\theta_{n})\,\sqrt{2}\Gamma_{12}^{(12)} \label{1.16}%
\end{equation}
where the bound state intertwiner $\Gamma_{12}^{(12)}$ and the values of
$\theta_{1},\,\theta_{2},\,\theta_{(12)}$ are given in general in
\cite{K1,KT1,BK}.

\item[(v)] Naturally, since we are dealing with relativistic quantum field
theories we finally have
\begin{equation}
F_{1\dots n}^{\mathcal{O}}(\theta_{1}+\mu,\dots,\theta_{n}+\mu)=e^{s\mu
}\,F_{1\dots n}^{\mathcal{O}}(\theta_{1},\dots,\theta_{n}) \label{1.18}%
\end{equation}
if the local operator transforms under Lorentz transformations as
$\mathcal{O}\rightarrow e^{s\mu}\mathcal{O}$ where $s$ is the
\textquotedblleft spin\textquotedblright\ of $\mathcal{O}$.
\end{itemize}

\noindent There exist bound states of $r$ fundamental particles $(\rho
_{1}\dots\rho_{r})$ (with $\rho_{1}<\dots<\rho_{r}$) which transform as the
anti-symmetric $SU(N)$ tensor representation of rank $r,~(0<r<N)$.

\setcounter{equation}{0}

\section{Asymptotic behavior for $W\rightarrow\infty$}

\label{sa0}

We use the short notations of Section \ref{s4}.

\paragraph{S-matrix:}

For $W\rightarrow\infty$ (up to higher order)%
\begin{align}
a(\theta+W)  &  \rightarrow e^{-i\pi\left(  1-\frac{1}{N}\right)  }%
e^{-i\eta\left(  1-\frac{1}{N}\right)  \frac{1}{\theta+W}}\nonumber\\
\tilde{b}(\theta+W)  &  \rightarrow1+i\eta\frac{1}{\theta+W}-\eta^{2}\frac
{1}{\left(  \theta+W\right)  ^{2}}\label{Sa}\\
\tilde{c}(\theta+W)  &  \rightarrow-i\eta\frac{1}{W}.\nonumber
\end{align}

\paragraph{Minimal form factor function $F$ and $\tilde{\phi}$-, $\tau
$-function:}

The functions defined in (\ref{F}) -- (\ref{tau}) satisfy the asymptotic
behavior for $W\rightarrow\infty$ (up to higher order)%
\begin{align}
F\left(  \theta+W\right)   &  \rightarrow X(W)^{-\left(  1-\frac{1}{N}\right)
}\left(  e^{\frac{1}{2}\left(  \theta-i\pi\right)  }\right)  ^{1-\frac{1}{N}%
}\label{Fa}\\
\tilde{\phi}(\theta\pm W)  &  \rightarrow X(W)e^{\mp\frac{1}{2}\theta}e^{\mp
i\pi\frac{1}{2}\left(  1-\frac{1}{N}\right)  }\left(  1\mp\frac{1}{NW}\left(
\theta+i\pi\left(  1-\frac{1}{N}\right)  \right)  \right) \label{phia}\\
\tau(\theta+W)  &  \rightarrow X(W)^{-2}e^{\theta}\left(  1+\frac{2\theta}%
{NW}\right) \label{taua}\\
\bar{F}\left(  \theta+W\right)   &  \rightarrow X(W)^{-\frac{1}{N}}\left(
e^{\frac{1}{2}\left(  \theta-i\pi\right)  }\right)  ^{\frac{1}{N}}=\left(
\left(  2\pi\right)  ^{-1-\frac{1}{N}}W^{\frac{1}{N}}e^{\frac{1}{2}W}%
e^{\frac{1}{2}\left(  \theta-i\pi\right)  }\right)  ^{\frac{1}{N}}%
\label{Fba}\\
X(W)  &  =\left(  2\pi\right)  ^{1+\frac{1}{N}}W^{-\frac{1}{N}}e^{-\frac{1}%
{2}W}. \label{X}%
\end{align}

\paragraph{p-functions:}

The asymptotic behavior for $W\rightarrow\infty$ of the p-functions
(\ref{pJ}), (\ref{pphi}) and (\ref{ppsi}) are given by
\begin{equation}
p^{J}(\underline{\theta}_{W},\underline{\underline{z}}_{W})\rightarrow
e^{-\frac{1}{2}W\left(  k-k_{1}-k_{N-1}\right)  }\left(  {\textstyle\sum
\nolimits_{i=1}^{k}} e^{-\hat{\theta}_{i}}\right)  p^{J}(\underline
{\hat{\theta}},\underline{\underline{\hat{z}}})p^{J}(\underline{\check{\theta
}},\underline{\underline{\check{z}}}) \label{pJa}%
\end{equation}%
\begin{equation}
p^{T}(\underline{\theta}_{W},\underline{z}_{W})\rightarrow\frac{\sum
e^{\hat{z}_{i}}}{\sum e^{\hat{\theta}_{i}}}-\frac{\sum e^{-\check{z}_{i}}%
}{\sum e^{-\check{\theta}_{i}}}=p^{T_{+}}(\underline{\hat{\theta}}%
,\underline{\hat{z}})+p^{T-}(\underline{\check{\theta}},\underline{\check{z}})
\label{pTa}%
\end{equation}%
\begin{equation}
p^{\phi}(\underline{\theta}_{W},\underline{z}_{W})\rightarrow e^{-W\left(
\left(  1-\frac{1}{N}\right)  k-k_{1}\right)  }p^{\phi}(\underline{\hat
{\theta}},\underline{\hat{z}})p^{\phi}(\underline{\check{\theta}}%
,\underline{\check{z}}) \label{pphia}%
\end{equation}%
\begin{equation}
p^{\psi}(\underline{\theta}_{W},\underline{z}_{W})\rightarrow e^{-W\frac{1}%
{2}\left(  \left(  1-\frac{1}{N}\right)  k-k_{1}\right)  }p^{\psi}%
(\underline{\hat{\theta}},\underline{\hat{z}})p^{\psi}(\underline
{\check{\theta}},\underline{\check{z}}). \label{ppsia}%
\end{equation}

\paragraph{F-function:}

Using (\ref{Fa}) we calculate for $F(\underline{\theta})$ defined in
(\ref{FK}) for $W\rightarrow\infty$ (up to a constant factor)%
\begin{equation}
F(\underline{\theta}_{W})\rightarrow F_{0}(\underline{\theta}%
,W)=X(W)^{-\left(  1-\frac{1}{N}\right)  kl}F(\underline{\hat{\theta}%
})F(\underline{\check{\theta}})\left(  {\textstyle\prod\nolimits_{i=1}^{k}}
e^{\frac{1}{2}l\left(  1-\frac{1}{N}\right)  \hat{\theta}_{i}}\right)  \left(
{\textstyle\prod\nolimits_{i=1}^{l}} e^{-\frac{1}{2}k\left(  1-\frac{1}%
{N}\right)  \check{\theta}_{i}}\right)  \label{FFa}%
\end{equation}
with $X(W)$ defined in (\ref{X}).

\paragraph{The h-funktion}

defined in (\ref{h}) satisfies for $W\rightarrow\infty$ (up to a constant
factor)%
\begin{equation}
\tilde{h}(\underline{\theta}_{W},\underline{z}_{W})\rightarrow\tilde{h}%
_{0}(\underline{\theta},\underline{z},W)\left(  1+\frac{1}{W}\tilde{h}%
_{1}(\underline{\theta},\underline{z})+O(W^{-2})\right)  \label{ha}%
\end{equation}
where
\begin{align}
\tilde{h}_{0}(\underline{\theta},\underline{z},W)  &  =X(W)^{lk_{1}%
+kl_{1}-2k_{1}l_{1}}\tilde{h}(\underline{\hat{\theta}},\underline{\hat{z}%
})\tilde{h}(\underline{\check{\theta}},\underline{\check{z}})\label{h0}\\
&  \times\left(  {\textstyle\prod\nolimits_{i=1}^{k}} e^{-\frac{1}{2}l_{1}%
\hat{\theta}_{i}}\right)  \left(  {\textstyle\prod\nolimits_{i=1}^{l}}
e^{\frac{1}{2}k_{1}\check{\theta}_{i}}\right)  \left(  {\textstyle\prod
\nolimits_{i=1}^{k_{1}}} e^{\left(  l_{1}-\frac{1}{2}l\right)  \hat{z}_{i}%
}\right)  \left(  {\textstyle\prod\nolimits_{i=1}^{l_{1}}} e^{-\left(
k_{1}-\frac{1}{2}k\right)  \check{z}_{i}}\right) \nonumber
\end{align}
and (up to a constant\footnote{$\frac{1}{N}i\pi\left(  1-\frac{1}{N}\right)
\left(  lk_{1}+kl_{1}\right)  $} which will not contribute to our results)%
\begin{equation}
\tilde{h}_{1}(\underline{\theta},\underline{z})=\frac{1}{N}\left(  k_{1}%
\sum\check{\theta}_{i}-l_{1}\sum\hat{\theta}_{i}-\left(  l-2l_{1}\right)
\sum\hat{z}_{i}+\left(  k-2k_{1}\right)  \sum\check{z}_{i}\right)  .
\label{h1}%
\end{equation}

\subparagraph{The complete h-function}

defined by (\ref{hh}) satisfies%
\begin{align}
\tilde{h}\,(\underline{\theta}_{W},\underline{\underline{z}}_{W})  &
=\prod_{j=0}^{N-2}\tilde{h}(\underline{z}_{W}^{(j)},\underline{z}_{W}%
^{(j+1)})\rightarrow\tilde{h}_{0}\,(\underline{\theta},\underline
{\underline{z}},W)\left(  1+\frac{1}{W}\tilde{h}_{1}\,(\underline{\theta
},\underline{\underline{z}})+O(W^{-2})\right) \label{hha}\\
\tilde{h}_{0}\,(\underline{\theta},\underline{\underline{z}},W)  &
=\prod_{j=0}^{N-2}\tilde{h}_{0}(\underline{z}^{(j)},\underline{z}%
^{(j+1)},W)\label{hh0}\\
\tilde{h}_{1}\,(\underline{\theta},\underline{\underline{z}})  &  =\sum
_{j=0}^{N-2}\tilde{h}_{1}\left(  \underline{z}^{(j)},\underline{z}%
^{(j+1)}\right)  \label{hh1}%
\end{align}
where $\tilde{h}_{0}$ and $\tilde{h}_{1}$ given by (\ref{h0}) and (\ref{h1})
(with $\underline{z}^{(0)}=\underline{\theta}$ and $k,k_{1}\rightarrow
k_{j},k_{j+1}$ etc.). This means that in leading order%
\begin{align*}
\tilde{h}\,(\underline{\theta}_{W},\underline{\underline{z}}_{W})  &
\rightarrow\tilde{h}\,(\underline{\hat{\theta}},\underline{\underline{\hat{z}%
}})\tilde{h}\,(\underline{\check{\theta}},\underline{\underline{\check{z}}})\\
&  \times X(W)^{\sum_{j=0}^{N-2}\left(  l_{j}k_{j+1}+k_{j}l_{j+1}%
-2k_{j+1}l_{j+1}\right)  }\left(  {\textstyle\prod\nolimits_{i=1}^{k}}
e^{-\frac{1}{2}l_{1}\hat{\theta}_{i}}\right)  \left(  {\textstyle\prod
\nolimits_{i=1}^{l}} e^{\frac{1}{2}k_{1}\check{\theta}_{i}}\right) \\
&  \times\prod_{j=1}^{N-2}\left(  \left(  {\textstyle\prod\nolimits_{i=1}%
^{k_{j}}} e^{-\frac{1}{2}\left(  l_{j-1}-2l_{j}+l_{j+1}\right)  \hat{z}%
_{i}^{(j)}}\right)  \left(  {\textstyle\prod\nolimits_{i=1}^{l_{j}}}
e^{\frac{1}{2}\left(  k_{j-1}-2k_{j}+k_{j+1}\right)  \check{z}_{j}^{(j)}%
}\right)  \right) \\
&  \times\left(  {\textstyle\prod\nolimits_{i=1}^{k_{N-1}}} e^{-\frac{1}%
{2}\left(  l_{N-2}-2l_{N-1}\right)  \hat{z}_{i}^{(N-1)}}\right)  \left(
{\textstyle\prod\nolimits_{i=1}^{l_{N-1}}} e^{\frac{1}{2}\left(
k_{N-2}-2k_{N-1}\right)  \check{z}_{i}^{(N-1)}}\right)  .
\end{align*}

\paragraph{Function $F\ast h$:}

in leading order\textbf{ }%
\begin{align}
&  F(\underline{\theta}_{W})\tilde{h}\,(\underline{\theta}_{W},\underline
{\underline{z}}_{W})\label{Fha}\\
&  \rightarrow F(\underline{\hat{\theta}})\tilde{h}\,(\underline{\hat{\theta}%
},\underline{\underline{\hat{z}}})F(\underline{\check{\theta}})\tilde
{h}\,(\underline{\check{\theta}},\underline{\underline{\check{z}}})X^{-\left(
1-\frac{1}{N}\right)  kl+\sum_{j=0}^{N-2}\left(  l_{j}k_{j+1}+k_{j}%
l_{j+1}-2k_{j+1}l_{j+1}\right)  }\nonumber\\
&  \times\left(  {\textstyle\prod\nolimits_{i=1}^{k}} e^{\frac{1}{2}\left(
l\left(  1-\frac{1}{N}\right)  -l_{1}\right)  \hat{\theta}_{i}}\right)
\left(  {\textstyle\prod\nolimits_{i=1}^{l}} e^{-\frac{1}{2}\left(  k\left(
1-\frac{1}{N}\right)  -k_{1}\right)  \check{\theta}_{i}}\right) \nonumber\\
&  \times\prod_{j=1}^{N-2}\left(  \left(  {\textstyle\prod\nolimits_{i=1}%
^{k_{j}}} e^{-\frac{1}{2}\left(  l_{j-1}-2l_{j}+l_{j+1}\right)  \hat{z}%
_{i}^{(j)}}\right)  \left(  {\textstyle\prod\nolimits_{i=1}^{l_{j}}}
e^{\frac{1}{2}\left(  k_{j-1}-2k_{j}+k_{j+1}\right)  \check{z}_{j}^{(j)}%
}\right)  \right) \nonumber\\
&  \times\left(  {\textstyle\prod\nolimits_{i=1}^{k_{N-1}}} e^{-\frac{1}%
{2}\left(  l_{N-2}-2l_{N-1}\right)  \hat{z}_{i}^{(N-1)}}\right)  \left(
{\textstyle\prod\nolimits_{i=1}^{l_{N-1}}} e^{\frac{1}{2}\left(
k_{N-2}-2k_{N-1}\right)  \check{z}_{i}^{(N-1)}}\right)  .\nonumber
\end{align}

\paragraph{Bethe state:}

By the asymptotic expansion of the S-matrix (\ref{St})%
\begin{align*}
\tilde{S}_{\alpha\beta}^{\delta\gamma}(\theta)  &  =\mathbf{1}_{\alpha\beta
}^{\delta\gamma}\tilde{b}(\theta)\mathbf{+P}_{\alpha\beta}^{\delta\gamma
}\tilde{c}(\theta)=\tilde{b}(\theta)\left(  \mathbf{1}_{\alpha\beta}%
^{\delta\gamma}\mathbf{-}\frac{i\eta}{\theta}\mathbf{M}_{\alpha\beta}%
^{\delta\gamma}+O(\theta^{-2})\right) \\
\mathbf{1}_{\alpha\beta}^{\delta\gamma}  &  =\delta_{\alpha}^{\gamma}%
\delta_{\beta}^{\delta},~\mathbf{M}_{\alpha\beta}^{\delta\gamma}%
=\mathbf{P}_{\alpha\beta}^{\delta\gamma}=\delta_{\alpha}^{\delta}\delta
_{\beta}^{\gamma}%
\end{align*}
we obtain \cite{BKZ} for the monodromy matrix (\ref{T})\footnote{Up to
$\frac{1}{w}$ terms from $\tilde{b}$, which will not contribute to our
calculations, so we will skip them in the following.}
\begin{align}
{\tilde{T}}_{\underline{\alpha}\beta}^{\beta^{\prime}\underline{\alpha
}^{\prime}}({\underline{\theta}+W},z)  &  =(\mathbf{1}-W^{-1}\,i\eta
\,\mathbf{M)}_{\underline{\alpha}\beta}^{\beta^{\prime}\underline{\alpha
}^{\prime}}+O(W^{-2})\label{4.3}\\
{\tilde{T}}_{\underline{\alpha}\beta}^{\beta^{\prime}\underline{\alpha
}^{\prime}}({\underline{\theta}},z+W)  &  =(\mathbf{1}+W^{-1}\,i\eta
\,\mathbf{M)}_{\underline{\alpha}\beta}^{\beta^{\prime}\underline{\alpha
}^{\prime}}+O(W^{-2}).\nonumber
\end{align}
The matrix elements of $\mathbf{M}_{\underline{\alpha}\beta}^{\beta^{\prime
}\underline{\alpha}^{\prime}}=\left(  \sum_{i=1}^{n}\mathbf{1}\cdots
\mathbf{P}_{i}\cdots\mathbf{1}\right)  _{\underline{\alpha}\beta}%
^{\beta^{\prime}\underline{\alpha}^{\prime}}$, as a matrix in the auxiliary
space, yields the $su(N)$ Lie algebra generators.

More general, the product ${\tilde{T}}_{\underline{\alpha}\underline{\beta}%
}^{\underline{\beta}^{\prime}\underline{\alpha}^{\prime}}({\underline{\theta}%
},\underline{z})=\left(  {\tilde{T}}({\underline{\theta}},z_{m})\mathbf{\dots
}{\tilde{T}}({\underline{\theta}},z_{1})\right)  _{\underline{\alpha
}\underline{\beta}}^{\underline{\beta}^{\prime}\underline{\alpha}^{\prime}}$
satisfies%
\begin{align*}
{\tilde{T}}_{\underline{\alpha}\underline{\beta}}^{\underline{\beta}^{\prime
}\underline{\alpha}^{\prime}}({\underline{\theta}+W},\underline{z})  &
\rightarrow(\mathbf{1}-W^{-1}\,i\eta\,\mathbf{M)}_{\underline{\alpha
}\underline{\beta}}^{\underline{\beta}^{\prime}\underline{\alpha}^{\prime}%
},~~{\tilde{T}}_{\underline{\alpha}\underline{\beta}}^{\underline{\beta
}^{\prime}\underline{\alpha}^{\prime}}({\underline{\theta}},\underline
{z}+W)\rightarrow(\mathbf{1}+W^{-1}\,i\eta\,\mathbf{M)}_{\underline{\alpha
}\underline{\beta}}^{\underline{\beta}^{\prime}\underline{\alpha}^{\prime}}\\
\mathbf{M}_{\underline{\alpha}\underline{\beta}}^{\underline{\beta}^{\prime
}\underline{\alpha}^{\prime}}  &  =\sum_{i^{\prime}=1}^{m}\left(
\mathbf{1+\dots+M}_{i^{\prime}}\mathbf{+\dots+1}\right)  _{\underline{\alpha
}\underline{\beta}}^{\underline{\beta}^{\prime}\underline{\alpha}^{\prime}}.
\end{align*}
The basic Bethe ansatz state (\ref{Phi0}) for level $0$ may be written as%
\[
\tilde{\Phi}_{\underline{\alpha}}^{\underline{\beta}}(\underline{\theta
},\underline{z})=\left(  \Omega\tilde{C}^{\beta_{m}}(\underline{\theta}%
,z_{m})\cdots\tilde{C}^{\beta_{1}}(\underline{\theta},z_{1})\right)
_{\underline{\alpha}}={\tilde{T}}_{\underline{\alpha},\underline{1}%
}^{\underline{\beta},\underline{1}}({\underline{\theta}},\underline{z})
\]
and
\[
\tilde{\Phi}_{\underline{\alpha}}^{\underline{\beta}}(\underline{\theta}%
_{W},\underline{z}_{W})={\tilde{T}}_{\underline{\alpha},\underline{1}%
}^{\underline{\beta},\underline{1}}({\underline{\theta}}_{W},\underline{z}%
_{W})\rightarrow\tilde{\Phi}_{0}\,_{\underline{\alpha}}^{\underline{\beta}%
}(\underline{\theta},\underline{z})+\frac{1}{W}\tilde{\Phi}_{1}\,_{\underline
{\alpha}}^{\underline{\beta}}(\underline{\theta},\underline{z})
\]
where%
\begin{align}
\tilde{\Phi}_{0}\,_{\underline{\alpha}}^{\underline{\beta}}(\underline{\theta
},\underline{z})  &  =\tilde{\Phi}_{\underline{\hat{\alpha}}}^{\underline
{\hat{\beta}}}({\underline{\hat{\theta}}},\underline{\hat{z}})\tilde{\Phi
}_{\underline{\check{\alpha}}}^{\underline{\check{\beta}}}({\underline
{\check{\theta}}},\underline{\check{z}})\nonumber\\
\tilde{\Phi}_{1}\,_{\underline{\alpha}}^{\underline{\beta}}(\underline{\theta
},\underline{z})  &  =i\eta\left(  \left(  {\tilde{T}}_{\underline{\hat
{\alpha}},\underline{\hat{\beta}}^{\prime}}^{\underline{\hat{\beta}%
},\underline{\hat{\alpha}}^{\prime}}({\underline{\hat{\theta}}},\underline
{\hat{z}})\right)  \left(  {\tilde{T}}_{\underline{\check{\alpha}}^{\prime
},\underline{\check{1}}}^{\underline{\check{\beta}}^{\prime},\underline
{\check{1}}}({\underline{\check{\theta}}},\underline{\check{z}})\mathbf{\check
{M}}_{\underline{\check{\alpha}},\underline{\hat{1}}}^{\underline{\hat{\beta}%
}^{\prime},\underline{\check{\alpha}}^{\prime}}\right)  -\left(
\mathbf{\hat{M}}_{\underline{\hat{\alpha}}^{\prime},\underline{\check{\beta}%
}^{\prime}}^{\underline{\check{\beta}},\underline{\hat{1}}}{\tilde{T}%
}_{\underline{\hat{\alpha}},\underline{\hat{1}}}^{\underline{\hat{\beta}%
},\underline{\hat{\alpha}}^{\prime}}({\underline{\hat{\theta}}},\underline
{\hat{z}})\right)  \left(  {\tilde{T}}_{\underline{\check{\alpha}}%
,\underline{\check{1}}}^{\underline{\check{\beta}}^{\prime},\underline
{\check{1}}}({\underline{\check{\theta}}},\underline{\check{z}})\right)
\right)  . \label{PHI1a}%
\end{align}
Similarly, for the higher levels of the Bethe ansatz. The complete Bethe
ansatz state (\ref{PHI}) satisfies%
\begin{equation}
\tilde{\Phi}_{\underline{\alpha}}(\underline{\theta}_{W},\underline
{\underline{z}}_{W})=\tilde{\Phi}_{0}\,_{\underline{\alpha}}(\underline
{\theta},\underline{\underline{z}})+\frac{1}{W}\tilde{\Phi}_{1}\,_{\underline
{\alpha}}(\underline{\theta},\underline{\underline{z}})+O(W^{-2}) \label{PHIa}%
\end{equation}
where%
\begin{align}
\tilde{\Phi}_{0}\,_{\underline{\alpha}}(\underline{\theta},\underline
{\underline{z}})  &  =\tilde{\Phi}_{\underline{\hat{\alpha}}}({\underline
{\hat{\theta}}},\underline{\underline{\hat{z}}})\tilde{\Phi}_{\underline
{\check{\alpha}}}({\underline{\check{\theta}}},\underline{\underline{\check
{z}}})\label{PHI0}\\
\tilde{\Phi}_{1}\,_{\underline{\alpha}}(\underline{\theta},\underline
{\underline{z}})  &  =\sum_{j=0}^{N-2}\tilde{\Phi}_{0}^{(N-2)}\,_{\underline
{\alpha}_{N-2}}^{\underline{\alpha}_{N-1}}\,(\underline{z}^{(N-2)}%
,\underline{z}^{(N-1)})\dots\tilde{\Phi}_{1}^{(j)}\,_{\underline{\alpha}_{j}%
}^{\underline{\alpha}_{j+1}}(\underline{z}^{(j)},\underline{z}^{(j+1)}%
)\dots\tilde{\Phi}_{0}\,_{\underline{\alpha}}^{\underline{\alpha}_{1}%
}\,(\underline{\theta},\underline{z}^{(1)})\,. \label{PHI1}%
\end{align}

\setcounter{equation}{0}

\section{Exponential behavior}

\subparagraph{Iso-scalar:}

Let $\mathcal{O}$ be an iso-scalar operator, then (\ref{w}) implies
(\ref{nphi})
\[
n_{j}=n\left(  1-j/N\right)  \Rightarrow l_{j}=n\left(  1-j/N\right)  -k_{j}%
\]
The asymptotic relation (\ref{Fha}) gives the exponential behavior%
\begin{equation}
F(\underline{\theta}_{W})\tilde{h}\,(\underline{\theta}_{W},\underline
{\underline{z}}_{W})\propto X(W)^{-\left(  1-\frac{1}{N}\right)  kl+\sum
_{j=0}^{N-2}\left(  l_{j}k_{j+1}+k_{j}l_{j+1}-2k_{j+1}l_{j+1}\right)  }
\label{Fhex}%
\end{equation}
with $X(W)$ given in (\ref{X}). By elementary calculations one shows that the
exponent of $X(W)$ can be written as%
\begin{equation}
-\left(  1-1/N\right)  kl+\sum_{j=0}^{N-2}\left(  l_{j}k_{j+1}+\left(
k_{j}-2k_{j+1}\right)  l_{j+1}\right)  =\tilde{k}_{1}^{2}+\tilde{k}_{N-1}%
^{2}+\sum_{j=1}^{N-2}\left(  \tilde{k}_{j+1}-\tilde{k}_{j}\right)  ^{2}
\label{expsc}%
\end{equation}
where $\tilde{k}_{j}=k_{j}-k\left(  1-j/N\right)  $.

\subparagraph{Adjoint representation:}

Let $\mathcal{O}$ transform as the adjoint representation, then (\ref{w})
implies (\ref{nJ})
\[
n_{j}=n\left(  1-j/N\right)  -1\Rightarrow l_{j}=n\left(  1-j/N\right)
-k_{j}-1
\]
Therefore in (\ref{expsc}) there is in addition%
\[
-k+2k_{1}-\sum_{j=1}^{N-2}\left(  -k_{j+1}+k_{j}\right)  =-k+k_{1}%
+k_{N-1}=\tilde{k}_{1}+\tilde{k}_{N-1}%
\]
and the exponent of $X(W)$ in (\ref{Fhex}) is%
\begin{multline}
-\left(  1-1/N\right)  kl+\sum_{j=0}^{N-2}\left(  l_{j}k_{j+1}+k_{j}%
l_{j+1}-2k_{j+1}l_{j+1}\right) \label{expad}\\
=\tilde{k}_{1}^{2}+\tilde{k}_{N-1}^{2}+\sum_{j=1}^{N-2}\left(  \tilde{k}%
_{j+1}-\tilde{k}_{j}\right)  ^{2}+\tilde{k}_{1}+\tilde{k}_{N-1}.
\end{multline}

\subparagraph{Iso-vector:}

Let $\mathcal{O}$ be an iso-vector operator, then (\ref{w}) implies
(\ref{npsi})%
\[
n_{j}=\left(  n-1\right)  \left(  1-j/N\right)  \Rightarrow l_{j}=n\left(
1-j/N\right)  -k_{j}-\left(  1-j/N\right)  .
\]
Therefore in (\ref{expsc}) there is in addition%
\begin{align*}
&  -\left(  k-2k_{1}\right)  \left(  1-1/N\right)  -\sum_{j=1}^{N-2}\left(
\left(  1-j/N\right)  k_{j+1}+\left(  k_{j}-2k_{j+1}\right)  \left(  1-\left(
j+1\right)  /N\right)  \right)  \\
&  =k_{1}-k\left(  1-1/N\right)  =\tilde{k}_{1}%
\end{align*}
and the exponent of $X(W)$ in (\ref{Fhex}) is%
\begin{multline}
-\left(  1-1/N\right)  kl+\sum_{j=0}^{N-2}\left(  l_{j}k_{j+1}+k_{j}%
l_{j+1}-2k_{j+1}l_{j+1}\right)  \label{expvec}\\
=\tilde{k}_{1}^{2}+\tilde{k}_{N-1}^{2}+\sum_{j=1}^{N-2}\left(  \tilde{k}%
_{j+1}-\tilde{k}_{j}\right)  ^{2}+\tilde{k}_{1}.
\end{multline}

\providecommand{\href}[2]{#2}\begingroup\raggedright\endgroup

\end{document}